\documentclass[10pt,a4paper,dvipsnames,usenames]{article}
\usepackage{amssymb}
\usepackage{amsmath}
\usepackage{amsthm}
\usepackage{latexsym}
\usepackage[dvips]{epsfig}
\usepackage{mathrsfs}
\usepackage{eufrak}
\usepackage{bm}
\usepackage{stmaryrd}
\usepackage{authblk}

\usepackage[dvipsnames]{xcolor}
\usepackage{tikz}
\theoremstyle{plain}
\newtheorem{proposition}{Proposition}
\newtheorem{lemma}{Lemma}
\newtheorem{theorem}{Theorem}
\newtheorem{assumption}{Assumption}

\newtheorem{remark}{Remark}

\def\bmg{{\bm g}}

\def\bmk{{\bm k}}
\def\bml{{\bm l}}
\def\bmn{{\bm n}}

\def\bmv{{\bm v}}

\def\bmx{{\bm x}}
\def\bmy{{\bm y}}
\def\bmz{{\bm z}}

\def\bmL{{\bm L}}

\def\bmR{{\bm R}}
\def\bmS{{\bm S}}

\def\bmbeta{{\bm \beta}}

\def\bmepsilon{{\bm \epsilon}}

\def\bmomega{{\bm \omega}}

\def\bmnu{{\bm \nu}}

\def\bmsigma{{\bm \sigma}}

\def\bmell{{\bm \ell}}

\def\bmpartial{{\bm \partial}}

\def\bmell{{\bm \ell}}

\usepackage{mathtools}% http://ctan.org/pkg/mathtools
\usepackage{xparse}% http://ctan.org/pkg/xparse
\makeatletter
\newcommand{\raisemath}[1]{\mathpalette{\raisem@th{#1}}}
\newcommand{\raisem@th}[3]{\raisebox{#1}{$#2#3$}}
\makeatother
\NewDocumentCommand{\newrbar}{O{0pt} O{0pt}}{
  \ensuremath{\mathrlap{\raisemath{#2}{\hspace*{#1}{\mathchar'26\mkern-9mu}}}r}}

\newcounter{mnotecount}%[section]

\newcommand{\mnotex}[1]%{}
{\protect{\stepcounter{mnotecount}}$^{\mbox{\footnotesize $\bullet$\themnotecount}}$ 
\marginpar{%\color{red}%
\raggedright\tiny\em
$\!\!\!\!\!\!\,\bullet$\themnotecount: #1} }

\newcounter{mnote}

\usepackage{color}
\usepackage{xcolor}
\usepackage[normalem]{ulem}
\usepackage{soul}
\usepackage{hyperref}

\begin{document}

\title{\textbf{Conformal geodesics in spherically symmetric vacuum spacetimes with Cosmological
  constant}}
 
\author[,1,3]{A. Garc\'{\i}a-Parrado G\'omez-Lobo \footnote{E-mail
    address:{\tt alfonso@math.uminho.pt}}}
\author[,2]{E. Gasper\'{\i}n \footnote{E-mail address:{\tt e.gasperingarcia@qmul.ac.uk}}}
\author[,2]{J.A. Valiente Kroon \footnote{E-mail address:{\tt j.a.valiente-kroon@qmul.ac.uk}}}
%\author[1]{Con T. Ributor}

\affil[1]{F\'isica Te\'orica, Universidad del Pa\'is Vasco, Apartado 644, 48080 Bilbao, Spain.}

\affil[2]{School of Mathematical Sciences, Queen Mary, University of London,
Mile End Road, London E1 4NS, United Kingdom.}

\affil[3]{Centro de Matem\'atica, Universidade do Minho, 4710-057 Braga, Portugal.}

\maketitle

\begin{abstract}
An analysis of conformal geodesics in the Schwarzschild-de Sitter and
Schwarzschild-anti de Sitter families of spacetimes is given. For both
families of spacetimes we show that initial data on a spacelike
hypersurface can be given such that the congruence of conformal
geodesics arising from this data cover the whole maximal extension of
canonical conformal representations of the spacetimes without forming
caustic points. For the Schwarzschild-de Sitter family, the resulting
congruence can be used to obtain global conformal Gaussian systems of
coordinates of the conformal representation. In the case of the
Schwarzschild-anti de Sitter family, the natural parameter of the
curves only covers a restricted time span so that these global
conformal Gaussian systems do not exist.

\end{abstract}

\tableofcontents

\section{Introduction}
The purpose of this article is to analyse the behaviour of conformal
geodesics in vacuum spherically symmetric spacetimes with a
Cosmological constant. Conformal geodesics are a powerful tool for the
analysis of global properties of spacetimes. In
addition to their conformal invariance, their relevance stems from the
fact that they single out \emph{privileged} representatives of the conformal class of a
solution to the Einstein field  equations. General properties of
conformal geodesics in the context of General Relativity have been
studied in \cite{FriSch87,Fri03c}. In particular, in the later
reference it has been shown that for the Schwarzschild spacetime it is
possible to construct a \emph{non-singular} congruence of conformal
geodesics covering the whole of the \emph{Kruskal-Sz\'ekeres} maximal
extension of the spacetime. This congruence can be used, in turn, to
construct a \emph{conformal Gaussian gauge system} consisting of a system of
coordinates and an adapted frame which are suitably propagated off a
fiduciary hypersurface in the spacetime. The construction and analysis
of the properties of this class of gauge systems can be regarded as a
basic first step towards the analysis, by means of conformal methods,
of  generic spacetimes with a global
structure similar to that of the Schwarzschild spacetime.
 The use of conformal
Gaussian systems in conjunction with the \emph{conformal Einstein
  field equations} renders a particularly attractive system of evolution
equations for which all the conformal
fields, save for the Weyl tensor, satisfy mere transport
equations along the conformal geodesics ---see
e.g. \cite{Fri95,Fri03a,CFEBook}.

\medskip
The main result of our analysis is that, as in the case of the
Schwarzschild spacetime, it is possible to construct congruences of
conformal geodesics covering the maximal extensions of the subextremal
and extremal
Schwarzschild-de Sitter (SdS) solutions. These congruences allow, in
turn, to construct global systems of conformal Gaussian coordinates. For the case of the  Schwarzschild-anti de Sitter (SadS)
spacetime, the situation is more subtle: although the congruence of
conformal geodesics covers the whole of the maximal extension of the
spacetime, the natural parameter of the curves of the congruence only
describes  a portion of the time span of the curves. A similar phenomenon
has been observed in the anti de Sitter spacetime ---see \cite{Fri95}.

 Applications of the constructions described in this article to the analysis of more
general (i.e. non-symmetric) classes of vacuum spacetimes in the case
of a de Sitter-Like Cosmological constant are given
elsewhere ---see \cite{GasVal17a}.

\subsection*{Notations and conventions}
In what follows $a,\,b,\,c\ldots$ will denote spacetime abstract
tensorial indices, while $i,\,j,\,k,\ldots $ are
spatial tensorial indices ranging from 1 to 3. By contrast, $\mu,\,
\nu,\,\lambda,\ldots$ and  $\alpha,\, \beta,\gamma,\ldots$ will
correspond, respectively, to spacetime and spatial coordinate
indices. 

\medskip
 The signature convention for spacetime metrics is $(+,-,-,-)$. Thus, the
induced metrics on spacelike hypersurfaces are negative definite.

\medskip 
An index-free notation will be often used. Given a 1-form ${\bmomega}$
and a vector ${\bmv}$, we denote the action of ${\bmomega}$ on
${\bmv}$ by $\langle {\bmomega},{\bmv}\rangle$. Furthermore,
${\bmomega}^\sharp$ and ${\bmv}^\flat$ denote, respectively, the
contravariant version of ${\bmomega}$ and the covariant version of
${\bmv}$ (raising and lowering of indices) with respect to a given
Lorentzian metric.  This notation can be extended to tensors of higher
rank (raising and lowering of all the tensorial indices).
The  conventions for the curvature tensors will fixed by the relation
\[
(\nabla_a \nabla_b -\nabla_b \nabla_a) v^c = R^c{}_{dab} v^d.
\]

\section{The Schwarzschild-de Sitter and Schwarzschild-anti de Sitter
  spacetimes}
\label{Section:SdS-SadS}

In the remaining of this article we will be concerned with the analysis
of \emph{spherically symmetric} spacetimes
$(\tilde{\mathcal{M}},\tilde{\bm g})$
satisfying the vacuum Einstein equations with Cosmological constant $\lambda$
\begin{equation}
\tilde{R}_{ab} -\frac{1}{2}\tilde{R}\, \tilde{g}_{ab}
-\lambda \tilde{g}_{ab} =0.
\label{VacuumEFE}
\end{equation}
In the previous expression, $\tilde{R}_{ab}$ and $\tilde{R}$ denote,
respectively, the Ricci tensor and the Ricci scalar of the 
metric $\tilde{\bmg}$. It follows that
\[
\tilde{R}_{ab} = -\lambda \tilde{g}_{ab}, \qquad \tilde{R}=-4\lambda.
\]
The minus sign in front of the Cosmological constant in equation
\eqref{VacuumEFE} has been chosen to ensure that $\lambda>0$
corresponds to \emph{de Sitter-like spacetimes} while $\lambda<0$ is
associated to \emph{anti de Sitter-like spacetimes}. The assumption of
spherical symmetry dramatically reduces the number of solutions to the
field equations \eqref{VacuumEFE}. Indeed, a generalisation of
Birkhoff's theorem shows that the only spherically symmetric solutions
to the vacuum Einstein field equations with Cosmological constant are
the Schwarzschild-de Sitter, the Schwarzschild-anti de Sitter and the
Nariai solutions ---see \cite{Sta98}.

\medskip
The line element of the Schwarzschild-de Sitter and Schwarzschild-anti
de Sitter spacetimes is given, in standard coordinates
$(t,r,\theta,\varphi)$, by
\begin{equation}
\hspace{-10mm}{\tilde{\bmg}} = \left(
1-\frac{2m}{r}-\frac{\lambda}{3}r^2\right) \mathbf{d}t\otimes
\mathbf{d}t -\left( 1-\frac{2m}{r}-\frac{\lambda}{3}r^2\right)^{-1}
\mathbf{d} r\otimes \mathbf{d} r- r^2 {\bmsigma} \label{StandardSadS}
\end{equation}
where
\[
{\bmsigma} \equiv (\mathbf{d} \theta \otimes \mathbf{d} \theta +
\sin^2 \theta \mathbf{d}\varphi \otimes \mathbf{d} \varphi)
\]
is the standard metric of $\mathbb{S}^2$ and
\[
t\in(-\infty,\infty), \qquad r\in(0,\infty), \qquad \theta\in [0,\pi],
\qquad \varphi=[0,2\pi).
\]

In the conventions used in this article, the case $\lambda>0$
corresponds to the \emph{Schwarzschild-de Sitter} spacetime and the
case $\lambda<0$ to the \emph{Schwarzschild-anti de Sitter}
spacetime. In what follows, to simplify the computations we perform a
rescaling of the line element in \eqref{StandardSadS} so as to obtain
the expression
\begin{equation}
{\tilde{\bmg}} = \left(1-\frac{M}{r}+\epsilon r^2\right)
\mathbf{d}t\otimes \mathbf{d}t -\left( 1-\frac{M}{r}+\epsilon
r^2\right)^{-1} \mathbf{d} r\otimes \mathbf{d} r- r^2
{\bmsigma},\qquad M\equiv 2m\sqrt{\frac{-\epsilon\lambda}{3}}\;.
\label{eq:rescaledStandardSadS}
\end{equation}
The constant $\epsilon$ takes the value $-1$  for Schwarzschild-de
Sitter case and $+1$ for the Schwarzschild-anti de Sitter case ---hence, $M$ is always kept strictly positive. It is convenient to
define
\begin{equation}
D(r)\equiv 1-\frac{M}{r}+\epsilon r^2= \frac{1}{r}\left(\epsilon
r^3+r-M\right).
\label{DefinitionD}
\end{equation}
Note that in the representation given by the line element
\eqref{eq:rescaledStandardSadS} both $M$ and $r$ are dimensionless
quantities ---this corresponds to a choice of units in which
$\lambda=-3\epsilon$, hence the Cosmological constant is also
dimensionless.

\medskip
\noindent In the sequel, it will be necessary to make use of alternative
coordinate systems for the metric $\tilde{\bmg}$. In particular, an
\emph{isotropic radial coordinate} $\rho$ can be introduced by means
of the relations 
\begin{equation}
 r=r(\rho)\;,\qquad  r'(\rho)=\frac{r\sqrt{D(r)}}{\rho}
\label{eq:isotropic}
\end{equation}
so that the metric line element in equation
\eqref{eq:rescaledStandardSadS} transforms into
\[
\tilde{\bmg} = D(r)\mathbf{d}t\otimes \mathbf{d}t-
\frac{r^2}{\rho^2}(\mathbf{d}\rho\otimes \mathbf{d}\rho+\rho^2\bmsigma).
\]
The explicit form of the function $r(\rho)$ is not needed in our
computations. In addition, we also make use of
\emph{Eddington-Finkelstein null coordinates} defined by the relations  
\begin{equation}
u =t - \newrbar[-1pt][-4pt], \qquad v =t + \newrbar[-1pt][-4pt]
\label{eq:ed-fink}
\end{equation}
where $\newrbar[-1pt][-3.5pt]$ is the \emph{tortoise} coordinate defined via
\begin{equation}\label{tortoise}
\newrbar[-1pt][-3.8pt](r) =\int \frac{\mbox{d}r}{D(r)}
\end{equation}
the constant of integration in the last expression is usually to be chosen
so that 
\[
\lim_{r \to\infty}\newrbar[-1pt][-3.8pt](r)=0.
\]
consequently $u, \, v \in \mathbb{R}$.  Following standard conventions we
refer to $u$ as the \emph{retarded null
coordinate} while $v$ is the \emph{advanced null
coordinate}. These coordinates render the line elements 
\[
\tilde{\bmg}=D(r)\mathbf{d}u\otimes \mathbf{d}u -
(\mathbf{d}u \otimes \mathbf{d} r + \mathbf{d} r\otimes \mathbf{d}
u) -r^2 {\bmsigma}, \qquad \tilde{\bmg}=D(r)\mathbf{d}v\otimes \mathbf{d}v +
(\mathbf{d}v \otimes \mathbf{d} r + \mathbf{d} r\otimes \mathbf{d}
v) -r^2 {\bmsigma}. 
%\label{eq:met-null}
\]

\subsection{Specific properties of the Schwarzschild-de Sitter
  spacetime}
\label{SubSection:SdSSpecific}
In this Section the Schwarzschild-de Sitter case is analysed, consequently $\lambda>0$ will be assumed.
 Observe that, for $0<M<2/(3\sqrt{3})$, the
 function $D(r)$ can be rewritten as
\begin{equation}\label{eq:DrExplicitRoots}
D(r)=-\frac{1}{r}(r-r_{b})(r-r_{c})(r-r_{-})
\end{equation}
where $r_b$ and $r_c$ are two different positive real roots, and a $r_{-}$ is
a negative real. Moreover, one has that
\begin{equation}\label{eq:RelationRoots}
0<r_b<r_c, \qquad r_c + r_b + r_-=0. 
\end{equation}
The root $r_b$ corresponds to a black hole-type  horizon, while
the root $r_c$ is associated to a Cosmological de Sitter-like 
horizon.  It can be
readily verified that $D(r)>0$ for $r_b < r< r_c$ while $D(r)<0$ for $0 < r< r_b$  and $r>r_{c}$. Consequently, 
the metric is static in the region $r_b<r<r_c$ between the
two horizons and there are no other static regions outside this range for $r$.  
Using Cardano's
formula for cubic equations it is possible to find the explicit values
of $r_b$, $r_c$ and $r_-$.  One finds that 
\begin{subequations}
\begin{eqnarray}
&&r_-=-\frac{2}{\sqrt{3}}\cos\left(\frac{\phi}{3}\right), 
\label{LocationHorizon0}\\
&&r_b=\frac{1}{\sqrt{3}}\left(\cos\left(\frac{\phi}{3}\right)-\sqrt{3}\sin\left(\frac{\phi}{3}\right)\right), \label{LocationHorizon1}\\
&&r_c=\frac{1}{\sqrt{3}}\left(\cos\left(\frac{\phi}{3}\right)+\sqrt{3}\sin\left(\frac{\phi}{3}\right)\right), \label{LocationHorizon2}
\end{eqnarray}
\end{subequations}
where the parameter $\phi$ is defined through the relation 
\begin{equation}
M=\frac{2\cos\phi}{3\sqrt{3}},\qquad \phi\in(0,\tfrac{1}{2}\pi).
\label{eq:define-phi}
 \end{equation}
We will refer to this case, for which $\phi \in (0,\frac{1}{2}\pi)$,  as the
\emph{subextremal Schwarzschild de-Sitter spacetime}. In contrast,  the case
$\phi=\pi/2$ for which   $M=2/(3\sqrt{3})$ will be referred as the 
\emph{extremal Schwarzschild de-Sitter spacetime}. This is of special
 interest as one has 
\[
r_{\mathcal{H}}\equiv r_b=r_c=\frac{3M}{2}=\frac{1}{\sqrt{3}}
\]
so that the function $D(r)$ takes the form
\begin{equation}
D(r)=-\frac{1}{r}\left(\frac{2}{\sqrt{3}}+r\right)
\left(r-\frac{1}{\sqrt{3}}\right)^2.
\label{eq:define-D-extremal} 
 \end{equation}

\medskip
Finally, it is observed that one can also consider the
\emph{hyperextremal Schwarzschild de-Sitter spacetime}
 characterised by the condition
 $M>2/(3\sqrt{3})$. In this case, the spacetime does 
not contain horizons.

\subsection{Specific properties of the Schwarzschild anti de Sitter spacetime} 

For the Schwarzschild-anti de Sitter spacetime $\lambda<0$,
consequently, in this case, the function $D(r)$ has a single real root
$r_b$. Using Cardano's formula one finds that
\begin{equation}
 r_b=\frac{2}{\sqrt{3}}\sinh\frac{\phi}{3}, \qquad
 M=\frac{2\sinh\phi}{3\sqrt{3}},\qquad \phi\in(0,\infty).
\label{HorizonAntiDeSitter}
\end{equation}
The root $r_b$ is associated to an event horizon separating two
regions: a black hole region with a singularity at $r=0$ and an
exterior region defined by the condition $r>r_b$ where the spacetime
is static ---see Figure \ref{Fig:SadS}.
One can eliminate  the parameter $\phi$ in equation \eqref{HorizonAntiDeSitter}
 and write $M$ in terms of 
$r_b$ as
\begin{equation}
 M=r_b+r_b^3.
\end{equation}

\subsection{Penrose diagrams}

The Penrose diagrams of the Schwarzschild-de Sitter and the
Schwarzschild-anti de Sitter spacetimes 
are well known ---see e.g., \cite{GriPod09} for details on their
construction. These diagrams are given in Figures
\ref{Fig:SubSdS}-\ref{Fig:SadS}.  A discussion of the general procedure for the
construction of Penrose diagrams in spherically symmetric spacetimes
can be found in \cite{CFEBook}, Chapter 6.  

\begin{figure}[t]
\centering
\includegraphics[width=1\textwidth]{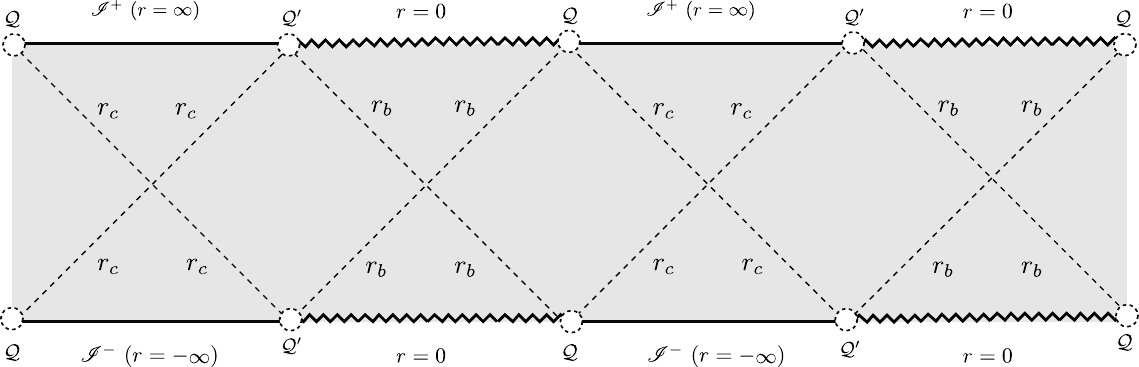}
\caption{Penrose diagram for the subextremal Schwarzschild-de Sitter
  spacetime. The serrated line denotes the location of the
  singularity; the continuous black line denotes the conformal
  boundary; the dashed line shows the location of the black hole and
  cosmological horizons which are located
at $r=r_{b}$ and $r=r_{c}$ respectively.  The excluded points $\mathcal{Q}$
  and $\mathcal{Q'}$ where the singularity seems to meet the
  conformal boundary correspond to asymptotic regions of the
  spacetime that does not belong to the singularity nor the
  conformal boundary.}
\label{Fig:SubSdS}
\end{figure}

\begin{figure}[t]
\includegraphics[width=1\textwidth]{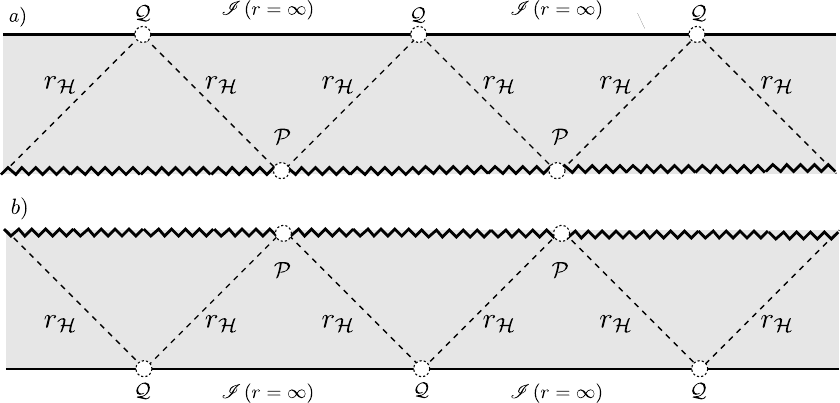}
\caption{Penrose diagrams for the extremal Schwarzschild-de Sitter
  spacetime. Figure (a) corresponds to a \emph{white hole} which evolves
  towards a de Sitter final state while Figure (b) is a model of a black
  hole with a future singularity. The Killing horizon 
  is located at $r=r_{\mathcal{H}}$
 as described in the main text. Similar to the subextremal case,
  the excluded points denoted by $\mathcal{P}$, $\mathcal{Q}$
  represent asymptotic regions of the spacetime that do not belong to
  the singularity nor the conformal boundary.}
\label{Fig:eSdS}
\end{figure}

\begin{figure}[t]
\centering
\includegraphics[scale=1]{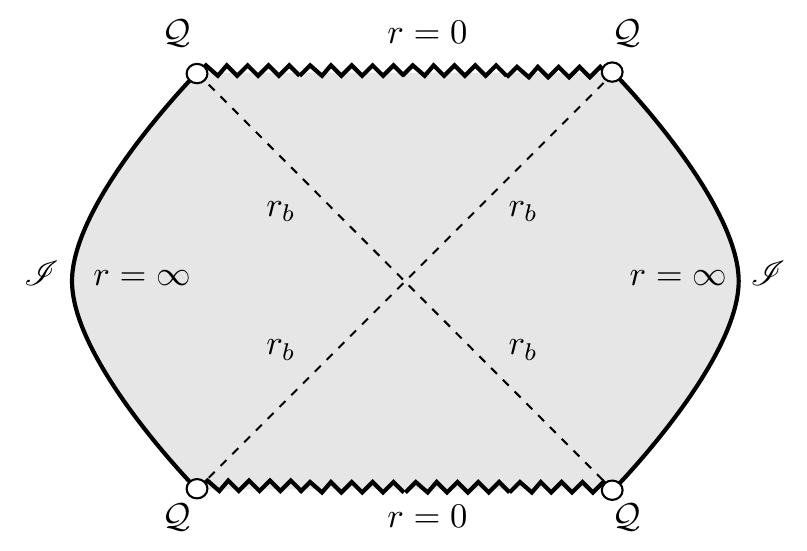}
% \put(-210,70){$\mathscr{I}$}
% \put(0,70){$\mathscr{I}$}
% \put(-105,0){$r=0$}
% \put(-105,140){$r=0$}
% \put(-170,0){$\mathcal{Q}$}
% \put(-40,0){$\mathcal{Q}$}
% \put(-170,140){$\mathcal{Q}$}
% \put(-40,140){$\mathcal{Q}$}
% \put(-70,50){$r_b$}
% \put(-140,50){$r_b$}
% \put(-70,90){$r_b$}
% \put(-140,90){$r_b$}
% \put(-190,70){$r=\infty$}
% \put(-35,70){$r=\infty$}
\caption{Penrose diagram of the Schwarzschild-anti de Sitter
  spacetime.  Similar notation as that used for the Penrose diagrams
  of the Schwarzschild-de Sitter spacetime has been used.  The black
  hole horizon is located at at $r=r_{b}$ and the asymptotic points are denoted 
  by $\mathcal{Q}$.
  Notice that the singularity is of spacelike nature and, in contrast
  with the de-Sitter like case, the conformal boundary is
  timelike. Following the discussion in\cite{KloStr96,FidHubKleShe04},
generically, if the singularity is represented as a straight
horizontal line, the conformal boundary must be represented as
curved. A more accurate representation of the singularity can be found in figures 
\ref{fig:cgSchAdS} and \ref{fig:cgSchAdS2}.}
\label{Fig:SadS}
\end{figure}

\section{Conformal geodesics in vacuum spacetimes with Cosmological constant}
\label{Section:ConformalCurves}

In this Section  some basic results concerning conformal
geodesics in vacuum spacetimes are briefly reviewed. Full details can be found in 
\cite{Fri95,Fri03c,LueVal13b} ---see also \cite{CFEBook}. In what follows,
 let $(\tilde{\mathcal{M}},\tilde{\bm g})$ denote a
spacetime satisfying the vacuum Einstein field equations with
Cosmological constant \eqref{VacuumEFE}. 

\subsection{Basic definitions}
\label{Section:ConformalCurvesBasics}

Given an interval $I\subseteq\mathbb{R}$, let ${x}(\tau)$, 
$\tau \in I$ denote a curve in $(\mathcal{M},\tilde{\bm g})$ and let 
${\bmbeta}(\tau)$ denote a 1-form along ${\bm x}(\tau)$. Furthermore, let
$\dot{x}\equiv \mbox{d}{\bm x}/\mbox{d}\tau$ denote the tangent
vector field of the curve ${x}(\tau)$. The \emph{conformal
  geodesic equations} are then given by:
\begin{subequations}
\begin{eqnarray}
&& \tilde{\nabla}_{\dot{\bm x}} \dot{\bmx} = -2 \langle {\bmbeta},
\dot{\bmx} \rangle
\dot{\bmx} + \tilde{\bmg}(\dot{\bmx},\dot{\bmx}) {\bmbeta}^\sharp, \label{ConformalCurve1} \\
&& \tilde{\nabla}_{\dot{\bm x}} {\bmbeta} = \langle {\bmbeta}, \dot{\bm x}
\rangle {\bmbeta} - \tfrac{1}{2} \tilde{\bmg}^\sharp ({\bmbeta},{\bmbeta})
\dot{\bmx}^\flat + \tilde{\bm L}(\dot{\bmx}, \cdot), \label{ConformalCurve2}
\end{eqnarray}
\end{subequations}
where $\tilde{\bm\nabla}$ denotes  
the Levi-Civita connection of the physical metric $\tilde{\bmg}$ and
 $\tilde{\nabla}_{\dot{\bmx}}$ denotes a derivative in the direction of $\dot{\bmx}$. Notice that in the last expression the indices of the vectors and covectors are raised or lowered  using $\tilde{\bmg}$ 
---unless otherwise stated, we follow this convention in the 
rest of this article. The symbol 
$\tilde{\bmL}$ denotes  the \emph{Schouten tensor} of $\tilde{\bmg}$ defined by:
\[
 \tilde{L}_{ab} \equiv\frac{1}{2}\bigg(
\tilde{R}_{ab} - \frac{1}{6} \tilde{R}\, \tilde{g}_{ab}\bigg).
\] 
In the case of a spacetime satisfying the vacuum equations with
Cosmological constant \eqref{VacuumEFE} one has that:
\begin{equation}
\tilde{L}_{ab} = -\frac{1}{6}\lambda\tilde{g}_{ab}.
\label{SchoutenTensor:Vacuum}
\end{equation}
In the remainder of this article \emph{it is assumed that the spacetime
$(\tilde{\mathcal{M}},\tilde{\bmg})$ satisfies condition \eqref{SchoutenTensor:Vacuum}}. 

\subsection{Conformal factors along conformal geodesics}
Conformal geodesics allow to single out a \emph{canonical
  representative} of the conformal class  
of $\tilde{\bmg}$.
 Let
$(\mathcal{M},{\bmg})$ denote a conformal extension of
$(\tilde{\mathcal{M}},\tilde{\bm g})$. Hence, there exists a scalar
$\Theta$ such that the metrics $\tilde{\bmg}$ and
${\bmg}$ are related via
\begin{equation}
{\bmg} = \Theta^2 \tilde{\bmg}.
\label{ConformalTransformation:metric}
\end{equation}
The conformal factor $\Theta$ can be fixed by requiring ${\bmx}(\tau)$ to be
timelike and imposing the normalisation condition
\begin{equation}
{\bmg}(\dot{\bmx},\dot{\bmx}) =1.
\label{Normalisation:Unphysical}
\end{equation}

Using equation \eqref{SchoutenTensor:Vacuum}, repeated differentiation of
condition \eqref{Normalisation:Unphysical} together with the conformal
geodesic equations and the Einstein field equations expressed as in equation  \eqref{SchoutenTensor:Vacuum}  one obtains the relations
\[
\dot{\Theta} = \langle {\bmbeta}, {\bmx}\rangle \Theta, \qquad
\ddot{\Theta} = \Theta^{-1}\left(\frac{1}{2}
\tilde{\bmg}^\sharp({\bmbeta},{\bmbeta}) -\frac{1}{6}\lambda\right),
\qquad \dddot{\Theta}=0,
\]
where $\dot{\Theta} \equiv \tilde{\nabla}_{\dot{x}} \Theta$,
etc. Integrating the last of these equations \emph{along a given
  conformal geodesic} one finds
\begin{equation}
\Theta = \Theta_\star + \dot{\Theta}_\star (\tau-\tau_\star) +
\frac{1}{2}\ddot{\Theta}_\star(\tau-\tau_\star)^2, \label{ConformalFactor}
\end{equation}
where $\Theta_\star$, $\dot{\Theta}_\star$ and $\ddot{\Theta}_\star$
are prescribed at a fiduciary value $\tau_\star$ of the parameter
along the conformal geodesic $x(\tau)$. The coefficients $\dot{\Theta}_\star$ and
$\ddot{\Theta}_\star$ satisfy the constraints
\begin{equation}
\dot{\Theta}_\star = \langle {\bmbeta}_\star, \dot{\bmx}_\star \rangle
\Theta_\star, \qquad \Theta_\star \ddot{\Theta}_\star =
\frac{1}{2}\tilde{\bmg}({\bmbeta}_\star,{\bmbeta}_\star)
-\frac{1}{6}\lambda,
\label{Constraints}
\end{equation}
where ${\bmbeta}_\star$ and $\dot{\bmx}_\star$ denote, respectively,
the value of ${\bmbeta}$ and $\dot{\bmx}$ at $\tau=\tau_\star$. A
further computation exploiting the above expressions shows that
\[
\tilde{\nabla}_a\Theta \tilde{\nabla}^a\Theta = \tfrac{1}{3}\lambda \quad \text{at} \quad \mathscr{I}.
\]
Thus, as it is well known for vacuum spacetimes, the causal character
of the conformal boundary is determined by the sign of $\lambda$
---spacelike if $\lambda>0$ and timelike if $\lambda<0$.

\begin{remark}
{\em Notice that by virtue of the normalisation condition
\eqref{Normalisation:Unphysical}, $\tau$ is the ${\bm g}$-proper time
of the conformal geodesic.}
\end{remark}

\subsection{The $\tilde{\bmg}$-adapted equations}
\label{Section:PhysicalMetricAdaptedEquations}

In specific computations, it is more convenient to consider a
parametrisation of the conformal geodesics in terms of the \emph{physical}
proper time. To this end, consider the parameter transformation
$\tilde{\tau}=\tilde{\tau}(\tau)$ given by
\begin{equation}
\tilde{\tau} = \tilde{\tau}_\star + \int_{\tau_\star}^{\tau}
\frac{\mbox{d} s}{\Theta(s)},
\label{Reparametrisation}
\end{equation}
with inverse $\tau=\tau(\tilde{\tau})$. In what follows, let
$\tilde{\bmx}(\tilde{\tau})\equiv {\bmx}(\tau(\tilde{\tau}))$. Using this notation, it can
then be verified that
\begin{equation}
\tilde{\bm x}' \equiv \frac{\mbox{d} \tilde{x}}{\mbox{d}\tilde{\tau}}=
\Theta \dot{\bmx},
\label{eq:geodesic-rescaling}
\end{equation}
and that
\[
\tilde{\bmg}(\tilde{\bmx}',\tilde{\bmx}')=1,
\]
so that $\tilde{\tau}$ is the $\tilde{\bmg}$-proper time of the curve
$\tilde{\bmx}$.  In order to write the equation for the curve
$\tilde{x} (\bar{\tau})$ in a convenient way, we consider the split
\begin{equation}
{\bmbeta}=\tilde{\bmbeta} + \varpi \dot{\bmx}^\flat,
\label{eq:define-bhat}
\end{equation}
where the 1-form $\tilde{\bmbeta}$ satisfies
\[
\langle \tilde{\bmbeta}, \dot{\bmx}\rangle =0, \qquad \varpi \equiv
\frac{\langle {\bmbeta},
  \dot{\bmx}\rangle}{\tilde{\bmg}(\dot{\bmx},\dot{\bmx})}, \qquad
     {\bmg}^\sharp({\bmbeta},{\bmbeta}) =\langle
     {\bmbeta},\dot{\bmx}\rangle^2 +
     {\bmg}^\sharp(\tilde{\bmbeta},\tilde{\bmbeta}).
\] 
Moreover one can verify that
\begin{equation}\label{AdaptedParameters-ConformalFactor}
 \tilde{\bmg}=\Theta^{-2}{\bmg}, \qquad \langle \bmbeta, \dot{x} \rangle=
 \Theta^{-1}\dot{\Theta}, \qquad \varpi = \Theta^{-1}\dot{\Theta}.
\end{equation}
In terms of these objects the $\tilde{\bmg}$-\emph{adapted equations
  for the conformal curves} are given by 
\begin{subequations}
\begin{eqnarray}
&& \tilde{\nabla}_{\tilde{\bmx}'} \tilde{\bmx}' =
  \tilde{\bmbeta}^\sharp,
 \label{PhysicalMetricAdaptedCC1}\\
&& \tilde{\nabla}_{\tilde{\bmx}'} \tilde{\bmbeta} = \beta^2
 \tilde{\bmx}'{}^\flat,
\label{PhysicalMetricAdaptedCC2}
\end{eqnarray}
\end{subequations}
where
\begin{equation}
\beta^2 \equiv -\tilde{\bm g}^\sharp(\tilde{\bmbeta},\tilde{\bmbeta})
\label{eq:define-beta}
\end{equation}
is constant along a given conformal geodesic.

\subsection{The deviation equations}
\label{Section:ConformalDeviationEquations}

When working with  congruences of conformal geodesics it is important
to analyse whether they develop conjugate points. To this
end, let ${x}(\tau,\sigma)$ and ${\bmbeta}(\tau,\sigma)$ denote a
family of conformal curves depending smoothly on a parameter
$\sigma$. Following \cite{Fri03c}, let
\[
{\bmz}\equiv\partial_\sigma {x}, \qquad {\bmomega} \equiv
\tilde{\nabla}_{\bmz} {\bmbeta}.
\]
The fields ${\bmz}$ and ${\bmomega}$ denote, respectively, the
\emph{deviation vector field} and the \emph{deviation 1-form}. The
\emph{conformal Jacobi equation} and the \emph{1-form deviation
  equation} are given by
\begin{subequations}
\begin{eqnarray}
&& \tilde{\nabla}_{\dot{\bmx}} \tilde{\nabla}_{\dot{\bmx}} {\bmz} = \tilde{\bmR}(\dot{\bmx},{\bmz}) 
\dot{\bmx} - {\bmS}({\bmomega};\dot{\bmx},\dot{\bmx}) - 2
{\bmS}({\bmbeta}; \dot{\bmx}, \tilde{\nabla}_{\dot{\bmx}} {\bmz}),\label{eq:dev-1} \\
&& \tilde{\nabla}_{\dot{\bmx}} {\bmomega} = -{\bmbeta}\cdot\tilde{\bmR}(\dot{\bmx},{\bmz})
 + \frac{1}{2}\left( {\bmomega}\cdot
{\bmS}({\bmbeta};\dot{\bmx},\,\cdot\,) + 
{\bmbeta}\cdot {\bmS}({\bmomega};\dot{\bmx},\,\cdot\,) + {\bmbeta}\cdot {\bmS}
 ({\bmbeta};\tilde{\nabla}_{\dot{\bmx}}{\bmz},\,\cdot\,)\right),\label{eq:dev-2}
\end{eqnarray}
\end{subequations}
where $\tilde{\bmR}(\,\cdot\,,\,\cdot\,)$ denotes the Riemann tensor of the
metric $\tilde{\bmg}$, and 
\begin{eqnarray*}
&& {\bmS}({\bmbeta}; \dot{\bmx},{\bmy}) \equiv \langle {\bmbeta}, \dot{\bmx}
\rangle {\bmy} + \langle {\bmbeta},{\bmy}\rangle \dot{\bmx} -\tilde{\bmg}
(\dot{\bmx},{\bmy}) {\bmbeta}^\sharp, \\
&& {\bmomega}\cdot {\bmS}({\bmbeta};\dot{\bmx},\,\cdot\,) \equiv \langle
{\bmomega},\dot{\bmx}\rangle {\bmbeta} +\langle {\bmbeta},\dot{\bmx} \rangle
 {\bmomega} - \tilde{\bmg}^\sharp({\bmomega},{\bmbeta}) \dot{\bmx}^\flat.
\end{eqnarray*}

\medskip

A $\tilde{\bmg}$-adapted version of the conformal geodesics can be
readily computed. Accordingly, let $\tilde{x}\equiv
x(\tilde{\tau},\sigma)$ be a reparametrisation of ${x}(\tau,\sigma)$
in terms of the physical proper time $\tilde{\tau}$. Moreover, let 
\[
\tilde{\bmz} \equiv \partial_\sigma \tilde{x},
\qquad \tilde{\bmomega} \equiv \tilde{\nabla}_{\tilde{\bmz}} \tilde{\bmbeta}.
\]
In terms of these new variables one has that equations \eqref{eq:dev-1}-\eqref{eq:dev-2}
take the form
\begin{subequations}
\begin{eqnarray}
&& \tilde{\nabla}_{\tilde{\bmx}'} \tilde{\nabla}_{\tilde{\bmx}'} {\bmz} =
\tilde{\bmR}(\tilde{\bmx}',{\bmz})\tilde{\bmx}' +
\tilde{\bmomega}^\sharp, \label{PhysicalMetricAdaptedDevEqns1} \\
&& \tilde{\nabla}_{\tilde{\bmx}'} \tilde{\bmomega} = -\tilde{\bmbeta}\cdot 
\tilde{\bmR}(\tilde{\bmx}',\tilde{\bmz})+
\tilde{\bmx}^{\prime\flat}\tilde{\nabla}_{\tilde{\bmz}}\beta^2  +
\beta^2\tilde{\nabla}_{\tilde{\bmx}'} \tilde{\bmz}^\flat. \label{PhysicalMetricAdaptedDevEqns2}
\end{eqnarray}
\end{subequations}
A computation shows that
\[
\tilde{\nabla}_{\tilde{x}'} \tilde{\nabla}_{\tilde{z}} \beta^2=0.
\]
Therefore, $\tilde{\nabla}_{\tilde{z}} \beta^2$ is constant along a given
conformal geodesic. 

\subsection{Formulae in warped product spaces}
\label{sec:FormulaeWarpedProductSpaces}

The line element \eqref{StandardSadS} is in the form of a warped
product. This structure can be exploited to simplify the analysis of
the $\tilde{\bm g}$-adapted conformal geodesic equations
\eqref{PhysicalMetricAdaptedCC1}-\eqref{PhysicalMetricAdaptedCC2}.  To
have a self-contained discussion, the main formulae for warped product
spacetimes, originally derived in \cite{Fri03c}, are given in this
Section.
 
\medskip
In what follows, the discussion will be particularised to
 spacetimes whose metric can be
written in the  warped product form:
\begin{equation}
\tilde{\bmg}  = \tilde{l}_{AB}
\mathbf{d}x^A \otimes \mathbf{d} x^B + f^2 \tilde{k}_{ij} \mathbf{d}x^i\otimes 
 \mathbf{d}x^j,
\label{WarpedProductMetric}
\end{equation}
with 
\[
\tilde{l}_{AB} = l_{AB} (x^C), \quad \tilde{k}_{ij} = k_{ij}(x^k), \quad f=f(x^A)>0,
\]
and $A,\, B, \,C=0,\,1$ and $i,\, j,\, k=2,\,3$. In addition, it is
assumed that the 2-dimensional metric given by 
$\tilde{\bml}\equiv \tilde{l}_{AB} \mathbf{d}x^A \otimes \mathbf{d} x^B$ is
Lorentzian, while  $\tilde{\bmk} =\tilde{k}_{ij}
\mathbf{d}x^i\otimes \mathbf{d}x^j$ is a negative-definite 2-dimensional
 Riemannian metric. \emph{In view of this structure it is natural to consider
solutions to the conformal curve equations satisfying $\dot{x}^i=0$
and ${b}_j=0$.} In the case of the metric given by the line element
\eqref{StandardSadS} this Ansatz leads to solutions to the conformal
curves which have no evolution on the angular coordinates, and one has
only to consider evolution equations for the coordinates $(t,r)$
---or alternatively $(u,r)$ or $(v,r)$. A direct
computation shows that for this type of conformal geodesics the
$\tilde{\bm g}$-adapted equations for the conformal geodesics imply
\begin{subequations}
\begin{eqnarray}
&& \not{\!\!D}_{\tilde{\bmx}'} \tilde{\bmx}' = \tilde{\bmbeta}^\sharp,
 \label{WarpedProductConformalCurveEquations1}\\
&& \tilde{\bmbeta} = \pm \beta \tilde{\bmepsilon}_{\tilde{\bml}}(\tilde{\bmx}',\,\cdot\,),
 \label{WarpedProductConformalCurveEquations2}
\end{eqnarray}
\end{subequations}
with 
\[
{\bmepsilon}_{\tilde{\bml}} \equiv \sqrt{|\Delta|} \mathbf{d}x^0 \wedge
\mathbf{d}x^1, \quad
\Delta\equiv \det \tilde{l}_{AB}
\]
denoting the volume form of $\tilde{\bml}$
and where $\not{\!\!\!D}$ denotes the Levi-Civita covariant derivative of $\tilde{\bm
l}$. A further computation shows that under the present Ansatz, the
$\tilde{\bmg}$-adapted deviation equations
\eqref{PhysicalMetricAdaptedDevEqns1}-\eqref{PhysicalMetricAdaptedDevEqns2}
are equivalent to each other and to the equation
\begin{equation}
\not{\!\!D}_{\tilde{\bmx}'} \not{\!\!D}_{\tilde{\bmx}'} \tilde{\bmz} = \frac{1}{2} R[\tilde{\bml}]\,
 {\bmepsilon}_{\tilde{\bml}}(\tilde{\bmx}',\tilde{\bmz}) {\bmepsilon}_{\tilde{\bml}}(\tilde{\bmx}',
  \,\cdot\,) ^\sharp \pm \left( \not{\!\!D}_{\tilde{\bmz}} \beta\, {\bmepsilon}_{\tilde{\bml}}
  (\tilde{\bmx}',\, \cdot\,)^\sharp + \beta {\bmepsilon}_{\tilde{\bml}}
  (\not{\!\!D}_{\tilde{\bmx}'} \tilde{\bm
  z},\,\cdot \,)\right),
\label{WarpProductDeviationEquation}
\end{equation}
where $R[\tilde{\bml}]$ denotes the Ricci scalar of $\tilde{\bml}$. 

\medskip
For conformal curves satisfying $\dot{x}^a=0$ and ${b}_c=0$, the
question of whether the deviation vector field $\tilde{\bmz}$ is
non-vanishing can be rephrased in terms of a similar question for the
scalar 
\begin{equation}
\tilde{\omega} \equiv {\bmepsilon}_{\tilde{\bml}}(\tilde{\bmx}',\tilde{\bmz}).
\label{eq:define-omega}
\end{equation}
Notice that as long as $\tilde{\omega}\neq 0$, $\tilde{\bmx}'$ and $\tilde{\bmz}$
are linearly independent. A computation using the deviation equation
\eqref{WarpProductDeviationEquation} yields
\begin{equation}
\not{\!\!D}_{\tilde{\bmx}'} \not{\!\!D}_{\tilde{\bmx}'}\tilde{\omega}
  = \left( \beta^2 + \frac{1}{2}R[\tilde{\bml}]\right)\tilde{\omega} +
 \not{\!\!D}_{\tilde{\bmz}} \beta.
\label{ReducedWarpProductDeviationEquation}
\end{equation}
If we regard the congruence of conformal geodesics as a congruence in
a conformal extension $(\mathcal{M},\bmg)$ with $\bmg=\Theta^2
\tilde{\bmg}$ then the scalar which we need to analyse is
\begin{equation}\label{DefinitionOmegaCaustics}
{\omega}\equiv {\bmepsilon}_{\bm l}(\dot{\bmx},{\bmz}) 
\end{equation}
where ${\bml}\equiv\Theta^2 \tilde{\bml}$. The scalar ${\omega}$ has
a similar geometric meaning as $\tilde{\omega}$ for the congruence
of conformal geodesics regarded as curves in the unphysical
spacetime. From equation \eqref{eq:geodesic-rescaling} we deduce that
$\tilde{\bmx}'=\Theta \dot{\bmx}$ and in addition ${\bm
\epsilon}_{{\bml}}=\Theta^2{\bmepsilon}_{\tilde{\bml}}$. Hence
\begin{equation}
{\omega}=\Theta \tilde{\omega}.
\label{eq:omega-rescaled}
\end{equation}

\subsection{Explicit expressions for the reduced conformal geodesic equations}
\label{Section:ExplicitExpressions}

 The corresponding 2-dimensional metric ${\bm l}$, as determined by  the warped product form \eqref{WarpedProductMetric}, for the line element of equation (\ref{eq:rescaledStandardSadS}), is given by 
\begin{equation}
\tilde{\bml} = D(r) \mathbf{d}t \otimes \mathbf{d}t -D^{-1}(r)
\mathbf{d}r \otimes \mathbf{d}r.
\label{eq:2-l}
\end{equation}
It follows that the reduced conformal geodesic equations
\eqref{WarpedProductConformalCurveEquations1}-\eqref{WarpedProductConformalCurveEquations2}
reduce to
\begin{subequations}
\begin{eqnarray}
&& t^{\prime \prime} +
  \frac{\partial_{r}D(r)}{D(r)}r^{\prime}t^{\prime}
  = \frac{1}{D(r)}\beta\ r^{\prime}, \label{StandardCGEqn1}
  \\ && r^{\prime \prime} -
  \frac{1}{2}\frac{\partial_{r}D(r)}{D(r)}r^{\prime 2} + \frac{1}{2} D(r)
    \partial_{r}D(r)t^{\prime 2} =
  D(r)\beta\ t^{\prime}, \label{StandardCGEqn2}
\end{eqnarray}
\end{subequations}
where consistent with the notation of Section
\ref{Section:PhysicalMetricAdaptedEquations} we have set
$r \equiv r(\tilde{\tau})$ and $t\equiv t(\tilde{\tau})$.
Initial data for these equations can be prescribed following the
discussion of Sections
\ref{Section:ConformalCurvesBasics}-\ref{Section:PhysicalMetricAdaptedEquations}. Observe
that equations \eqref{StandardCGEqn1}-\eqref{StandardCGEqn2} can be
decoupled by making use of the $\tilde{\bm g}$-normalisation condition
\begin{equation}
D(r)\, t^{\prime 2} - \frac{1}{D(r)}
r^{\prime 2} =1.
\label{PhysicalNormalisation}
\end{equation}
Solving the latter for $ t'\geq0$ and substituting into equation 
\eqref{StandardCGEqn2}, one obtains that
\begin{equation}
r'' + \tfrac{1}{2} \partial_{r} D(r) - \beta
\sqrt{D(r) + r^{\prime 2}}=0.
\label{rprimeprime}
\end{equation}
This equation can be integrated once to yield
\[
\sqrt{D(r) + r^{\prime 2}} -\beta r = \gamma,
\]
where $\gamma$ is a constant given in terms of the initial value of
the radial coordinate $r_\star$ by
\begin{equation}\label{RelationGammaWithInitialData}
\gamma \equiv -\beta r_\star+\sqrt{D_\star+ r'_\star}\;,\qquad
D_\star\equiv D(r_\star)>0\;,\qquad r'_\star\equiv
r'(\tilde\tau_\star).
\end{equation}
It follows that
\begin{equation}
\label{ReducedEquation}
r' = \pm \sqrt{(\gamma+\beta r)^2 - D(r)},
\end{equation}
with the sign depending on the value of $r_\star$.  
Substituting the latter expression into equation
 \eqref{PhysicalNormalisation} we get
\begin{equation}\label{Equation-t'}
t' = \frac{|\gamma + \beta r|}{ |D(r)|}.
\end{equation}

\noindent One can formally integrate equations \eqref{ReducedEquation}
and \eqref{Equation-t'} as
\begin{subequations}
\begin{eqnarray}
&&  \tilde{\tau}(r)-\tilde{\tau}_{\star}=
 \int_{r_{\star}}^{r} \frac{1}{\sqrt{(\gamma + \beta \bar{r})^2
     -D(\bar{r})}}\mbox{d}\bar{r}, \label{GeneralFormulaForTauTilde}\\
&& t(r)-t_{\star}=\int_{r_{\star}}^{r}\frac{|\gamma
  + \beta \bar{r}|}{\sqrt{(\gamma^2+\beta
   \bar{r})^2-D(\bar{r})}}\mbox{d}\bar{r}. \label{GeneralFormulaFort}
\end{eqnarray}
\end{subequations}

\begin{remark}
{\em As a notational remark we stress that $\bar{r}$ is
only used as a variable for integration. Notice, in
contrast, that $\tau$ and $\tilde{\tau}$ represent the $\bmg$ and
$\tilde{\bmg}$ proper time.}
\end{remark}

The solutions of equations \eqref{ReducedEquation} and
 \eqref{Equation-t'}, or equivalently expressions
 \eqref{GeneralFormulaForTauTilde} and \eqref{GeneralFormulaFort},
 completely determine $ r= r(\tilde\tau,r_\star)$ and
 $ t= t(\tilde\tau,r_\star)$ which in turn imply 
 a congruence of conformal geodesics. Observe that,
the coordinate $t$ is not well defined at the cosmological
and black hole horizons, located at  $r=r_c$ and  $r=r_b$, respectively.
Therefore, to describe a conformal geodesic crossing the horizons
 it is necessary to replace $t$ by either
$u$ or $v$; the retarded and advanced null coordinates defined 
in equation  \eqref{eq:ed-fink}.  
The $\tilde{\bmg}$-normalisation condition in the $(u,r)$ and
$(v,r)$ coordinates is written as
\begin{equation}\label{NormalisationNullCoords}
D(r)u'^2 -2u'r' = 1, \qquad \qquad D(r)v'^2
+2v'r' = 1
\end{equation}
respectively. The normalisation condition
\eqref{NormalisationNullCoords} renders 

\begin{equation}\label{NullEquations}
 u'= \frac{r' \pm \sqrt{D(r)+ r'^2}}{D(r)} ,
 \qquad \qquad v'= \frac{-r' \pm \sqrt{D(r)+
     r'^2}}{D(r)}.
\end{equation}

\noindent Using the chain rule and equations \eqref{NullEquations} and
\eqref{ReducedEquation} one gets
\begin{subequations}
\begin{eqnarray}
&&\frac{\mbox{d}u}{\mbox{d}r} = \frac{1}{D(r)} \left( 1 \pm \frac{
  |\gamma + \beta r| } {\sqrt{(\gamma + \beta r)^2 -
    D(r)}} \right), \label{EqNullu}\\
&& \frac{\mbox{d}v}{\mbox{d}r} = \frac{1}{D(r)} \left(- 1 \pm \frac{
  |\gamma + \beta r| } {\sqrt{(\gamma + \beta r)^2 -
    D(r)}} \right). \label{EqNullv}
\end{eqnarray}
\end{subequations}

\textbf{Remark.} For the subsequent analysis observe that  analysing expression
\eqref{EqNullu} for $u$ choosing the second sign is equivalent to analyse $-v$
choosing the first sign in \eqref{EqNullv} and vice versa.

\bigskip
In the following Sections we carry out a case-by-case discussion of the solutions of the reduced conformal
geodesics of Section \ref{Section:ExplicitExpressions} for the various
spherically symmetric vacuum spacetimes with Cosmological constant
introduced in Section \ref{Section:SdS-SadS}. 

\subsection{Explicit expressions for the reduced conformal deviation equations}
\label{Basic-setup-InitialDataDeviation}

A direct calculation shows that if the Lorentzian metric
$\tilde{\bml}$ has the form of equation \eqref{eq:2-l} then
\[
R[\tilde{\bml}]=-\partial^2_r D(r).
\]
 Therefore,  equation  \eqref{ReducedWarpProductDeviationEquation}
takes  the form
\begin{equation}
\not{\!\!D}_{\tilde{\bmx}'} \not{\!\!D}_{\tilde{\bmx}'}\tilde{\omega}
  = \left( \beta^2-\frac{1}{2}\partial^2_{r} 
D(r)\right)\tilde{\omega} + \not{\!\!D}_{\tilde{\bmz}} \beta.
\label{eq:gdeviationD}
\end{equation}

In order to derive initial data for $\tilde{\omega}$ it is necessary
to specify the initial data for the deviation vector $\bmz$.  To do
so, first observe that we can parametrise the congruence as
$r=r(\tilde{\tau},r_{\star})$ where $r_{\star}$ determines the
intersection of the curve with the initial hypersurface $\mathcal{S}$.
Alternatively, considering an approach similar to the one used in
\cite{LueVal13b} we can make use of the standard isotropic coordinate
defined by the relations in \eqref{eq:isotropic} and use
$\rho_{\star}$ to parametrise the congruence. This approach leads to
the choice $\bmz_{\star}={\bmpartial}_{\rho_{\star}}$ for the initial
data for the deviation vector. Accordingly, we consider equation
\eqref{eq:gdeviationD} with the choice
$\tilde{\bmz}_{\star}={\bmpartial}_{\rho_\star}$. Consistent with this
choice equation \eqref{eq:gdeviationD} reads
\begin{equation}
 \frac{\mbox{d}^2\tilde{\omega}}{\mbox{d}\tilde\tau^2}
=\left(\beta^2-\frac{1}{2}\partial^2_{r}D(r)\right)
\tilde{\omega}+\frac{\partial\beta}{\partial\rho_{\star}}\;,
 \qquad \qquad \frac{\partial\beta}{\partial
   \rho_\star}=\frac{r_\star\sqrt{D_\star}}{\rho_\star}
\frac{\partial\beta}{\partial   r_\star},
 \label{eq:gdeviationDiso}
\end{equation}
where $r=r(\tilde\tau,r_\star)$. The solution of this equation
will be determined once the initial conditions
$\tilde{\omega}_\star\equiv \tilde{\omega}(0,r_\star)$ and
$\tilde{\omega}'_\star\equiv \tilde\omega'(0,r_\star)$ have been prescribed
and the corresponding value of $\beta$ replaced.

\noindent 
To compute the initial data for $\tilde{\omega}$ we use equation
\eqref{eq:define-omega} with 
\[
{\bmepsilon}_{\tilde{\bml}}=\mathbf{d}t\wedge \mathbf{d}r\;,
\qquad \tilde{\bmx}'_\star={t}'_\star\bmpartial_t+{ r}'_\star\bmpartial_r\;,\qquad
\tilde{\bmz}_\star=\bmpartial_{\rho_\star}. 
\]
From the chain rule one finds that
\[
\bmpartial_{\rho_\star}=
\frac{r_\star}{\rho_\star}\sqrt{D_\star}\bmpartial_{r_\star}.
\]
Hence
\begin{eqnarray}
 && \tilde{\omega}_\star=(\mathbf{d}t\wedge \mathbf{d}r)
\left({ t}'{\bmpartial_t}+{ r}'{\bmpartial_r},\frac{r_\star}{\rho_\star}\sqrt{D_\star}
 \left( \frac{\partial r}{\partial r_\star}{\bmpartial_r}+
\frac{\partial t}{\partial r_\star}{\bmpartial_t}\right)
 \right) \nonumber \\
&& \phantom{\omega_\star}=\frac{r_\star}{\rho_\star}\sqrt{D_\star}\left({ t}'\frac{\partial r}{\partial r_\star}+
{ r}'\frac{\partial t}{\partial r_\star}\right)\;.
\label{eq:omega-value}
\end{eqnarray}
Combining expression \eqref{eq:omega-value} with the initial data
 \eqref{CGInitialData}
we get the initial conditions
\begin{equation}
 \tilde{\omega}_\star =\frac{r_\star}{\rho_\star}\;,\qquad
 \tilde{\omega}'_\star=0.
 \label{eq:initial-omega}
\end{equation}

\section{Analysis of the conformal geodesics in the subextremal Schwarzschild-de Sitter case}
\label{subsec:non-extremal}

In accordance with the discussion of Section \ref{SubSection:SdSSpecific},
for the subextremal Schwarzschild-de Sitter spacetime 
it is assumed that $\lambda >0$  and 
$0<M<2/(3\sqrt{3})$. In this case one expects to be able to construct
a congruence of conformal geodesics which combines characteristics of
analogous congruences in the de Sitter and Schwarzschild spacetime.

\subsection{Basic setup}
\label{InitialDataSdS}
The de Sitter spacetime can be covered by a congruence of conformal
geodesics which have no initial acceleration ---see
\cite{LueVal09,CFEBook}. Based on this result we consider as initial
hypersurface the time symmetric slice $\tilde{\mathcal{S}}=\{t=0\}$
and require the conformal geodesics to be orthogonal to
$\tilde{\mathcal{S}}$ ---see Figure \ref{fig:cgSchdS}. Exploiting the
periodicity of the spacetime, it is sufficient to restrict the
analysis to initial data in the range
\[
r_b \leq r_\star \leq r_c .
\]
If $r_\star\neq r_b,r_c$ then
we consider initial data for the congruence of conformal geodesics of
the form
\begin{equation}
t_\star=0, \qquad  t'_\star = \frac{1}{\sqrt{D_\star}},
\qquad r'_\star=0, \qquad \tilde{\beta}_{t\star}=0, \qquad
\tilde{\beta}_{r\star} =0,
\label{CGInitialData}
\end{equation}
so that the tangent vector $\tilde{\bmx}'$ coincides with the future
unit normal to $\tilde{\mathcal{S}}$.
\emph{It follows that in this
case $\beta^2=0$ so that the resulting congruence of conformal
geodesics is, after reparametrisation, a congruence of metric
geodesics.}
 Metric geodesics in the Schwarzschild-de Sitter spacetime
have been studied in the literature ---see \cite{HacLam08}. For
completeness, we here we carry out an independent analysis adapted to
the present setting.
Observe that equation \eqref{ReducedEquation} together with the condition 
$\beta^2=0$ readily yields 
\begin{equation}
\gamma=\sqrt{D_\star}.
\label{eq:gammabeta0}
\end{equation}
Using the latter one can write 
\begin{equation}
\gamma^2 -D(r)=
\frac{{ r}-r_\star}{ r}\left({ r}^2+{ r}
 r_\star-\frac{M}{r_\star}\right).
\label{eq:drm1}
\end{equation}
\bigskip
For convenience of the subsequent discussion we introduce the polynomial
\begin{eqnarray*}
&& Q(r)\equiv { r}^2+{ r} r_\star-\frac{M}{r_\star}
, \\
&& \phantom{Q(r)} =(r-\alpha_-)(r-\alpha_+),
\end{eqnarray*}
with
\[
\alpha_{\pm}=r_\star\left(-\frac{1}{2}\pm\sqrt{\frac{1}{4}+\frac{M}{r^3_\star}}\right).
\]
If $r_\star=r_c,\, r_b$ then the last three conditions in equation
\eqref{CGInitialData} can be still regarded a valid initial data but
not the conditions involving ${\tilde t}_\star$. 
The cases $r_{\star}=r_{b}$ and $r_{\star}=r_{c}$  and will be 
discussed in Sections
\ref{CoverConformalBoundarySubextremal} and \ref{CoverConformalBoundarySubextremalBifurcationBlackhole}. It is
convenient to define 
\[
 r_\circledast\equiv \left(\frac{M}{2}\right)^{1/3}.
\]

\noindent The role of $r_{\circledast}$ will be clarified Section
\ref{subsec:circledast}. An immediate observation
 regarding the location of $r_{\circledast}$ is the content 
of the following lemma.
\begin{lemma}
For $0< M < 2/3\sqrt{3}$ one has that
\[
r_b < r_\circledast < r_c.
\]
\label{lemma:rbrc}
\end{lemma}
\begin{proof}
In the parametrisation introduced in equation (\ref{eq:define-phi}) we have
\[
 r_\circledast=\left(\frac{\cos\phi}{3\sqrt{3}}\right)^{1/3}.
\]
In addition, if $\phi\in (0,\pi/2)$ one has the inequalities
\begin{eqnarray*}
&&\left(\frac{\cos\phi}{3\sqrt{3}}\right)^{1/3}>\frac{1}{\sqrt{3}}
\cos\frac{\phi}{3}-\sin\frac{\phi}{3}
\qquad \mbox{so that} \qquad  r_\circledast > r_b\;,\\
&&\left(\frac{\cos\phi}{3\sqrt{3}}\right)^{1/3}<\frac{1}
{\sqrt{3}}\cos\frac{\phi}{3}+\sin\frac{\phi}{3} 
 \qquad \mbox{so that} \qquad  r_\circledast< r_c.
\end{eqnarray*}
\end{proof}

\begin{remark}
{\em Observe that in the extremal case one has
$M=2/3\sqrt{3}$ consequently
$r_{\circledast}= r_{\mathcal{H}}$ where
$r_{\mathcal{H}}$ denotes the location of the Killing horizon.}
\end{remark}

 Given that ${\bmx}'_\star$ has been chosen to be orthogonal to
$\tilde{\mathcal{S}}$ and the range of $r_\star$ is bounded, we can set
\[
\tilde{\bmbeta}_\star={\bmbeta}_\star=0 \;,\qquad 
\Theta_\star=1, \qquad \dot{\Theta}_\star=0.
\]
 To simplify the subsequent expressions recall from Section
\ref{Section:SdS-SadS} that one can set $\lambda =-3\epsilon$ with 
$\epsilon=-1$ in the case of
a de-Sitter like Cosmological constant.
Consistent with this choice, equations 
\eqref{ConformalFactor}-\eqref{Constraints}  render
\begin{equation}
\Theta = 1- \frac{\tau^2}{4}.
\label{eq:value-of-theta}
\end{equation}
From this expression one can determine the relation between $\tilde{\tau}$
and $\tau$ using equation \eqref{Reparametrisation}. A calculation
gives
\begin{equation}
\tilde{\tau} = 2\mbox{arctanh}\left( \frac{\tau}{2} \right).
\label{eq:value-of-bartau}
\end{equation}
Observe that \[\lim_{\tau \rightarrow 2}\tilde{\tau}(\tau) = \infty,\]
so that the conformal boundary is reached
in a finite value of the unphysical proper time $\tau$.

\begin{remark}
{\em The coordinate $t$ is not well-defined
at the points where $D(r)$ vanishes. In the subextremal case, this
corresponds to the Cosmological and black hole horizons while in the
extremal case this corresponds to the Killing horizon. }
\end{remark}

 To analyse the
geodesics crossing the horizon one has to parametrise the curve in
terms of the coordinates $(u,r)$ or $(v,r)$, where $u$ and $v$ are
advanced and retarded null coordinates.  Setting $\beta=0$ equations
\eqref{NullEquations} and \eqref{ReducedEquation} render

\begin{equation}\label{choiceOfSignNullCoords}
  u'= \frac{1}{D(r)}(\sqrt{\gamma^2-D(r)} \pm
  |\gamma|), \qquad
  v'=\frac{1}{D(r)}(-\sqrt{\gamma^2-D(r)} \pm
  |\gamma|).
\end{equation}

\noindent Choosing the second sign in equation
\eqref{choiceOfSignNullCoords} and using L'H\^opital rule one finds
\begin{equation}
 \lim_{{r} \rightarrow {r}_{c}}u'
 =-\frac{1}{2|\gamma|},
\label{RetardedCoordinateCosomologicalHorizon}
\end{equation}
\noindent Similarly, choosing the first sign in equation
\eqref{choiceOfSignNullCoords} one gets
\begin{equation}
 \lim_{r \rightarrow r_b}v' =
 \frac{1}{2|\gamma|}. 
\label{AdvancedCoordinateBlackHoleHorizon}
\end{equation}

\noindent Therefore, one can use the coordinates $(u,r)$ to describe
the curves crossing the Cosmological horizon at $r=r_{c}$. Similarly, the
coordinates $(v,r)$ can be used to analyse curves crossing the black
hole horizon at $r=r_b$. 

\subsection{Qualitative analysis of the behaviour of the curves}
\label{Section:QualitativeSdS}
In this Section we analyse the different qualitative behaviours of the
conformal geodesics defined by the initial conditions introduced in
the previous Section. As it will be seen, in broad terms, there are
three types of conformal geodesics: those that escape to the
asymptotic points, those that escape to the conformal boundary and
those falling into the singularity.
 
\subsubsection{Conformal geodesics with constant $r$ (critical curve).}
\label{subsec:circledast}
In this Section we discuss whether it is possible to have conformal
geodesics for which $r$ is constant. Substituting the conditions
$r''=r'=0$ in equation \eqref{rprimeprime} one readily
obtains the condition
\[
\partial_{r} D(r)= \frac{M}{r^2}-2r=0. 
\]
The latter implies
\begin{equation}\label{CriticalCurveSubextremal}
r = r_\circledast \equiv \left(\frac{M}{2}\right)^{1/3}.
\end{equation}
The curve $r(\tau)=r_{\star}=r_{\circledast}$ will be called the
 \emph{critical curve}.

\begin{figure}[t]
\centering
\includegraphics[width=0.56\textwidth]{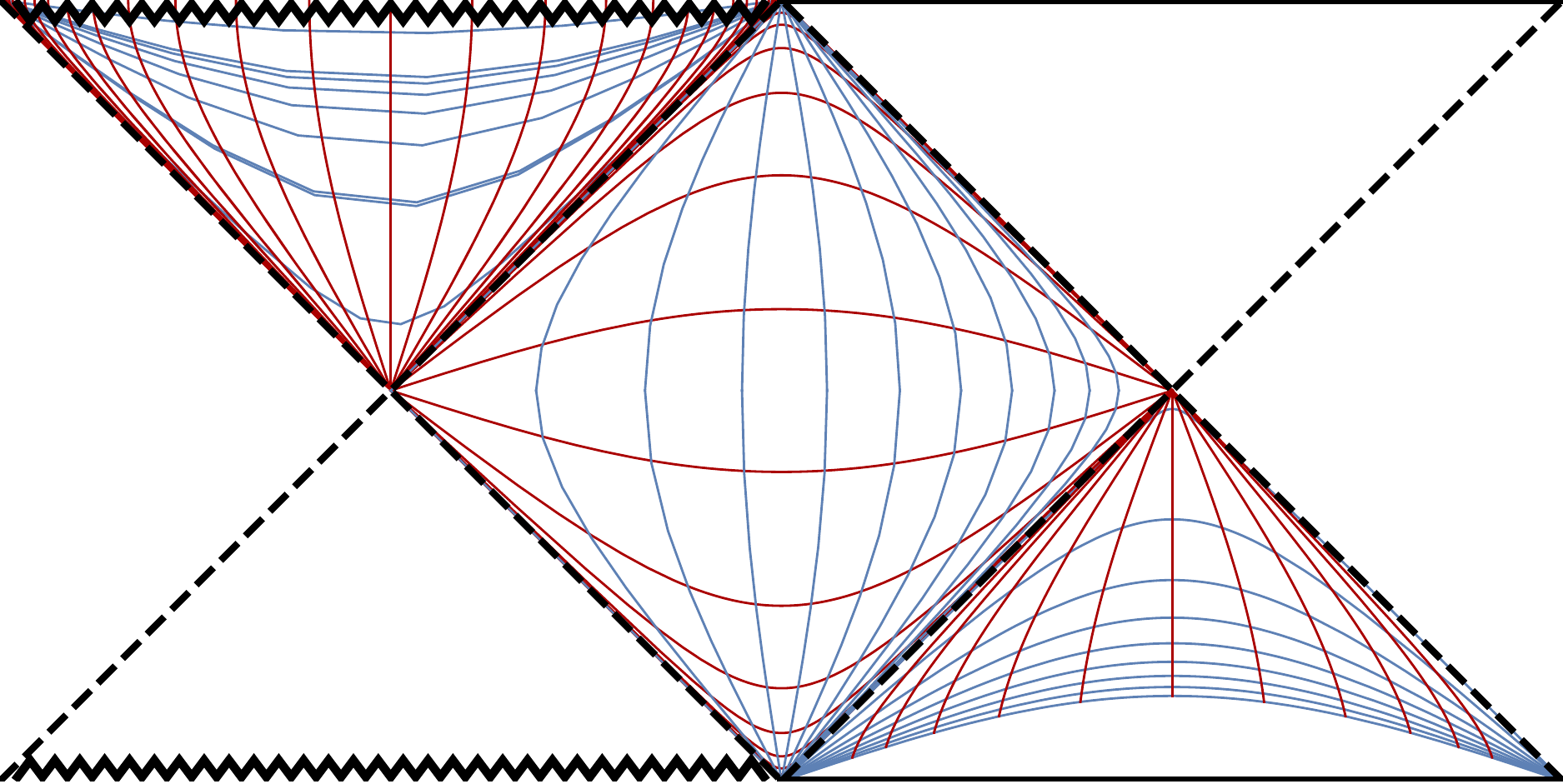}
\put(-240,120){\footnotesize{$\mathcal{Q}$}}
\put(-120,120){\footnotesize{$\mathcal{Q'}$}}
\put(-240,-10){\footnotesize{$\mathcal{Q}$}}
\put(-120,-10){\footnotesize{$\mathcal{Q'}$}} \put(0,
120){\footnotesize{$\mathcal{Q}$}} \put(0,
-10){\footnotesize{$\mathcal{Q}$}}
\put(-190,55){\footnotesize{$\mathcal{B}$}}
\put(-50,55){\footnotesize{$\mathcal{B'}$}}
\caption{ Curves with constant $t$ and $r$ (red and blue respectively)
  are plotted on the Penrose diagram of the subextremal
  Schwarzschild-de Sitter spacetime. Curves of constant $t$ intersect
  the bifurcation spheres $\mathcal{B}$, $\mathcal{B'}$ while the
  curves of constant $r$ approach the asymptotic points
  $\mathcal{Q}$ and $\mathcal{Q'}$. }
\label{fig:SdSConstantCurves}
\end{figure}
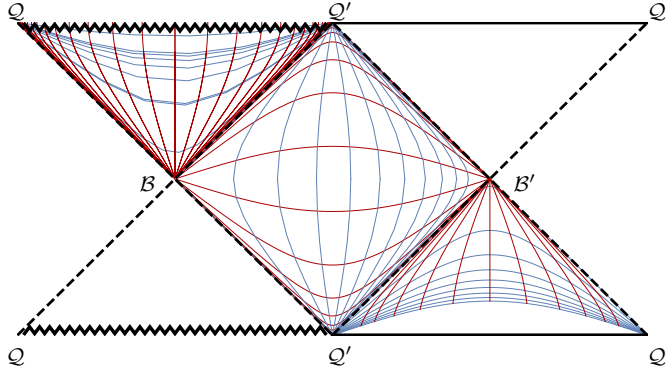

\begin{remark}
{\em Notice that the curve $r=r_{\star}=r_{\circledast}$ never crosses the horizon since
by virtue of Lemma \ref{lemma:rbrc} one has $r_{b}<r<r_{c}$. For this curves
$r$ is always finite  and
they approach one of the asymptotic points $\mathcal{Q}$ or $\mathcal{Q'}$.}
\end{remark}

In addition, one has the following result:

\begin{lemma}
\label{Lemma:RootsNonExtremalSdS} 
Assume $\lambda>0$ and $0<M<2/(3\sqrt{3})$, then one has the chain of
inequalities
\begin{eqnarray*} 
&& \alpha_-<0 <\alpha_+ < r_\circledast < r_\star\qquad \mbox{for}
  \qquad r_\star>r_\circledast ,\\ && \alpha_-<0 < r_\star <
  r_\circledast < \alpha_+ \qquad \mbox{for} \qquad
  r_\star<r_\circledast, \\ && \alpha_-<0 < r_\star= r_\circledast =
  \alpha_+ \qquad \mbox{for} \qquad r_\star=r_\circledast.
\end{eqnarray*}
\end{lemma}

\begin{proof}
To prove this result it is convenient to write $\alpha_\pm$ in terms
of $r_\circledast$. One finds that

\begin{equation}
\alpha_\pm=\frac{r_\star}{2}\left(-1\pm\sqrt{1+\left(\frac{2r_\circledast}{r_\star}\right)^3}\right).
\label{eq:alphacircledast}
\end{equation}

\noindent From the previous relation it is clear that $\alpha_-<0$ and that
$\alpha_-<\alpha_+$. If $r_\star<r_\circledast$ then

\begin{equation}
\sqrt{1+\left(\frac{2r_\circledast}{r_\star}\right)^3}>3\;,
\label{eq:alphapm}
\end{equation}

\noindent which entails $\alpha_+>r_\star$. Also for any $z>2$  we have

\begin{equation}
 z(z^2-z-2)>0 \qquad \mbox{if and only if} \qquad \sqrt{1+z^3}>1+z.
\end{equation}

\noindent If we set $z=2r_\circledast/r_\star$ in the last inequality and use the
inequality \eqref{eq:alphapm} we deduce $\alpha_+>r_\circledast$. This
proves the second chain of inequalities. Assume now that
$r_\star>r_\circledast$. Then

\begin{equation}
 \sqrt{1+\left(\frac{2r_\circledast}{r_\star}\right)^3}<3\;,
\label{eq:alphamp}
\end{equation}
which entails $\alpha_+<r_\star$. If $z<2$ then 
\begin{equation}
 z(z^2-z-2)<0 \qquad \mbox{if and only if} \qquad \sqrt{1+z^3}<1+z.
\end{equation}

\noindent Hence, setting $z=2r_\circledast/r_\star$ in the last
inequality and using, again, inequality \eqref{eq:alphapm} yields
$\alpha_+<r_\circledast$ thus proving the first chain of inequalities.
Finally, the last condition follows after setting $r_\star
=r_\circledast$ in expression \eqref{eq:alphacircledast}.
\end{proof}

\noindent Using equation \eqref{Equation-t'} and the initial data 
one readily finds that along the critical curve
\[
t=t(\tilde{\tau})= \frac{\tilde{\tau}}{\sqrt{D_\circledast}}.
\]
Hence, $t\rightarrow \infty$ as $\tilde{\tau}\rightarrow
\infty$. Alternatively, using formula \eqref{eq:value-of-bartau}, one
can express $t$ in terms of the unphysical proper time. One obtains
\begin{equation}\label{CoordinateTimeToUnphysicalProperTime}
t=\frac{2}{\sqrt{D_\circledast}}\mbox{arctanh}
\left( \frac{\tau}{2}\right). 
\end{equation}
To complete the discussion
observe that that along the critical curve one has
$r=r_{\star}=r_{\circledast}$, so that along the critical curve
$\newrbar[-1pt][-4pt]_{\circledast}=
\newrbar[-1pt][-4pt](r_{\circledast})$ is a constant.  Using equations
\eqref{eq:ed-fink} and \eqref{CoordinateTimeToUnphysicalProperTime}
one concludes that along the critical curve
\begin{equation*}
u(\tau)= \frac{2}{\sqrt{D_{\circledast}}}\arctan
\bigg(\frac{\tau}{2}\bigg)-\newrbar[-1pt][-4pt]_{\circledast}, \qquad
v(\tau)=\frac{2}{\sqrt{D_{\circledast}}}\arctan
\bigg(\frac{\tau}{2}\bigg) + \newrbar[-1pt][-4pt]_{\circledast},
\end{equation*}
thus at $\tau=2$ where $\Theta(\tau)$ vanishes, one has $u=\infty$ and
$v=\infty$.

\medskip
\noindent
\textbf{Asymptotic behaviour of the critical curve.} In the rest of this Section we will analyse
more closely the behaviour of the critical curve. 
Exploiting the  above notation and using the initial data \eqref{CGInitialData}, the integral 
\eqref{GeneralFormulaForTauTilde} can be then
rewritten as
\begin{equation}\label{rGeodesicIntegralCriticalCurve}
\tilde{\tau}=
\int^{r}_{r_{\star}}\sqrt{\frac{\bar{r}}{(\bar{r}-r_{\star})(\bar{r}-\alpha_{-}
(r_{\star}))(\bar{r}-\alpha_{+}(r_{\star}))}}\mbox{d}\bar{r}.
\end{equation}

To study the behaviour close to the critical curve consider
$r_{\star}=(1+ \epsilon)r_{\circledast}$. For small $\epsilon>0$ and
 $\bar{r}>r_{\star}$ one can expand the right hand side of
equation \eqref{rGeodesicIntegralCriticalCurve} in Taylor series as
\begin{equation}
\tilde{\tau}=
\int_{r_{\star}}^{r}\sqrt{\frac{\bar{r}}{\bar{r}+2r_{\circledast}}}
\left(\frac{1}{\bar{r}-r_{\circledast}}-\frac{3r_{\circledast}^2\bar{r}\epsilon^2
}{2(\bar{r}-r_{\circledast})^3}\right)\mbox{d}\bar{r} + \mathcal{O}(\epsilon^3).
\end{equation}
Integrating we obtain
\begin{eqnarray*}
\tilde{\tau}=
-\frac{2}{\sqrt{3}}\text{arctanh}\left( \sqrt{3}\sqrt{\frac{1+
    \epsilon}{3+\epsilon}}\right) + 2\ln\left(
\sqrt{r_{\circledast}(1+\epsilon)}+\sqrt{r_{\circledast}(3+\epsilon)}
\right) -\frac{2}{\sqrt{3}}\text{arctanh} \left(
\frac{3r}{r+2r_{\circledast}} \right) \\ + 2\ln 
\left(\sqrt{\tilde{r}} +
\sqrt{r + 2r_{\circledast}} \right)
\ -\frac{3}{4}r_{\circledast}\sqrt{1+2r_{\circledast}}(1+ 2\epsilon)
-\frac{3}{4}r_{\circledast}{}^2\sqrt{1+2r_{\circledast}}\frac{(2r
-r_{\circledast})\epsilon^2}{(r_{\circledast}-r)^2}
+ \mathcal{O}(\epsilon^3).
\end{eqnarray*}
As expected,  the last expression
diverges as $\epsilon \rightarrow 0$.
 The divergent term can be expanded for small $\epsilon>0$ 
as
\[
\text{arctanh} \left(\sqrt{3}\sqrt{\frac{1+\epsilon}{3+\epsilon}}
\right) = \frac{1}{2}\ln \left(  \left\lvert-\frac{6}{\epsilon} + 4 +
\frac{\epsilon}{6}+ \mathcal{O}(\epsilon^2) \right\rvert \right)
\]
and the second term can be expanded as
\[
 \ln\left(
 \sqrt{r_{\circledast}(1+\epsilon)}+\sqrt{r_{\circledast}(3+\epsilon)}
 \right) = \ln \left( (1 + \sqrt{3})\sqrt{r_{\circledast}}\right) +
 \frac{\epsilon}{2\sqrt{3}} - \frac{\epsilon^2}{6\sqrt{3}} +
 \mathcal{O}(\epsilon^3).
\]
Hence, to leading order one has
\[
\tilde{\tau}(r) =\frac{1}{\sqrt{3}}\ln \epsilon + f(r) +
\mathcal{O}(\epsilon)
\]
where
\[
f(r)= 2\ln \left( (1 +
\sqrt{3})\sqrt{r_{\circledast}}\right) -\frac{2}{\sqrt{3}}
\text{arctanh} \left( \frac{3r}{r+2r_{\circledast}} \right) 
\\ + 2\ln
\left(\sqrt{r
} + \sqrt{r + 2r_{\circledast}} \right)
\ -\frac{3}{4}r_{\circledast}\sqrt{1+2r_{\circledast}}.
\]

\noindent Rewriting the last expression in terms of the 
unphysical proper time using equation \eqref{eq:value-of-bartau} one has,
to leading order, that
\begin{equation}
\tau(\epsilon) =\frac{2(\epsilon^p-1)}{\epsilon^{p}+ 1}
\end{equation}
\noindent where $p=1/\sqrt{3}$. Differentiating we get
\[
\frac{\mbox{d}\tau}{\mbox{d}\epsilon}= \frac{4p\epsilon^{p-1}}{(\epsilon^p +1)^2}.
\]

\noindent Since $p<1$ then one has that
$\mbox{d}\tau/\mbox{d}\epsilon$ diverges as $\epsilon \rightarrow 0$.
Consequently, the critical curve becomes tangent to the horizon
as it
approaches the asymptotic points $\mathcal{Q}$ and $\mathcal{Q'}$.  In
other words, the critical curve becomes null asymptotically.  This
behaviour is similar to that observed in the Schwarzschild spacetime
where $\mathcal{Q}$ plays the role of timelike infinity $i^{+}$ as
discussed in \cite{Fri03a} and the subextremal Reissner-Nordstr\"om
spacetime in \cite{LueVal13b}.

\begin{remark}
{\em The change in the the causal behaviour of the critical curve
  discussed in the previous Section is evidence of the strong singular
  behaviour of the conformal structure at the asymptotic points. }
\end{remark}

\subsubsection{Conformal geodesics with $r_\circledast<r_\star\leq r_c$ }
\label{subsec:r>r*}
Direct evaluation of equation \eqref{rprimeprime} at $\tilde{\tau}=0$
for $r_\star >r_\circledast$ shows that $r''_\star >0$. As
$r'_\star=0$ one has that $r_\star$ is a local minimum
of $r$. Hence, the curves
with $r_\circledast<r_\star\leq r_c$ have $r$ initially
increasing ---i.e $r'>0$. Accordingly, making use of the relations see equations
\eqref{eq:gammabeta0} and\eqref{eq:drm1} one is led to analyse the
ordinary differential equation  
\begin{equation}
r' =\sqrt{D_\star -D(r)} = \sqrt{\frac{ r-r_\star}{r_\star } Q( r)}.
\label{SdS:Increasing}
\end{equation}

Making use of Lemma \ref{Lemma:RootsNonExtremalSdS} one sees that the
turning points of the equation are located at a value of $r$
which is smaller than $r_\star$. As $r$ is initially
increasing it follows then that $r'\neq 0$ for
$\tilde{\tau}\in (0,\infty)$. Moreover, as $r>r_\circledast$
it follows that $r''\neq 0$ for $\tilde{\tau} \in
(0,\infty)$. Finally, from equation \eqref{SdS:Increasing} it follows
that for the initial data under consideration $r'\rightarrow
\infty$ if and only if $r\rightarrow \infty$. It only remains
to be checked that $r'$ (or $r$) do not blow up in a
finite amount of physical proper time $\tilde{\tau}$. From
equation \eqref{SdS:Increasing} the physical proper time is given by
\begin{equation}
\tilde\tau=\int_{r_\star}^{ r}\sqrt{\frac{r_\star \bar{r}}{(\bar{r}-r_\star)\big(r_\star
    \bar{r}(\bar{r}+r_\star)-M\big)}} \mbox{d} \bar{r}. 
\label{eq:ds-propertime1}
\end{equation}

According to Lemma \ref{Lemma:RootsNonExtremalSdS} if
$r_\star>r_\circledast$ one has that $r_\circledast>\alpha_+$. 
Then, for any $r>r_\circledast$ it holds that $r>\alpha_+$.
Given that $r_\star>0$ we have
\[
r_\star r(r+r_\star)-M>0. 
\]
Therefore the integral (\ref{eq:ds-propertime1}) is real for any
$\tilde r>r_\star$. In addition,  one has that 
\[
\frac{r_\star r}{(r-r_\star)\big(r_\star r(r+r_\star)-M\big)}\geq
\frac{r_\star r}{r_\star r (r-r_\star)(r+r_\star)}=
\frac{1}{r^2-r_\star^2}\geq
\frac{1}{r^2}.
\]
Thus
\[
\int_{r_\star}^{r}\sqrt{\frac{r_\star \bar{r}}{(\bar{r}-r_\star)\big(r_\star
    \bar{r}(\bar{r}+r_\star)-M\big)}} \mbox{d}\bar{r}
\geq \int_{r_\star}^{ r}\sqrt{\frac{1}{\bar{r}^2}} \mbox{d}\bar{r}=
\log r-\log r_\star.
\]
Hence
\[
r(\tilde{\tau}) \leq r_{\star}e^{\tilde{\tau}}, 
\]
showing that $r$ does not  blow up in finite physical proper time. 

\begin{remark}
{\em Summarising, the analysis in the previous paragraphs shows that
  if $r_\circledast<r_\star\leq r_c$, then 
  $r(\tilde{\tau})\rightarrow\infty$ as $\tilde{\tau}\rightarrow \infty$. }
\end{remark}

As discussed in a remark given in Section \ref{InitialDataSdS} the
coordinate $t$ is not well behaved at the horizon. To follow the
curves crossing the Cosmological  horizon at $r=r_c$ one needs to use
null coordinates $(u,r)$.  In the following we first show that along
the conformal geodesics the coordinate $u$ remains well defined 
for any finite value of $r$, in particular,
for $r=r_{c}$. Then, we show that $u$ is
finite at the conformal boundary. As discussed in Section
\ref{InitialDataSdS} it is enough to analyse equation \eqref{EqNullu}
choosing the second sign so that $\mbox{lim}_{r \rightarrow
  r_{b}}u' = -1/2|\gamma|$. To start the discussion observe that the coordinate $u$
 is determined by integration 
of equation \eqref{EqNullu} by 
\begin{equation}
u( r) - u_{\star} = \int ^{r} _{r_\star}
\frac{1}{D(\bar{r})} \left( 1 - \frac{|\gamma|}{\sqrt{\gamma^2 - D(\bar{r})}}
\right) \mbox{d}\bar{r}.
\label{UCrossingHorizon}
\end{equation}

\medskip
\noindent
\textbf{Analysis for the region with $r_{\star}\leq r_{c}$.} Observe
that, by continuity, one has that $|u(r) - u_{\star}|$ for $r \in
[r_{\star},r_{c}]$ would diverge if and only if the derivative
$\mbox{d}u/\mbox{d}r$ diverges for some $r_{q} \in
[r_{\star},r_{c}]$. Direct inspection of equation \eqref{EqNullu}
using expressions \eqref{eq:DrExplicitRoots}, \eqref{eq:RelationRoots}
and \eqref{eq:drm1} show that $\mbox{d}u/\mbox{d}r$ is finite for $r
\in [r_{\star},r_{c})$.  Thus it only remains to compute the
limit
\[
\lim_{r \rightarrow r_{c}}\frac{\mbox{d}u}{\mbox{d}r}.
\]
Choosing the second sign in equation \eqref{EqNullu} and rearranging one has
\begin{equation}
\frac{\mbox{d}u}{\mbox{d}r} = 
\frac{1}{D({r})} \left( \frac{\sqrt{\gamma^2 - D({r})} - |\gamma|}{\sqrt{\gamma^2 - D({r})}}\right).
\end{equation}
A direct application of the L'H\^opital rule shows that
\[
\lim_{r \rightarrow r_{c}}\frac{\mbox{d}u}{\mbox{d}r}=-\frac{1}{2\gamma^2}.
\]
Consequently $u(r)$ does not diverge for  $r \in [r_{\star},r_{c}]$. In other words
\begin{equation}\label{eq:EstimateSubextremal1}
| u(r_{c})-u_{\star}| \leq k^2 
\end{equation}
for some constant $k$.

\medskip
\noindent
\textbf{Analysis for the region with $r>r_{c}$.} First let us found
an appropriate estimate for the integral \eqref{UCrossingHorizon}  for
$r>r_{c}$ and show that the coordinate $u$ acquires a finite value at
$\mathscr{I}$.
\noindent Using that $D(r)\leq0$ for $r>r_{c}$, we observe that
\begin{equation}
\sqrt{\gamma^2-D(r)} \geq |\gamma|.
\label{obs1Sub}
\end{equation}

\noindent Recall that according to the initial data 
given in Section \ref{InitialDataSdS} 
 one has $\gamma= \sqrt{D}_{\star} \neq 0$.
 Exploiting the last observation we find an estimate for
$|u|$ as follows:

\begin{eqnarray}
 |u(r) - u(r_{c})| = \left|\int ^{r}
 _{r_{c}} \frac{1}{D(\bar{r})} \left( 1 -
 \frac{|\gamma|}{\sqrt{\gamma^2 - D(\bar{r})}} \right)
 \mbox{d}\bar{r}\right| \nonumber
 \\ \leq \int ^{r} _{r_{c}} \frac{1}{|D(\bar{r})|} \left|1
 - \frac{|\gamma|}{\sqrt{\gamma^2 - D(\bar{r})}} \right|
 \mbox{d}\bar{r}. \label{estimatingIntegral1Sub}
\end{eqnarray}

\noindent To estimate the last integral of expression
\eqref{estimatingIntegral1Sub} observe that inequality
 \eqref{obs1Sub} can be written
as $ 1 \geq{|\gamma|}/{\sqrt{\gamma^2-D(r)}},$ equivalently \[1
-\frac{|\gamma|}{\sqrt{\gamma^2-D(r)}} \geq 0, \] 
thus
\[
\frac{1}{|D(r)|} \left| 1- \frac{|\gamma|}{\sqrt{\gamma^2-D(r)}}\right|
= \frac{1}{|D(r)|} \left( 1-
\frac{|\gamma|}{\sqrt{\gamma^2-D(r)}}\right).
\]

\noindent Using again  inequality \eqref{obs1Sub}, written as $ 1 \leq
          {\sqrt{\gamma^2-D(r)}}/{|\gamma|} $ and taking into account
          that $D(r)\leq 0$, we get the following

\begin{eqnarray*}
&& \frac{1}{|D(r)|} \left( 1-
\frac{|\gamma|}{\sqrt{\gamma^2-D(r)}}\right) \leq \frac{1}{|D(r)|}
\left( \frac{\sqrt{\gamma^2 - D(r)}}{|\gamma|}-
\frac{|\gamma|}{\sqrt{\gamma^2-D(r)}}\right) = \frac{1}{|D(r)|}
\frac{\gamma^2 - D(r) -\gamma^2}{|\gamma|\sqrt{\gamma^2 -D(r)}} \\
&& \hspace{5cm}=\frac{1}{|D(r)|} \frac{|D(r)|}{|\gamma|\sqrt{\gamma^2-D(r)}} =
\frac{1}{|\gamma|\sqrt{\gamma^2 -D(r)}}.
\end{eqnarray*}

\noindent Finally, using again inequality \eqref{obs1Sub}, we have that
\[
\frac{1}{\sqrt{\gamma^2-D(r)}} \leq \frac{1}{|\gamma|},
\]
\noindent thus,

\[
\frac{1}{|\gamma|\sqrt{\gamma^2 -D(r)}} \leq \frac{1}{\gamma^2}.
\]

\noindent Following the last chain of inequalities one can estimate
the integral \eqref{estimatingIntegral1Sub} as follows

\begin{equation}\label{FirstEstimateSub}
|u(r) - u(r_{c})| =\left|
\int^{r}_{r_{\star}} \frac{1}{D(\bar{r})} \left( 1 -
\frac{|\gamma|}{\sqrt{\gamma^2 - D(\bar{r})}} \right)\mbox{d}\bar{r} \right|
 \leq  \int
^{r} _{r_\star} \frac{1}{\gamma^2} \mbox{d}\bar{r}  =
\frac{r-r_{\star}}{\gamma^2}.
\end{equation}

\noindent  Using equation \eqref{LocationHorizon2} and the
triangle inequality one can verify that
 the location of the cosmological horizon is
constrained by  \[1-\frac{1}{\sqrt{3}}\leq r_{c}\leq 1+
\frac{1}{\sqrt{3}}.\]
Now, consider
\[
r_{\bullet}>\max\{1,r_{c}\},
\]
and let  $u_{\infty} \equiv
\lim_{r\rightarrow \infty}u$. Using this notation we can
estimate $|u_{\infty}-u(r_{\bullet})|$ as follows

\[
|u_{\infty} - u(r_{\bullet})| =\left| \int^{\infty}_{r_{\bullet}}
 \frac{1}{D(\bar{r})}
\left( 1 - \frac{|\gamma|}{\sqrt{\gamma^2 - D(\bar{r})}} \right)\mbox{d}\bar{r}
 \right| \leq  
\int^{\infty}_{r_{\bullet}} \left| \frac{1}{D(\bar{r})} 
\left( 1 - \frac{|\gamma|}{\sqrt{\gamma^2 - D(\bar{r})}} \right)\right|\mbox{d}\bar{r} \nonumber 
\]
\begin{equation}\label{estimateIntegral2Sub}
  \leq \int ^{\infty}_{r_\bullet} \frac{1}{|D(\bar{r})|} +
  \frac{1}{|D(\bar{r})|}\left|
  \frac{|\gamma|}{\sqrt{\gamma^2-D(\bar{r})}} \right| \mbox{d}\bar{r}
  \leq \int ^{\infty} _{r_\bullet} \frac{2}{|D(\bar{r})|}
  \mbox{d}\bar{r}.
\end{equation}

\noindent In the last inequality we have used equation \eqref{obs1Sub}, in the
form ${|\gamma|}/{\sqrt{\gamma^2 - D(r)}} \leq 1$.  Using the
functional form for $D(r)$ in the subextremal case, one has that, for
$r>r_{\bullet}$

\[
|D(r)|=-D(r) = \frac{1}{r}(r^3-r+m)>r^2-1.
\]
Consequently,
\[
\frac{1}{|D(r)|}< \frac{1}{r^2-1}.
\]
Integrating the last inequality we get

\[\int_{r_{\bullet}}^{r}\frac{2}{|D(\bar{r})|} \mbox{d}\bar{r} < \ln 
\bigg|\frac{1-r}{1+r}\bigg|^{r}_{r_{\bullet}}.
\]
Thus,
\begin{equation}\label{SecondEstimateSub}
|u_{\infty} - u(r_{\bullet})| <  \ln \bigg|
\frac{1+r_{\bullet}}{1-r_{\bullet}}\bigg| .
\end{equation}

\noindent   Consequently, using expressions
\eqref{FirstEstimateSub}, \eqref{SecondEstimateSub} and the triangle
inequality one obtains

\begin{equation}\label{eq:EstimateSubextremal2}
 |u_{\infty} - u(r_{c})| \leq |u_{\infty}-u(r_{\bullet})|
 +|u(r_{\bullet})-u(r_{c})|< \ln \bigg|
\frac{1+r_{\bullet}}{1-r_{\bullet}}\bigg|
+  \frac{r_{\bullet}-r_{c}}{\gamma^2} < \infty. 
\end{equation}
Finally from expressions \eqref{eq:EstimateSubextremal1} and \eqref{eq:EstimateSubextremal2} one concludes that

\begin{equation}\label{eq:EstimateSubextremal3}
 |u_{\infty} - u_{\star}| \leq |u_{\infty}-u(r_{c})|
 +|u(r_{c})-u(r_{\star})|< \ln \bigg|
\frac{1+r_{\bullet}}{1-r_{\bullet}}\bigg|
+  \frac{r_{\bullet}-r_{c}}{\gamma^2} + k^2 < \infty. 
\end{equation}

\begin{remark}
{\em Summarising, the analysis in the previous paragraphs shows that
  the conformal geodesics with $r_\circledast<r_\star\leq r_c$ cross the 
Cosmological horizon at
  $r=r_c$ and eventually escape to the conformal boundary. In
  particular, the curves do not touch the asymptotic points
  $\mathcal{Q}$ and $\mathcal{Q'}$ for which $u=\infty$ and $u=-\infty$.}
\end{remark}

\subsubsection{Conformal geodesics with $r_b\leq r_*< r_\circledast$}
\label{subsec:r<r*}
Direct evaluation of equation \eqref{rprimeprime} at $\tilde{\tau}=0$
for $r_b\leq r_\star<r_\circledast$ shows that $r''_\star<0$
so that $r_\star$ is a local maximum of $r$ ---that is,
$r'_\star<0$. Hence, the curves with $r_b\leq
r_\star<r_\circledast$ have $r$ initially decreasing and one
has
\[
r' =-\sqrt{\gamma^2-D( r)}.
\]
From the above relation we deduce that the proper time $\tilde\tau$ of
the curves is given by
\begin{equation}
 \tilde\tau=\int_{ r}^{r_\star}\sqrt{\frac{r_\star
     \bar{r}}{(r_\star-\bar{r})\big(M-r_\star \bar{r}(\bar{r}+r_\star)\big)}}
 \mbox{d}\bar{r}.
\label{eq:ds-propertime2}
\end{equation}
Lemma \ref{Lemma:RootsNonExtremalSdS} tells us that if
$r_\star<r_\circledast$ one has that
$\alpha_-<0<r_\circledast<\alpha_+$ and therefore for any
$0<r<r_\circledast$ it holds that $\alpha_-<r<\alpha_+$. Hence
$Q(r)<0$ implies that
\[
M-r_\star r(r+r_\star)>0\;,\qquad 0<r<r_\circledast.
\]
This means that the integrand in equation \eqref{eq:ds-propertime2} is
real for any $r$ in $[0,r_\circledast]$. In addition, one has
that
\[
\frac{r_\star r}{(r_\star-r)\big(M-r_\star
  r(r+r_\star)\big)}\leq\frac{r_\star^2}{(r_\star-r)}\;,\qquad
0<r<r_\circledast.
\]
Consequently, one can conclude that
\[
\tilde\tau=\int_{ r}^{r_\star}\sqrt{\frac{r_\star
    \bar{r}}{(r_\star-\bar{r})\big(M-r_\star \bar{r}(\bar{r}+r_\star)\big)}}
 \mbox{d}\bar{r}\leq
r^2_\star \int_{ r}^{r_\star}\frac{\mbox{d}\bar{r}}{\sqrt{r_\star
    -\bar{r}}}=2r^2_\star\sqrt{r_\star- r}.
\]
This means that $\tilde\tau$ remains finite for any value of $r$
 in the interval $[0,r_\star ]$. In particular the geodesics reach
the singularity $ r=0$ in a finite value of the proper time
$\tilde\tau$.

\medskip
\noindent
\textbf{Behaviour across the black hole horizon at $r=r_b$.} To
complete the discussion of this class of conformal geodesics we now
show that the advanced time coordinate $v$ can be used to follow the
curves across the horizon at $r=r_b$ and that the value of $v$ remains
finite as $r\rightarrow 0$. Given that $v'$ is finite at $r=r_b$ ---cf. the limit in equation
\eqref{AdvancedCoordinateBlackHoleHorizon}--- it follows that the
coordinate $v$ has a finite value at $r=r_b$. Moreover, it can be
readily verified that
\[
\lim_{r\rightarrow 0} v' =0.
\]
Thus, the coordinate $v$ acquires a finite limit $v_\lightning$ as
$r\rightarrow 0$. 

\begin{remark}
{\em Summarising, the conformal geodesics with $r_b\leq r_*< r_\circledast$
reach the black hole horizon at $r=r_b$ in a finite amount of proper
time and fall into the singularity at $r=0$ also in a finite proper
time. Further, the curves remain away from the asymptotic points
$\mathcal{Q}$ and $\mathcal{Q}'$ for which $v=\infty$ and $v=-\infty$.  }
\end{remark}

\subsubsection{Conformal geodesics through the bifurcation sphere of
  the cosmological horizon}
\label{CoverConformalBoundarySubextremal}

As shown in previous Sections the conformal geodesics with initial
data $r_{\star}>r_{\circledast}$ arrive at the conformal boundary at a
finite value of the unphysical proper time $\tau$. However it is
possible that the conformal geodesics accumulate in certain
region of $\mathscr{I}$. In order to see that this is not the case we
analyse the behaviour of the conformal geodesics starting at the
bifurcation sphere of the horizon at $r=r_c$.  

In order to proceed with the analysis, it is convenient to
introduce Kruskal type coordinates. First, recall that the
Schwarzschild-de Sitter metric can be written as
\[
\tilde{\bmg}=\tilde{\bml}-r^2\bm\sigma,
\]
where the induced metric $\tilde{\bml}$ on the quotient manifold
$\tilde{\mathcal{M}}/SO(3)$ as defined in equation
\eqref{eq:2-l}. Introducing \emph{Kruskal} coordinates via

\[
U =\frac{1}{2} \exp(bu), \qquad V= \frac{1}{2}\exp(bv),
\]
where $u$ and $v$ are the Eddington-Finkelstein null 
coordinates as defined in equation \eqref{eq:ed-fink} and 
$b$ is a constant which  can be freely chosen. 
In this coordinate system the metric
$\tilde{\bml}$ reads 
\[
\tilde{\bml}=G(r)(\mathbf{d}U\otimes \mathbf{d}V +
\mathbf{d}V\otimes \mathbf{d}U),
\]
with
\[
G(r)\equiv \frac{D(r)}{b^2} \exp(-2b\newrbar[-1pt][-4pt](r)),
\]
where $\newrbar[-1pt][-3.7pt](r)$ is the tortoise coordinate as
defined in equation \eqref{tortoise}. A straightforward computation
using \eqref{eq:DrExplicitRoots} taking into account that 
$r_{-}<0<r_{b}<r_{c}$ renders
\begin{multline}\label{TortoiseExplicitFunction}
\newrbar[-1pt][-3.7pt](r)=
 \frac{|r_{-}|}{(r_{c}+|r_{-}|)(r_{b}+ |r_{-}|)} \ln(r+ |r_{-}|)
+ \frac{r_{b}}{(r_{b}+|r_{-}|)(r_{c}-r_{b})}\ln |r-r_{b}|  
\\ -\frac{r_{c}}{(r_{c}+ |r_{-}|)(r_{c}-r_{b})}\ln|r-r_{c}|.
\end{multline}
One can verify that $\lim_{r \to \infty} \newrbar[-1pt][-3.7pt](r)=0$
and that 
\begin{equation}
G(r)=\frac{1}{b^2r}(r +
|r_{-}|)^{A_{n}}(r-r_{b})^{A_{b}}(r-r_{c})^{A_{c}},
\label{KruskalG}
\end{equation}
where $A_{n}, A_{b},A_{c}$ are constants defined via
\begin{equation*}
A_{n}\equiv 1 -\frac{2b|r_{-}|}{(r_{c}+|r_{-}|)(r_{b}+|r_{-}|)},\qquad
A_{b}\equiv 1 -\frac{2b r_{b}}{(r_{c} -r_{b})(r_{b}+|r_{-}|)}, \qquad
A_{c}\equiv 1 +\frac{2b r_{c}}{(r_{c} -r_{b})(r_{c}+|r_{-}|)}.
\end{equation*}
As discussed in \cite{BazFer85} one cannot construct coordinates that are
regular in the neighbourhood of the Cosmological and black hole horizons
simultaneously.
 Since in this Section we
are only interested in analysing the conformal geodesic starting at
$r_{\star}=r_{c}$ we observe that setting
\[
b \equiv -\frac{1}{2r_{c}}(r_{c}-r_{b})(r_{c} + |r_{-}|),
\]
so that $A_{c}=0$, one gets
\[
\tilde{\bml}=\frac{1}{b^2r}(r+
|r_{-}|)^{A_{n}}(r-r_{b})^{A_{b}}(\mathbf{d}U\otimes \mathbf{d}V +
\mathbf{d}V\otimes \mathbf{d}U).
\]
Introducing a further change of coordinates $ T= U + V$ and $\Psi =
V-U$ one obtains
\begin{equation}\label{QuotientMetricKruskal}
\tilde{\bml}=\frac{1}{2b^2r}(r +
|r_{-}|)^{A_{n}}(r-r_{b})^{A_{b}}(\mathbf{d}T\otimes\mathbf{d}T -\mathbf{d}\Psi\otimes\mathbf{d}\Psi),
\end{equation}
where the coordinates $T$ and $\Psi$ are related to $r$ and $t$ via
\begin{eqnarray}
T(r,t)=\cosh t \exp(b \hspace{1mm}\newrbar[-1pt][-3.7pt](r) ), \qquad
\Psi(r,t) =\sinh t \exp (b \hspace{1mm} \newrbar[-1pt][-3.7pt](r)).
\end{eqnarray}
A straightforward computation renders the additional relation
\begin{equation}\label{eg:ImplicitFunction-r-in-terms-of-T-Psi}
T^2 - \Psi^2 = (r + |r_{-}|)^{k_{n}}(r-r_{b})^{k_{b}}(r-r_{c}),
\end{equation}
with $k_{n} \equiv 1 - A_{n}<0$ and $k_{b}\equiv 1-A_{b}<0$.  Observe that
the metric \eqref{QuotientMetricKruskal} is regular at $r_{c}$.
Moreover, notice that, since
\[
\lim_{r \to r_{c}}\newrbar[-1pt][-3.7pt](r)= \infty
\]
and $b<0$ then  \[\lim_{r \to r_{c}}T=0, \qquad \lim_{r \to r_{c}}\Psi=0.\]
Consequently, the conformal geodesic starting at the bifurcation sphere
in these coordinates correspond to the geodesic with initial position 
$(T_{\star},\Psi_{\star}) = (0,0)$. Consistent with the
discussion of Sections 
\ref{sec:FormulaeWarpedProductSpaces}-\ref{InitialDataSdS} 
we consider curves with no evolution
in the angular coordinates and  $\beta=0$, 
so in these coordinates, the
geodesic equations  read

\begin{eqnarray}\label{eq:GeoEqKruskal}
&& T'' + \frac{\partial \ln G}{\partial \Psi} T'\Psi' +
  \frac{1}{2}\frac{\partial \ln G}{\partial T} (T'^2 + \Psi'^2)=0,\\ &&
  \Psi'' + \frac{\partial \ln G}{\partial T} T'\Psi' +
  \frac{3}{2}\frac{\partial \ln G}{\partial \Psi} \Psi'^2
  -\frac{1}{2}\frac{\partial \ln G}{\partial \Psi} T'^2 =0,
\end{eqnarray}
where as before $'$ denotes a derivative respect to the physical
proper time $\tilde{\tau}$.  Now, consider the curve with initial conditions
\begin{equation}\label{eg:DataGeoKruskal}
 T_{\star}=0, \qquad \Psi_{\star}=0, \qquad T'_{\star} \neq 0, \qquad \Psi'_{\star}=0,
\end{equation}
and recall that in the asymptotic region the curves with constant $t$
approach the bifurcation sphere ---see Figure
\ref{fig:SdSConstantCurves}.  Moreover, observe that
 in the $(T,\Psi)$-coordinate
system the curve with $t=0$ 
is described by $T=\exp(b \hspace{1mm}
\newrbar[-1pt][-4pt](r))$ and $\Psi=0$. 
In the subsequent discussion
 we show that
this curve is the required geodesic with the initial data given in
equation \eqref{eg:DataGeoKruskal}.  Assuming $\Psi=0$
equations
\eqref{eq:GeoEqKruskal} render the conditions
\begin{eqnarray}\label{eq:ConsistentGeoKruskal}
&& T'' + \frac{1}{2}\frac{\partial \ln G}{\partial T} T'^2
  =0,\label{eq:ConsistentGeoKruskal1} \\ && \frac{\partial \ln
    G}{\partial \Psi}\bigg|_{\Psi=0}
  =0.\label{eq:ConsistentGeoKruskal2}
\end{eqnarray}
Observe that equation \eqref{eq:ConsistentGeoKruskal1} can be
rewritten as an equation for $r$ while equation
\eqref{eq:ConsistentGeoKruskal2} has to be verified with the implicit
expression for $G(T,\Psi)$. Using the chain rule we have
\[
\frac{\partial \ln G}{\partial \Psi}\bigg|_{\Psi=0}=
\bigg(\frac{1}{G}\frac{\partial G}{\partial
  r}\bigg)\bigg|_{r=r_{c}}\frac{\partial r}{ \partial
  \Psi}\bigg|_{\Psi=0}.
\]
Using the implicit function theorem and equation
\eqref{eg:ImplicitFunction-r-in-terms-of-T-Psi} one can compute
$\partial r/\partial \Psi$ to obtain
\[
\frac{\partial r}{ \partial \Psi}= \frac{-2\Psi}{(r + |r_{-}|
  )^{k_{n}} (r-r_{b})^{k_{b}} \bigg( 1 + \displaystyle\frac{k_{c}(r-r_c)}{r-r_{b}}
  +\displaystyle\frac{k_{c}(r-r_{c})}{r+|r_{-}|} \bigg)}.
\]
Notice that $G(r_{c})\neq 0$ and ${\partial G}/{\partial
  r}\big|_{r=r_{c}} \neq 0$ while ${\partial r}/{ \partial \Psi}$
vanishes for $\Psi=0$ as long as $r \neq r_{b}$. 
Since $r(\tilde{\tau})\geq r_{c}>r_b$  one concludes
that equation \eqref{eq:ConsistentGeoKruskal2} is satisfied.
To verify that this curve corresponds also the geodesic with initial
data \eqref{eg:DataGeoKruskal} observe that using the chain rule and the
definition of the tortoise coordinate one has
\begin{equation}\label{Chainrule-T-data-Bifurcation-Sphere}
\frac{\mbox{d}T}{\mbox{d}\tilde{\tau}}=b T(r) \frac{\mbox{d}
  \newrbar[-1pt][-3.7pt]}{\mbox{d}r}\frac{\mbox{d}r}{\mbox{d}\tilde{\tau}}=
\frac{b T(r) }{D(r)}\frac{\mbox{d}r}{\mbox{d}\tilde{\tau}}.
\end{equation}
Since at $\tilde{\tau}=0$ one has $r=r_{\star}=r_{c}$ then
\[
T'_{\star}= \lim_{r\to r_{c}} \frac{bT(r)}{D(r)}\frac{\mbox{d}r}{\mbox{d}\tilde{\tau}}.
\]
To analyse this limit observe that
\[
T(r)= \exp (b \newrbar[-1pt][-3.7pt] )= \sqrt{ (r+
|r_{-}|)^{k_{n}}(r-r_{b})^{k_{b}}(r-r_{c})}.
\]
Using equations \eqref{ReducedEquation} and \eqref{eq:drm1} with
$\beta=0$ and $r_{\star}=r_{c}$ one has
\[
\frac{\mbox{d}r}{\mbox{d}\tilde{\tau}}=
\sqrt{\frac{(r-r_c)(r-\alpha_{-}(r_{c}))(r-\alpha_{+}(r_{c}))}{r}}
\]
where
\[
\alpha_{\pm}(r_{c}) = \frac{r_{c}}{2} \Bigg(-1 \pm \sqrt{1 + \bigg(
  \frac{2 r_{\circledast}}{r_{c}}\bigg)^3} \hspace{1mm} \Bigg),
\]
substituting equation \eqref{eq:DrExplicitRoots} one obtains
\[
\frac{T(r)}{D(r)}\frac{\mbox{d}r}{\mbox{d}\tilde{\tau}}=-\sqrt{r
  (r-\alpha_{-}(r_{c}))(r-\alpha_{+}(r_{c}))(r+
  |r_{-}|)^{k_{n}-2}(r-r_{b})^{k_{b}-2}}.
\]
Evaluating the last expression at $r=r_{c}$ and using equation
\eqref{Chainrule-T-data-Bifurcation-Sphere} one concludes that
 $T'_{\star} \neq 0$. Moreover one can verify that $T'_{\star}$
 is positive and finite. Finally,
notice that
\[
\lim_{r \to \infty}T(r)=1.
\]

\begin{remark}
{\em The  analysis in the previous paragraphs shows that conformal
  geodesics with $r_\star=r_c$ intersect the conformal boundary at
$(T,\Psi)=(1,0)$. In terms of the retarded Eddington-Finkelstein 
null coordinate
$u$ this corresponds to the condition $u=0$. Accordingly, these
conformal geodesics bisect the Cosmological region of the spacetime,
and in particular also the conformal boundary.}
\end{remark}

\subsubsection{Conformal geodesic through the bifurcation sphere of
  the black hole horizon}
\label{CoverConformalBoundarySubextremalBifurcationBlackhole}
The analysis of the previous Section can be adapted, \emph{mutatis
  mutandi}, to the case of conformal geodesics starting at $r_\star
=r_b$. In this case the function $b$ in the function $G(r)$ as given
by equation \eqref{KruskalG} is set to
\[
b \equiv \frac{1}{2r_{b}}(r_{c}-r_{b})(r_{c} + |r_{-}|),
\]
so that 
\[
\tilde{\bml}=\frac{1}{2b^2r}(r+
|r_{-}|)^{A_{n}}(r-r_{c})^{A_{c}}(\mathbf{d}T\otimes \mathbf{d}T-
\mathbf{d}\Psi\otimes \mathbf{d}\Psi).
\]
As $b>0$ and given that $\lim_{r\rightarrow
  r_b}\newrbar[-1pt][-3.7pt]=-\infty$, it follows that
\[
\lim_{r\rightarrow r_b} T =\lim_{r\rightarrow r_b} \Psi =0. 
\]
Using the same methods employed in the analysis of the geodesics with
$r_\star=r_c$, it can be readily shown that the curves with initial conditions given by
\[
T_\star =0, \quad \Psi_\star=0,\quad T'_\star \neq 0, \quad
\Psi'_\star =0,
\]
and such that 
\[
\Psi(\tilde{\tau}) =0, \qquad \tilde{\tau}\geq 0,
\]
are conformal geodesics. Moreover, it can be verified that
\[
\lim_{r\rightarrow 0} T = T_{\lightning}\equiv \sqrt{|r_-|^{k_n} r_b r_c^{k_c}} \neq 0.
\]

In what follows, we will denote by $(u_{\lightning},v_{\lightning})$
the value of the null coordinates $(u,v)$ associated to
$(T,\Psi)=(T_\lightning,0)$. Notice that, in particular $u_\lightning
= v_\lightning$.

\begin{remark}
{\em The analysis of the previous paragraphs shows that the black hole
region is bisected by timelike conformal geodesics emanating from the
bifurcation sphere at $r=r_b$.} 
\end{remark}

\subsection{Explicit expressions in terms of elliptic functions}
 The integral of equation \eqref{GeneralFormulaForTauTilde} can be
written in terms of special functions.  In particular, using
formulae 258.13 and 339.01 given in \cite{ByrFri13}, with parameters \[a=r_{\star}, \quad b=\alpha_{+}, \quad c=0, \quad d=\alpha_{-},\]
 one can
rewrite the integral \eqref{rGeodesicIntegralCriticalCurve} in terms of elliptic
functions. Using these formulae one concludes that the substitution
\[
\mbox{sn}^{2} \bar{w}= \bigg(\frac{\alpha_{+}-\alpha_{-}}{r_{\star}-\alpha_{-}}\bigg)\bigg(\frac{\bar{r}-r_{\star}}{\bar{r}-\alpha_{+}}\bigg),
\]
leads to
\begin{equation}\label{EllipticFunctionSubextremal}
\tilde{\tau}=\frac{2r_{\star}}{\alpha^2\sqrt{r_{\star}(\alpha_{+}-\alpha_{-})}}
\bigg( \kappa^2 w + (\alpha^2-\kappa^2)\Pi[\phi,\alpha^2,\kappa]\bigg),
\end{equation}
where $\mbox{sn}$ denotes the Jacobian elliptic function,
$\Pi[\phi,\alpha^2,\kappa]$ is the incomplete elliptic integral of the
third kind and
\begin{eqnarray*}
 \mbox{sn}^{2} w \equiv
 \bigg(\frac{\alpha_{+}-\alpha_{-}}{r_{\star}-\alpha_{-}}\bigg)
\bigg(\frac{r-r_{\star}}{r-\alpha_{+}}\bigg),
 \qquad \qquad \alpha^2 \equiv
 \frac{r_{\star}-\alpha_{-}}{\alpha_{+}-\alpha_{-}}, \\ \kappa^2
 \equiv
 \frac{\alpha_{+}(r_{\star}-\alpha_{-})}{r_{\star}(\alpha_{+}-\alpha_{-})},
 \qquad \qquad \phi \equiv \arcsin (\mbox{sn} w).
\end{eqnarray*}
In particular, it follows from the previous expressions and the
general theory of elliptic functions that  
 $\tilde{\tau}(r,r_{\star})$ as defined by the right hand
side of equation \eqref{EllipticFunctionSubextremal}, is an \emph{analytic
function of its arguments} ---see e.g. \cite{Law89}.  Notice that in
 Sections \ref{InitialDataSdS}-\ref{subsec:r<r*} we
have already analysed the behaviour of $\tilde{\tau}(r,r_{\star})$ 
 by constructing suitable 
estimates. Using equation \eqref{eq:value-of-bartau} one can write $\tau=2 \tanh
(\tilde{\tau}(r,r_{\star})/ 2)$ and integrate
\eqref{EqNullu} and \eqref{EqNullv}. From this
discussion it follows that $u(\tau,r_{\star})$ and $v(\tau,r_{\star})$
are \emph{analytic functions of their arguments}.  

\begin{remark}
\emph{
 Observe that due to the
 periodicity in the spatial direction of the maximal extension of the spacetime ---see  Figure \ref{Fig:SubSdS}, to analyse the behaviour of a congruence of conformal geodesics  considering  initial data with $r_b \leq r_\star \leq r_c$ is sufficient.
}
\end{remark}

%%%% 

\subsection{Analysis of the behaviour of the conformal deviation equations}

\label{AnalysisConfDeviationEquationsSubextremal}
 In Section \ref{Section:ConformalDeviationEquations} and
 \ref{sec:FormulaeWarpedProductSpaces}
 it was show how the requirement that the tangent vectors  of 
the congruence $\dot{\bmx}$ and the deviation vector for the congruence
$\bmz$ remain linearly independent can be studied
 by analysing a single scalar
quantity $\tilde{\omega}$. 
The purpose of this Section is to discuss the evolution of 
$\tilde{\omega}$ to analyse  the potential formation of
conjugate points in the congruences of conformal geodesics 
constructed in the previous Sections. 

\medskip
To start the discussion observe that 
in the subextremal Schwarzschild-de Sitter case one has $\beta=0$, consequently, equation \eqref{eq:gdeviationDiso} takes the simpler form
\begin{equation}
 \frac{\mbox{d}^2\tilde{\omega}}{\mbox{d}\tilde\tau^2}=\left(1+
 \frac{M}{ r^3}\right)\tilde{\omega}\;,\qquad
  r\equiv  r(\tilde\tau,r_\star).
 \label{eq:eq-ds}
 \end{equation}
We distinguish two possibilities according to whether
$r_\star=r_\circledast$ or $r_\star\neq r_\circledast$. 

\subsubsection{Conformal geodesics with $r_\star=r_\circledast$}
If $r_\star=r_\circledast$ then $ r=r_\circledast$ as shown in
Section \ref{subsec:circledast}. Hence, given that $r$ is a
constant, the differential equation \eqref{eq:eq-ds} with initial
conditions given by equation \eqref{eq:initial-omega} yields
\[
\tilde{\omega}(\tilde{\tau})=\frac{r_\circledast}{\rho_\circledast}\cosh
\left(\tilde\tau\sqrt{1+\frac{M}{r_\circledast^3}}\right)\;,\qquad r_\circledast=r(\rho_\circledast).
\]
 Using equations \eqref{eq:value-of-theta}-\eqref{eq:value-of-bartau}
and \eqref{eq:omega-rescaled} we obtain
\begin{equation}\label{OmegaCausticsCriticalCurve}
\tilde{\omega}(\tau)=r_\circledast\rho_\circledast\left(1-\frac{\tau^2}{4}\right)\cosh
\left(\sqrt{1+\frac{M}{r_\circledast^3}}\mbox{arctanh}\left(\frac{\tau}{2}\right)\right).
\end{equation}
This function is clearly non-zero for any value of
$\tau\in(-2,2)$ and thus no caustics develop along the curve defined
by $ r=r_\circledast$. We also deduce that
\[
\lim_{\tau\rightarrow\pm 2}\tilde{\omega}=\infty.
\]

\subsubsection{Conformal geodesics with $r_\star\neq r_\circledast$}
\label{sec:deviationEquationEstimates}
Observe that if $r_\star\neq
r_\circledast$, then either $r_{\star}<r_{\circledast}$ and the
conformal geodesics end at the singularity or
$r_{\star}>r_{\circledast}$ and the conformal geodesics reach the
conformal boundary $\mathscr{I}$. In any case one has that
\[
1+ \frac{M}{r}>1.
\]
The latter, together with equation \eqref{eq:gdeviationDiso} 
entails the differential inequality
\[
 \frac{\mbox{d}^2\omega}{\mbox{d}\tilde\tau^2}>\omega.
\]
It follows then that the scalars $\tilde{\omega}$ and $\omega\equiv
\Theta \tilde{\omega}$ satisfy
the inequalities
\[
 \omega\geq \varpi, \qquad \hat{\omega}\geq \Theta\varpi\;,
\]
where $\varpi$ is the solution of 
\[
 \frac{\mbox{d}^2\varpi}{\mbox{d}\tilde\tau^2}=\varpi\;,\qquad
\varpi(0,\rho_\star)=\frac{r_\star}{\rho_\star}\;,\qquad {\varpi}'(0,\rho_\star)=0.
\]
The solution to this last differential equation is given by
$\varpi=(r_\star/\rho_\star)\cosh\tilde\tau$. Using equations
\eqref{eq:value-of-theta}-\eqref{eq:value-of-bartau} we get the
inequality
\[
\omega\geq\left(1-\frac{\tau^2}{4}\right)\frac{r_\star}{\rho_\star}
\cosh\left(2\mbox{arctanh}\left(\frac{\tau}{2}\right)\right)=
\frac{r_\star}{\rho_\star}\left(1+\frac{\tau^2}{4}\right)>0.
\]
Consequently, one gets the limit
\[
\lim_{\tau\rightarrow\pm 2}{\omega}\geq\frac{2 r_\star}{\rho_\star}>0. 
\]
Hence, we conclude that the geodesics with $r_{\star}>r_{\circledast}$
which go to the conformal boundary $\mathscr{I}$ located at $\tau=2$
do not develop any caustic.  Now, for the case
$r_{\star}<r_{\circledast}$ in which the conformal geodesics fall into
the singularity it is sufficient to observe that

\[\tilde{\omega} \geq  \frac{r_{\star}}{\rho_{\star}}\mbox{cosh}(\tilde{\tau}).\]

Since the geodesics reach the singularity in a finite value of the
physical proper time $\tilde{\tau}$ and $\text{cosh}(\tilde{\tau})\geq2$
one conclude that the conformal geodesics falling into the singularity does not
form caustics as well.

\begin{figure}[t]
 \centering
\includegraphics[scale=0.8]{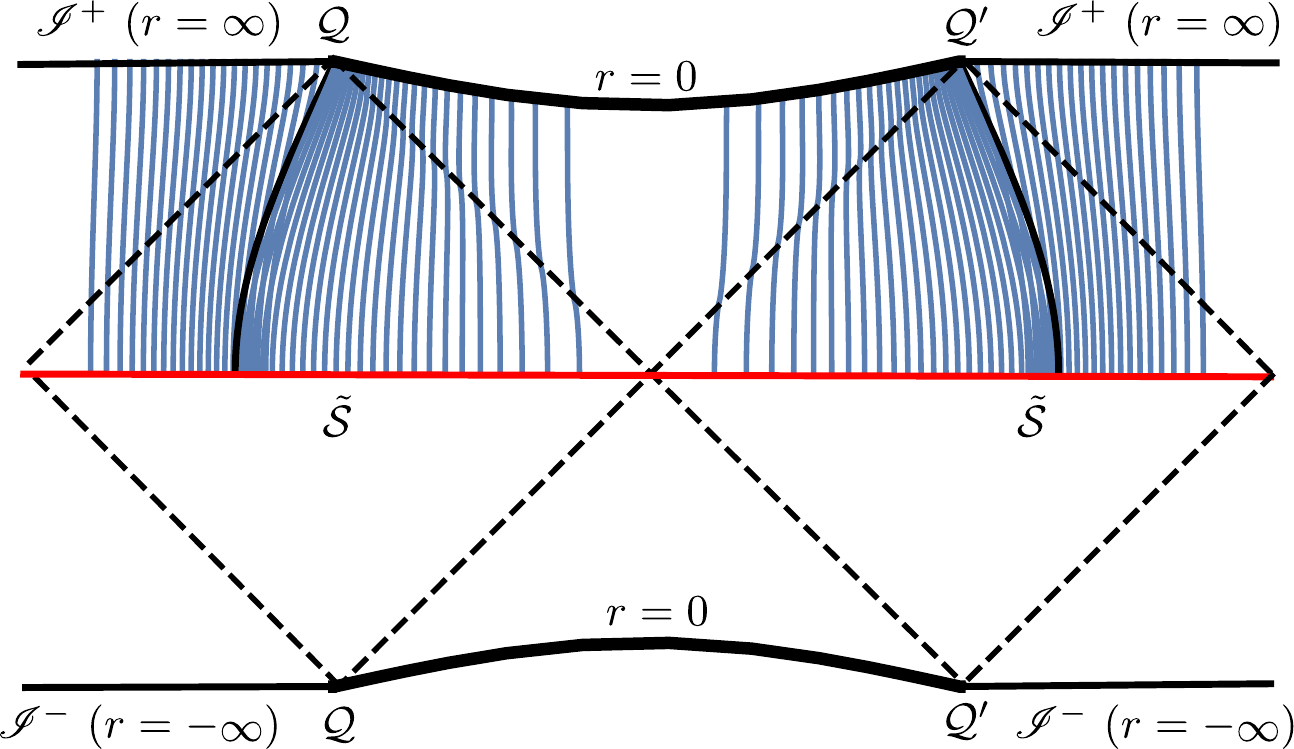}
 \caption{Numerical computation of the congruence
 of conformal geodesics given by Theorem \ref{Theorem:CGSdS} in the
 Penrose diagram of the Schwarzschild-de Sitter spacetime for $M=1/3$. The curves
 have been computed numerically using {\tt Mathematica}. The critical
 curves described by the condition $r=r_\circledast$ are depicted as
 solid black curves reaching the asymptotic points $\mathcal{Q}$ and $\mathcal{Q}'$. 
 Notice that the congruence is initially
 orthogonal to the initial hypersurface $\tilde{\mathcal{S}}$.}
\label{fig:cgSchdS}
 \end{figure}

\subsection{Conformal Gaussian coordinates in the subextremal
  Schwarzschild-de Sitter spacetime}
\label{Section:ConformalGaussianCoordinatesSdS}
In this Section we combine the results of the previous Sections to
show how the congruence of conformal geodesics defined by the initial
data in Section \ref{InitialDataSdS} covers the the whole maximal
extension of the subextremal Schwarzschild-de Sitter
spacetime. Accordingly, one can use this congruence to construct
global \emph{conformal Gaussian coordinates}. Because of the
periodicity in the spatial direction of the maximal extension of the
spacetime, it is only necessary to restrict the discussion to the
range $r_b\leq r_\star \leq r_c$.

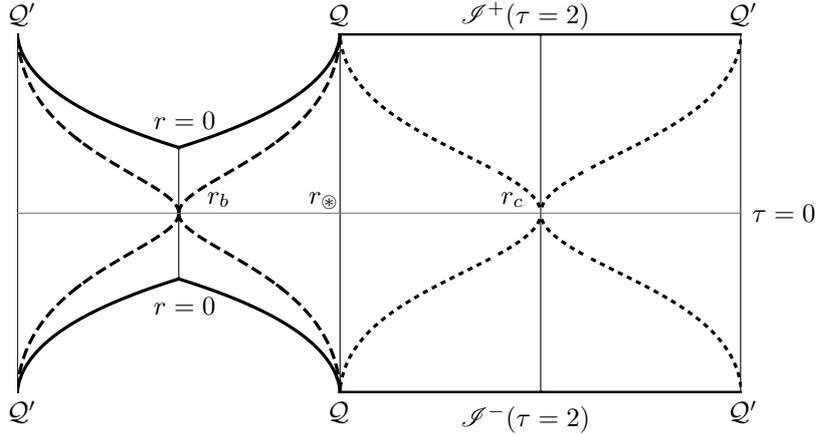
\begin{figure}[t]
\centering
\includegraphics[scale=0.5]{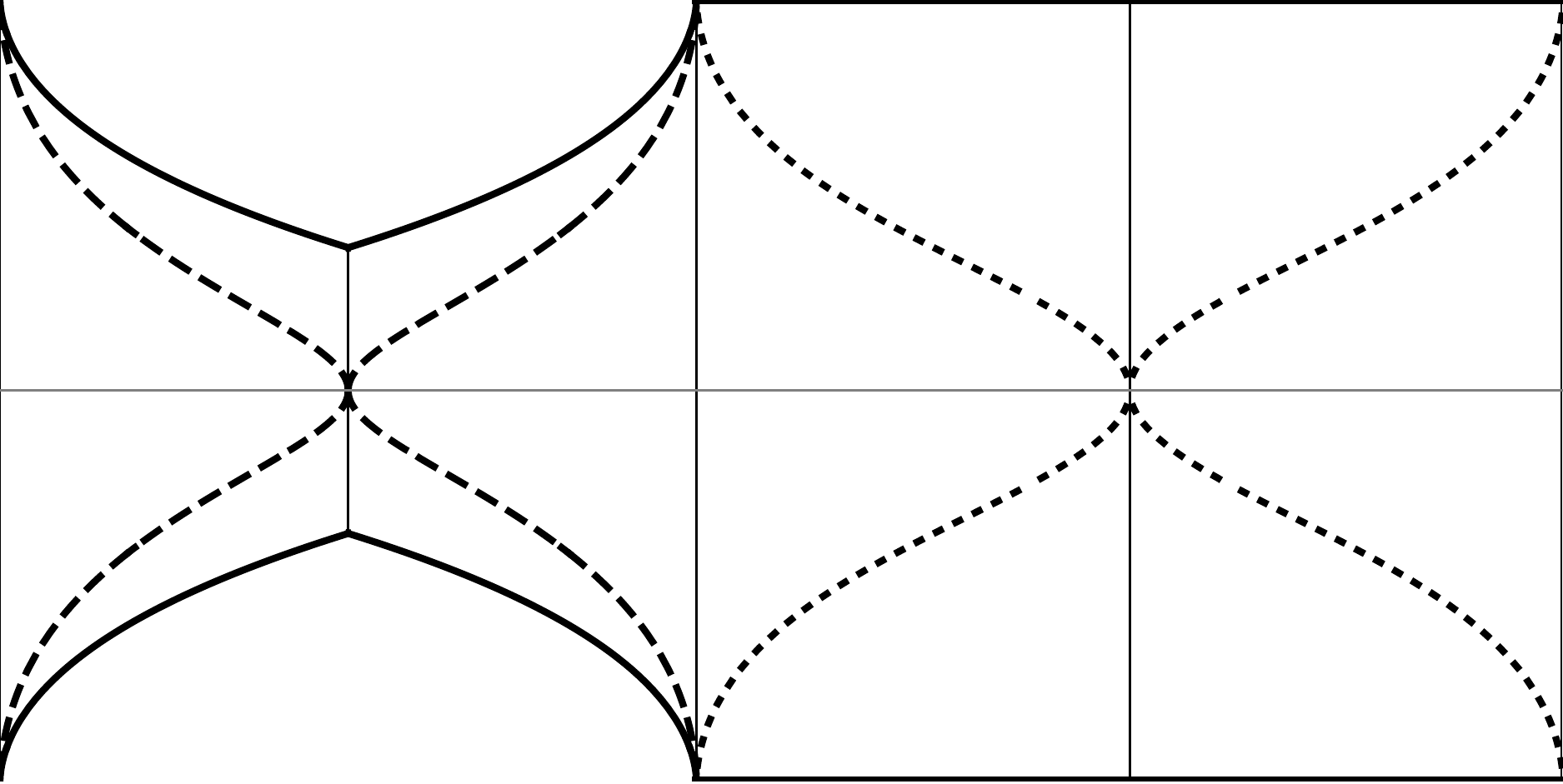}
\put(-5, -10){$\mathcal{Q'}$}
\put(-155, -10){$\mathcal{Q}$}
\put(-275, -10){$\mathcal{Q'}$}
\put(-5, 140){$\mathcal{Q'}$}
\put(-155, 140){$\mathcal{Q}$}
\put(-275, 140){$\mathcal{Q'}$}
\put(-162, 72){$r_{\circledast}$}
\put(-200, 72){$r_{b}$}
\put(-90, 72){$r_{c}$}
\put(-105, -12){$\mathscr{I}^{-}(\tau=2)$}
\put(-105, 140){$\mathscr{I}^{+}(\tau=2)$}
\put(4, 65){$\tau=0$}
\put(-220, 100){$r=0$}
\put(-220, 30){$r=0$}
\label{fig:GaussianSdS}
\caption{The Schwarzschild-de Sitter spacetime in conformal Gaussian
coordinates $(r_{\star},\tau)$ for $M=1/3$.  The horizontal line
represents the initial hypersurface $\tilde{\mathcal{S}}$ which is
parametrised by $r_{\star}$.  The vertical axis corresponds to the
unphysical proper time $\tau$. In particular, the location of the
black hole horizon in conformal Gaussian coordinates,
$(r_{b},\tau(r_{b},r_{\star}))$ with $r_{b}\leq r_{\star}\leq
r_{\circledast}$, corresponds to the dashed curve.  The thick
continuous line denotes the location of the singularity
$(0,\tau(0,r_{\star}))$ with $0\leq r_{\star}\leq
r_{\circledast}$. Similarly, the location of the cosmological horizon
in conformal Gaussian coordinates, $(r_{c},\tau(r_{c},r_{\star}))$
with $r_{\circledast}\leq r_{\star} \leq r_{c}$, is depicted in the
dotted curve.  The vertical lines correspond to the conformal
geodesics starting at the bifurcation spheres, $r_{\star}=r_{b}$ and
$r_{\star}=r_{c}$ respectively, and the critical conformal geodesic
$r_{\star}=r_{\circledast}$. The future and past conformal boundary
$\mathscr{I}^{\pm}$ are located at $\tau=\pm 2$ which are reached by
those conformal geodesics with initial datum $r_{\circledast}
<r_{\star}<r_{c}$.}
\end{figure}

\medskip
In what follows let denote by $\mbox{SdS}_I$ the region in the
conformal representation of the subextremal Schwarzschild-de Sitter
spacetime defined by the conformal factor $\Theta$ associated to the
congruence of conformal geodesics which is bounded by  the critical conformal
geodesic and the conformal geodesic passing by the bifurcation sphere
of the Cosmological horizon ---boundaries included. Similarly, we
define $\mbox{SdS}_{II}$ as the region bounded by the critical
geodesic and the geodesic passing through the bifurcation sphere of
the black hole horizon ---we exclude the former part of the boundary
and include the latter. 

\medskip
\noindent
\textbf{The region $\mbox{SdS}_I$.} For $r\in (0,\infty)$ and $u\in
(-\infty,\infty)$, let $z\equiv 1/r$ and $w\equiv \tanh u$. In terms
of these coordinates one has that
\[
\mbox{SdS}_I =\big\{ z\in [0,z_\circledast], \; w\in [-1,1]  \big\}.
\]
The analysis of the previous Sections shows that the map 
\[
(w,z): [0,2]\times [r_\circledast,r_c] \longrightarrow
[-1, 1]\times [0,z_\circledast]
\]
with
\[
w=\tanh u(\tau,r_\star), \qquad z=1/r(\tau,r_\star)
\]
as defined by the solutions to the conformal geodesics depends
analytically on its parameters. The analysis of the conformal geodesic
deviation equation implies that the Jacobian of the transformation is
non-zero for the given value of the parameters ---it can be readily
verified that the function
$\Theta\tilde{\omega}$ coincides with the Jacobian. Thus, it follows that
the inverse map
\[
(\tau,r_\star): [-1,1] \times
[0,z_\circledast]\longrightarrow [0,2]\times [r_\circledast,r_c]
\]
with
\[
\tau =\tau( \mbox{arctanh}\; w, 1/z), \qquad r_\star =r_\star(\mbox{arctanh}\; w, 1/z), 
\]
and $u\in\mathbb{R}$ is well defined and also an analytic function of its parameters. Thus,
ignoring angular coordinates, this inverse map defines a \emph{conformal Gaussian
  system of coordinates} in $\mbox{SdS}_I$. In particular, given any point in $\mbox{SdS}_I$,
there is a unique conformal geodesic passing through it. Thus, \emph{the
congruence of conformal geodesics covers the whole of $\mbox{SdS}_I$}. 

\medskip
\noindent
\textbf{The region $\mbox{SdS}_{II}$.} Let $y=\tanh v$. In terms of the coordinates
$(y,r)$ the region $\mbox{SdS}_{II}$ is described as
\[
\mbox{SdS}_{II} =\big\{ r\in(0,r_\circledast), \; y\in [-1,1] \big\}.
\]
As with the region $\mbox{SdS}_I$, the analysis from the previous
Sections shows that the map
\[
(y,r):[0,\tau_\lightning)\times [r_b,r_\circledast] \longrightarrow
[-1,1]\times (0,r_\circledast]
\]
is well defined and an analytic function of its parameters. The analysis of the conformal geodesic
deviation equation implies that the Jacobian of the transformation,
given by the function $\tilde{\omega}$, is
non-zero for the given value of the parameters. Thus, it follows that
the inverse map
\[
(\tau,r_\star): [-1,1]\times (0,r_\circledast]\longrightarrow  [0,\tau_\lightning)\times [r_b,r_\circledast]
\]
is well defined and also an analytic function of its parameters. Thus,
ignoring angular coordinates, this inverse map defines a \emph{conformal Gaussian
  system of coordinates} in $\mbox{SdS}_{II}$. In particular, \emph{the
congruence of conformal geodesics covers the whole of $\mbox{SdS}_{II}$}. 

\subsection{Summary of the analysis}\label{SummarySubExtremal}

The analysis of Section \ref{subsec:non-extremal} can be summarised in
the following proposition.

\begin{theorem}[\textbf{\em Conformal geodesics in the subextremal
 Schwarzschild-de Sitter spacetime}]
\label{Theorem:CGSdS}
The maximal extension of the subextremal Schwarzschild-de Sitter
spacetime can be covered by a non-singular congruence of conformal
geodesics. The existence of this congruence implies the existence of a
global system of conformal Gaussian coordinates in the spacetime.
\end{theorem}

A qualitatively accurate depiction of the congruence of conformal geodesics in the Penrose diagram of the Schwarzschild-de Sitter spacetime can be seen in Figure \ref{fig:cgSchdS}.

\section{Analysis of the conformal geodesics in the extremal \\
Schwarzschild-de Sitter case}
\label{eSdS}
In this Section, for concreteness we restrict our discussion to the
representation of the extremal Schwarzschild-de Sitter spacetime as a
whitehole so that there the singularity lies in the past and the
conformal boundary in the future ---see Figure \ref{Fig:eSdS} (a).

\medskip
Throughout this Section it is assumed that $M=2/(3\sqrt{3})$. In
this case, $D(r)$ vanishes at $r_{\mathcal{H}} \equiv
1/\sqrt{3}$ while $D(r)<0$ for $r_{\mathcal{H}} \in
\mathbb{R}^{+}/\{1/\sqrt{3}\}$. Hence for
$r\in \mathbb{R}^{+}/\{1/\sqrt{3}\}$, the coordinate
$r$ is a time coordinate while $t$ is a spatial one.
This means that the hypersurfaces $r=r_{\star}$ with
$r_{\star}\neq r_{\mathcal{H}}$ are spacelike while
those of constant $t$ are timelike ---see Figure \ref{fig:SdSConstantCurves}.
 Consequently, one purports
to construct a congruence of conformal geodesics by prescribing
initial data on hypersurfaces of constant
$r=r_{\star}$ with $r_{\star} \neq 0$ and
$r_{\star} \neq r_{\mathcal{H}}$. In addition, observe that
the critical curve located at $r_{\circledast}$ as defined in Section
\ref{subsec:circledast} coincides in this case
with the location of the Killing horizon $r_{\mathcal{H}}$.

\subsection{Basic setup} 
\label{InitialConditionsExt}
Similar to the subextremal case we set 
\begin{equation}\label{dotTheta-1}
 \Theta_\star=1.
\end{equation}
Hence \(\tilde{\bmbeta}_{*}= \Theta^{-1}_{\star}\mathbf{d}\Theta_{\star}
= 0 \) and we get $\beta_\star=0$ . This means that $\beta=0$ since
$\beta$ is a constant along a given conformal geodesic.  As in the
case of Section \ref{subsec:non-extremal}, this choice simplifies
the equations for the curve as these reduce to the equations of metric
geodesics. In the remainder of this Section let
\[
\tilde{\mathcal{S}}\equiv \{ r = r_\star\}\;,
\]
with $r_\star$ a constant $r_{\star}\neq 0$ and
$r_{\star} \neq 1/\sqrt{3}$ ---see Figure
\ref{Fig:SdSConstantCurves}.  As already pointed out, the above
hypersurface of the extremal Schwarzschild-de Sitter spacetime 
is spacelike. Particularising
equations \eqref{ReducedEquation} and \eqref{PhysicalNormalisation}
for the case $\beta=0$ we get
\begin{equation} 
r'_\star = \sqrt{\gamma^2 - D_\star}\;,\qquad t'_\star
= \frac{|\gamma|}{|D_\star|}.
\label{eqr-initial}
\end{equation} 

\begin{remark} {\em Notice that in the extremal case $D_{\star}<0$ and
in contrast with the initial data chosen for the subextremal case the
congruence is not necessarily orthogonal to the initial hypersurface
$\tilde{\mathcal{S}}$. To see this more clearly, observe that the unit
normal to $\tilde{\mathcal{S}}$ is given by
$\bmn=(1/\sqrt{|D_{\star}|})\mathbf{d}r$ while
$\tilde{\bmx}'_{\star}=r'_{\star}\bm\partial_{r}+t'_{\star}\bm\partial_{t}$.
In particular, notice that $\bmn^{\sharp}$ and $\tilde{\bmx}'_{\star}$
are not parallel and that $ \langle \bmn, \tilde{\bmx}'_{\star}
\rangle =\sqrt{\gamma^2/|D_{\star}|+1}$.  Compare this with the
discussion for the subextremal Schwarzschild-de Sitter spacetime given
in Section \ref{subsec:non-extremal} where $D_{\star}>0$ and the
symmetry of the spacetime suggested to consider the initial
hypersurface determined by the condition $t=0$ which has unit normal
$\bmn=\sqrt{D_{\star}}\mathbf{d}t$. In this case, setting
$\gamma=\sqrt{D_{\star}}$ leads to the initial data
\eqref{CGInitialData} for which $\tilde{\bmx}'_{\star} $ and
$\bmn^{\sharp}$ coincide ---see Figure \ref{fig:SdSConstantCurves}.
In the extremal case however there is not a apriori preferred choice
for $r_{\star} \neq r_{\mathcal{H}}$ determining the initial
hypersurface ---see Figure \ref{Fig:SdSConstantCurves}. Consequently,
instead of prescribing data corresponding to a congruence starting
orthogonal to $\tilde{\mathcal{S}}$ we will consider $\gamma$ as a
free parameter so that, in general, the congruence can start oblique
to the initial hypersurface.}
\end{remark}

\begin{remark}
{\em Observe that for $r_\star<r_{\mathcal{H}}$ the hypersurface
  $\tilde{\mathcal{S}}$ is below the horizon while if $r_\star>r_{\mathcal{H}}$
  then it lies above it. Our analysis will consider both situations simultaneously.}
\end{remark}

\subsubsection{Conformal factor}
\label{ConformalFactor-Extremal}
\medskip
Since $\beta=0$ one has that the 1-form $\tilde{\bmbeta}$ vanishes.
Using equation \eqref{eq:geodesic-rescaling} and
\eqref{AdaptedParameters-ConformalFactor} one observes that
$\bmbeta=\dot{\Theta}\dot{\bmx}^{\flat}$. Taking into account equation
\eqref{dotTheta-1} along with the constraints \eqref{Constraints} one
gets

\[\ddot{\Theta}_{\star}= \frac{1}{2}\dot{\Theta}^2_{\star}-\frac{1}{6}\lambda \]

\noindent For conciseness and consistency with the notation of Section
\ref{Section:SdS-SadS} one sets $\lambda =-3\epsilon$. For the case of
a de-Sitter like cosmological constant one sets $\epsilon=-1$.
Moreover, without loss of generality let us set $\tau_{\star}=0$, then
the conformal factor associated to the congruence defined by the
initial data discussed in Section \ref{InitialConditionsExt} is found
by using the general expression \eqref{ConformalFactor}

\[
\Theta(\tau)=1 + \dot{\Theta}_\star \tau +
\frac{1}{4}\left(\dot{\Theta}_\star ^2-1 \right )\tau^2,
\]
where $\dot{\Theta}_\star =\langle {\bm\beta}_\star ,\dot{\bmx}_\star
\rangle$. The roots of $\Theta(\tau)$ are given by
\begin{equation}\label{rootsCFExtreme}
\tau_{+}= \frac{2}{1-\dot{\Theta}_{\star}} \qquad
\tau_{-}=-\frac{2}{1+\dot{\Theta}_{\star}} .
\end{equation}
Although one has formally two roots of the conformal factor $\Theta$,
an inspection of the conformal diagram of Figure \ref{Fig:SdSConstantCurves}
 reveals that only one
component of the conformal boundary is actually realised. Moreover, in
order to have $ \tau_{\pm}$ finite, we restrict the possible values of
$\dot{\Theta}_\star$ to those that satisfy $|\dot{\Theta}_\star| \neq
1 $. This corresponds to restrict our analysis to compact conformal 
 representations in which $\mathscr{I}$ is located at 
a finite value of the unphysical proper
time $\tau$. 

\medskip
The relation between the unphysical proper time $\tau$ and
the physical proper time $\tilde{\tau}$ is readily obtained form
equation \eqref{Reparametrisation}:
\[
\tilde{\tau} = 2 \mbox{arctanh} \Big(
\frac{1}{2}(1-\dot{\Theta}^2_{\star})\tau-\dot{\Theta}_{\star}\Big)
  + 2 \mbox{arctanh} {\dot{\Theta}_{\star}}.
\]
Therefore,
\begin{equation}\label{TauAsAFunctionOfTauTilde}
\tau=\frac{2}{1-\dot{\Theta}^2_{\star}}\tanh \left( 
\tilde{\tau}/2 \right) +
\frac{2\dot{\Theta}_{\star}}{1-\dot{\Theta}_{\star}^2}.
\end{equation}
From these expressions we deduce that
$\lim_{\tilde{\tau}\rightarrow\pm\infty}\tau=\tau_{\pm}$. For
conciseness let us take $0<\dot{\Theta}_{\star}<1$ so that $\tau_{+}>0$
and $\tau_{-}<\tau_{+}$. Using equation
\eqref{TauAsAFunctionOfTauTilde} we see that the conformal factor can
be rewritten in terms of the $\tilde{\bmg}$-proper time $\tilde{\tau}$ as

\begin{equation}\label{ConformalFactor-InTermsOfTauTilde}
\Theta(\tilde{\tau})=\frac{1}{1 -\dot{\Theta}^2_{\star}} \mbox{sech} ^2  
( \tilde{\tau}/2).
\end{equation}

\begin{figure}[t]
\centering
\includegraphics[width=0.5\textwidth]{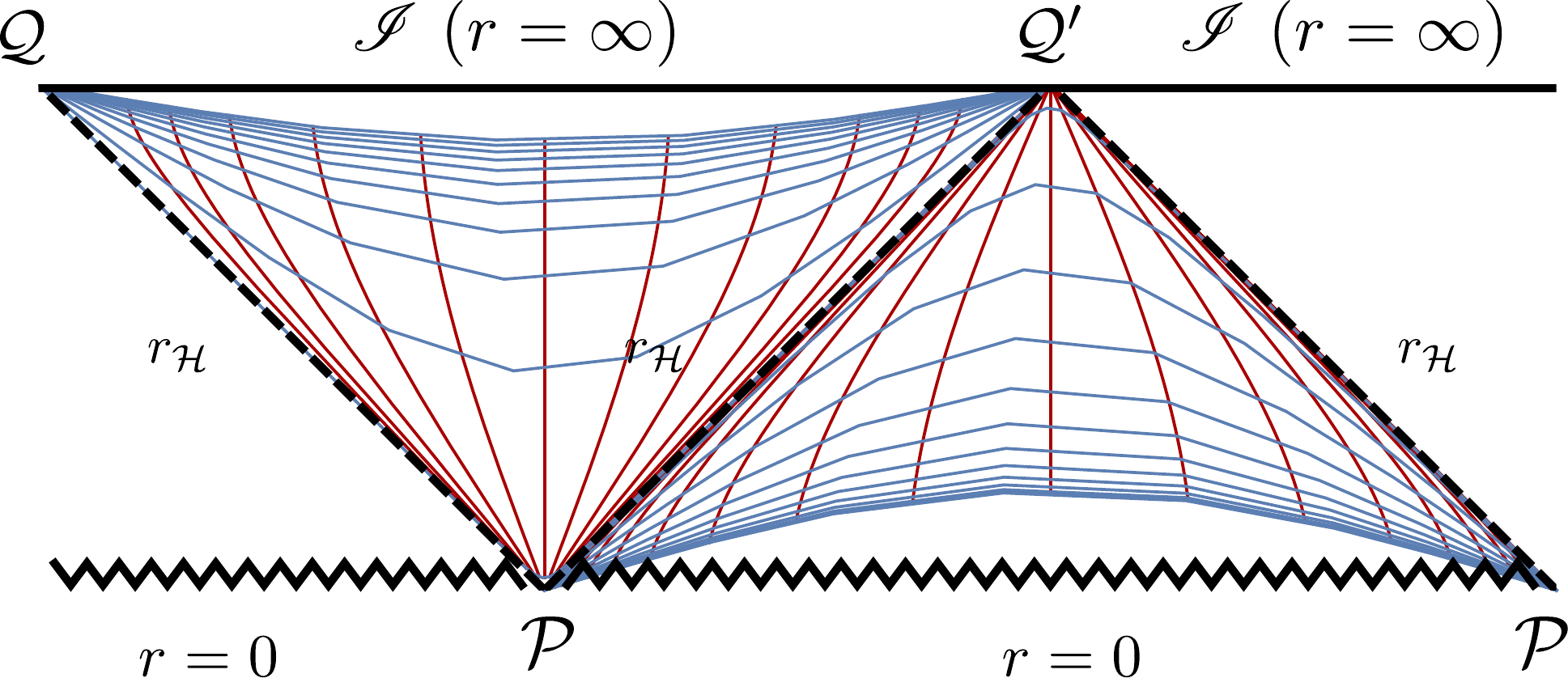}
\caption{ Curves with constant $t$
 and $r$ (red and blue respectively)
  are plotted on the Penrose diagram of the extremal Schwarzschild-de
  Sitter spacetime. In contrast with the subextremal case curves with
  constant $t$ in starting from some
  $r_{\star}<r_{\mathcal{H}}$ accumulate at the
  asymptotic points $\mathcal{Q}$ and $\mathcal{Q'}$ while those
  starting from $r_{\star}>r_{\mathcal{H}}$ accumulate
  at $\mathcal{P}$.}
\label{Fig:SdSConstantCurves}
\end{figure}

\subsubsection{Initial conditions for the congruence}

Unlike in the other cases studied in this article, we do not
require the congruence to be everywhere orthogonal to the initial
hypersurface $\tilde{\mathcal{S}}$. However, we are interested in
showing that the curves arrive orthogonally to the conformal boundary.
To do so we first compute
\begin{equation}
\label{PerpToScri}
\mathbf{d}\Theta = \left(\dot{\Theta}_{\star}+
\frac{1}{2}\left(\dot{\Theta}^2_{\star}-1\right)\tau\right)\mathbf{d}\tau.
\end{equation}
 Observe that $\mathbf{d}\Theta = \bm0$ at $ \tau=
2\dot{\Theta}_{\star}/(1-\dot{\Theta}^2_{\star})$. Hence, 
$\mathbf{d}\Theta \neq \bm0$ at $\mathscr{I}$.  Now from
\eqref{Reparametrisation} on has that
 \begin{equation}
\label{DiferentialTauAndTauTilde}
 \mathbf{d}\tau=\Theta \mathbf{d}\tilde{\tau}.
\end{equation}
Using equation \eqref{ReducedEquation} one can consider
$r=r(\tilde{\tau})$, equivalently
$\tilde{\tau}=\tilde{\tau}(r)$, therefore
$\mathbf{d}\tilde{\tau}=(1/r')\mathbf{d}r$. Using the last
observation and equations \eqref{PerpToScri} and
\eqref{DiferentialTauAndTauTilde} and one gets
\begin{equation}
\label{NormalToScri} 
\mathbf{d}\Theta = \Theta \left(\dot{\Theta}_{\star}+ \frac{1}{2}
\left(\dot{\Theta}^2_{\star}-1\right)\tau\right)\frac{1}{r'}\mathbf{d}r.
\end{equation}
Recall now that, as described in Section
\ref{sec:FormulaeWarpedProductSpaces}, we are considering geodesic
equations with no evolution in the angular coordinates, thus
\[
 \tilde{\bmx}'=t' \bm\partial_{t} + r'{\bm\partial}_{r}, 
\]
since one is effectively analysing the curves on a
2-dimensional manifold $\mathcal{\tilde{M}}/SO(3)$ with metric
$\tilde{\bml}$ as given in equation \eqref{eq:2-l}. Therefore, one can
define a unique orthogonal vector to $\tilde{\bmx}'$ by
$\tilde{\bmx}'^{\perp \flat}\equiv
\tilde{\bmepsilon}_{\tilde{\bml}}(\tilde{\bmx}',\,\cdot\,)$.
 A direct computation shows that 
\[ 
\tilde{\bmx}'^{\perp \flat}=-r' \mathbf{d}t +
 t'\mathbf{d}r,
\]
so that,
\[
 \tilde{\bmx}'^{\perp} = -\frac{r'}{D(r)}
 \bm\partial_{t} -t'D(r)
 \bm\partial_{r}.
\]
Hence, using \eqref{eq:geodesic-rescaling} one has that
\begin{equation} 
\label{NormalToCurve} 
\dot{\bmx}^{\perp} =
 -\frac{1}{\Theta}\left(\frac{r'}{D(r)} \bm\partial_{t} +
 t'D(r) \bm\partial_r \right) .
\end{equation}
Therefore using equations \eqref{NormalToScri},
\eqref{NormalToCurve}, \eqref{ReducedEquation} and \eqref{Equation-t'}
we get
\[ 
\langle \mathbf{d}\Theta, \dot{\bmx}^{\perp}\rangle 
= \frac{|\gamma|}{\sqrt{\gamma^2-D(r)}} 
\left( \dot{\Theta}_{\star}+
\frac{1}{2}\left(\dot{\Theta}^2_{\star}-1\right)\tau \right).
\]
 Taking into account that $ \lim_{\tilde{r}\rightarrow
  \infty} D(r)=-\infty$ and that
 \[
\dot{\Theta}_{\star}+ \frac{1}{2}
\left(\dot{\Theta}^2_{\star}-1\right)\tau_{\pm} =\mp 1,
\] 
and then assuming $\gamma \neq 0$ one concludes 
\[ 
\langle \mathbf{d}\Theta, \dot{\bmx}^{\perp}\rangle =0
 \quad\text{at}\quad \mathscr{I}. 
\]
This means that the curves arrive orthogonally to the
conformal boundary.
\begin{remark}
{\em A quick inspection of the last argument shows that
one can perform a similar calculation for the congruence
of conformal geodesics in the subextremal Schwarzschild-de Sitter spacetime
of Section \ref{subsec:non-extremal}
and obtain the same conclusion.
}
\end{remark}
\begin{remark}
{\em  This is in agreement 
with Proposition 3.1 given in \cite{FriSch87} where it is shown
that in an Einstein spacetime, any conformal geodesic leaving $\mathscr{I}$ orthogonally is
up to reparametrisation a geodesic of $\tilde{\bmg}$.}
\end{remark}

\subsection{Qualitative analysis of the behaviour of the curves}

In this Section we carry out a
qualitative study of the family of curves in the extremal
Schwarzschild-de Sitter analogous to that of Section
\ref{Section:QualitativeSdS} for the subextremal case.  
We distinguish three basic types of curves: those
escaping to the asymptotic points and those emanating from the
singularity in the past and escaping to the conformal boundary in the future.

\subsubsection{Conformal geodesics with constant $t$}
\label{subsec:circledastExtremal}

The condition determining the location of the curves with constant $r$
is given by
\[
\gamma^2-D(r)=0.
\]
Since $D(r)\leq 0$ and given that $D(r)$ only vanishes for positive $r$
at $r=r_{\mathcal{H}}$ then, the critical curve is
characterised by the conditions $\gamma=0$ and
$r=r_{\mathcal{H}}$.  This is consistent with the
analysis of Section \ref{subsec:circledast} since
$r_{\circledast}$ as given in equation
\eqref{CriticalCurveSubextremal} reduces to
$r_{\mathcal{H}}=1/\sqrt{3}$ for $M=2/3\sqrt{3}$. 
Moreover,
notice that by virtue of equation \eqref{Equation-t'} and the initial
data given in Section \ref{InitialConditionsExt} conformal geodesics
with $\gamma=0$ and $r\neq r_{\mathcal{H}}$ coincide
with curves of constant $t$.  These curves accumulate at the
asymptotic points $\mathcal{Q}$ and $\mathcal{Q'}$ ---see Figure
\ref{Fig:SdSConstantCurves}. Observe that for $\gamma=0$ and $\beta=0$ the expression
\eqref{GeneralFormulaForTauTilde} can be explicitly integrated to
yield
\begin{equation}
\label{eSdSgeodesicNotCrossingHorizon}
  \tilde{\tau}={r}_{\mathcal{H}} \ln
  \left(H(r)/H(r_{\star})\right)
\end{equation}
where 
\[
H(r)\equiv \left( \frac{\sqrt{3r}+
  \sqrt{r+2r_{\mathcal{H}}}}{(\sqrt{3r}-
  \sqrt{r+2r_{\mathcal{H}}})(\sqrt{r}+
  \sqrt{r+2r_\mathcal{H}})^{2\sqrt{3}}}\right).
\]
Observe that equation \eqref{eSdSgeodesicNotCrossingHorizon}, as
pointed out in \cite{Pod99}, implies that the geodesics with
$\gamma=0$ never cross the horizon since $\tilde{\tau}\rightarrow
\infty$ as $r\rightarrow r_{\mathcal{H}}$.
Using equation \eqref{TauAsAFunctionOfTauTilde} and setting
$\dot{\Theta}_{\star}=0$ for simplicity, we obtain
\begin{equation}
\label{ExplicitIntegrationUnphysicalProperTime}
\tau(r)=2\frac{W(r)-W_{\star}}{W(r)+W_{\star}},
\end{equation}
where $W(r)=H(r)^{r_{\mathcal{H}}}$.

\medskip
\noindent
\textbf{Asymptotic behaviour of the curve.} To analyse the behaviour of these curves as they approach the
asymptotic points $\mathcal{Q}$ and $\mathcal{Q'}$ let us consider
$r=r_{\mathcal{H}}+\epsilon$. Then one obtains that
for small $\epsilon>0$ that
\[
W(r)=\left(\frac{C_{1}}{\epsilon}\right)^{r_{\mathcal{H}}}(C_{2}+C_{3}\epsilon
+ \mathcal{O}(\epsilon^2))
\]
where $C_{1},C_{2},C_{3}$ are numerical factor not relevant for the
subsequent discussion. To leading order $W(r)=C/\epsilon^{p}$ where we
have used the value of $r_{\mathcal{H}}$ and introduced
$p=1/\sqrt{3}$ and $C=C_{1}^{r_{\mathcal{H}}}C_{2}$ to
simplify the notation. Consequently, to leading order we have
\[
\frac{\mbox{d}\tau}{\mbox{d}\epsilon}=\frac{-4W_{\star}Cp\epsilon^{p-1}}{(C +
  W_{\star}\epsilon^p)^2}.
\]
Since $p<1$ then one concludes that $\mbox{d}\tau/\mbox{d}\epsilon$ diverges as
$\epsilon \rightarrow 0$. Therefore the curves with $\gamma=0$ become
tangent to the Killing horizon as the approach the asymptotic points
$\mathcal{Q}$ and $\mathcal{Q'}$.  In other words, as in the
subextremal case, these curves become asymptotically null.

\subsubsection{Conformal geodesics with non-constant $r$}
Recalling that 
\[
r' =\sqrt{\gamma^2-D(r)}
\]
and observing that $D(r)\leq 0$, it follows that if $r'_\star\neq 0$
then, in fact, $r'>0$. Moreover, one can show that $r''_\star >0$ and
that $r''\neq 0$ for $r\in[r_\star,\infty)$. Thus, the curves escape
to the conformal boundary.

\medskip
\noindent
\textbf{Behaviour towards the conformal boundary.} We now show that the congruence of conformal
geodesics reaches the conformal boundary $\mathscr{I}$ in an infinite
amount of physical proper time. In order to see this, first observe from
equation \eqref{eq:define-D-extremal} that $D(r)\leq 0$,
consequently from equation \eqref{ReducedEquation} it follows that
$r(\tilde{\tau})$ is a monotonic function.  Moreover, using equations
\eqref{eq:define-D-extremal} and \eqref{Equation-t'} we find that
 \begin{equation}
\tilde{\tau}= \int^{
    r}_{r_\star} \frac{\mbox{d}\bar{r}}{\sqrt{\gamma^2 +
    \displaystyle
    \frac{1}{\bar{r}}\left(\displaystyle\frac{2}{\sqrt{3}}+\bar{r}\right)
    \left(\bar{r}-\displaystyle \frac{1}{\sqrt{3}}\right)^2}}.
\label{eq:integral-extremal}
\end{equation}
 We are interested in analysing the convergence of
\[
\tilde{\tau}_\infty \equiv \int^{\infty} _{r_\star}
\frac{1}{\sqrt{\gamma^2 +
    \displaystyle\frac{1}{\bar{r}}
\left(\displaystyle\frac{2}{\sqrt{3}}+\bar{r}\right)
    \left(\bar{r}-\displaystyle\frac{1}{\sqrt{3}}\right)^2}}
\mbox{d}\bar{r}. 
\]
Introducing a new variable $\xi\equiv r-1/\sqrt{3}$ then we can
rewrite the integral as
\[
\tilde{\tau}_\infty \equiv \int^{\infty} _{\xi_{\star}}
\frac{1}{\sqrt{\gamma^2 +
    \displaystyle\frac{\bar{\xi}^2(\bar{\xi}+\sqrt{3})}
{\bar{\xi}+\displaystyle\frac{1}{\sqrt{3}}}}}
\mbox{d}\bar{\xi}.
\] 
Since $r\geq r_{\star}>0$ one has $r\geq \delta_{\star}$ for some
$\delta_{\star}>0$ small. Thus $\xi \geq -1/\sqrt{3}
+\delta_{\star}$. Therefore we have that
\[
\frac{\xi + \sqrt{3}}{\xi + \displaystyle{ \frac{1}{\sqrt{3}}}} \leq
\kappa^2 \qquad \text{with} \qquad \kappa^2 \equiv 1+
  \frac{2}{\sqrt{3}\delta_{\star}}.
\]
Using the latter, we have that
\[
\sqrt{\gamma^2 +
  \frac{{\xi}^2(\xi+\sqrt{3})}{{\xi}
+\displaystyle\frac{1}{\sqrt{3}}}}
\leq\sqrt{\gamma^2+\kappa^2{\xi}}.
\]
so that
\[
\tilde{\tau}_\infty \geq \int^{\infty} _{\xi_\star} \frac{1}{
  \sqrt{\gamma^2 + \kappa^2 \bar{\xi}^2}} \mbox{d}\bar{\xi}.
\]
This last integral diverges and one concludes that
$\tilde{\tau}_\infty$  diverges as well. Accordingly, the conformal
boundary is reached in an infinite amount of physical proper time. 

\begin{remark}
{\em For initial hypersurfaces with $r_\star <r_{\mathcal{H}}$
  (i.e. lying below the horizon), a direct inspection of the integral in
\[
\tilde{\tau}_{\mathcal{H}} = \int^{
    r_{\mathcal{H}}}_{r_\star} \frac{\mbox{d}s}{\sqrt{\gamma^2 +
    \displaystyle
    \frac{1}{s}\left(\displaystyle\frac{2}{\sqrt{3}}+s\right)
    \left(s-\displaystyle \frac{1}{\sqrt{3}}\right)^2}},
\]
shows that it is finite ---there are no zeros in the denominator for
$s\in [r_\star,r_{\mathcal{H}}]$. Thus the horizon is reached in a
finite amount of proper time. 
}
\end{remark}

\medskip
\noindent
\textbf{Behaviour towards the singularity.}  We now shown that both the horizon and the
singularity are reached by the curves of the congruence in a finite
amount of physical proper time $\tilde{\tau}$. This behaviour is a
consequence of the fact that the integral in equation
\eqref{eq:integral-extremal} is finite for any $r$ with $0\leq
r\leq r_\star$. To prove this we start from the inequality
\[
 \sqrt{ \gamma^2 + \frac{1}{r}\left(\displaystyle\frac{2}{\sqrt{3}}+r\right)
\left(r-\displaystyle \frac{1}{\sqrt{3}}\right)^2} \geq |\gamma|.
\]
Therefore, for $\gamma \neq 0$ this entails
\[
 \tilde{\tau} = \int^{r}_{r_\star}
 \frac{\mbox{d}\bar{r}}{\sqrt{\gamma^2 + \displaystyle
     \frac{1}{\bar{r}}\left(\displaystyle\frac{2}{\sqrt{3}}+\bar{r}\right)
     \left(\bar{r}-\displaystyle \frac{1}{\sqrt{3}}\right)^2}}\leq
 \int^{r}_{r_\star}\frac{\mbox{d}\bar{r}}{|\gamma|}=
 \frac{r-r_\star}{|\gamma|}\;.
\]
This last inequality shows the assertion. In what follows we denote by
$\tilde{\tau}_\lightning$ (respectively $\tau_\lightning$) the value
of the proper time at which the singularity is reached.

\subsubsection{Behaviour of the congruence using null coordinates}
\label{AnalysisUsingNullCoords}

Most of the analysis performed for the asymptotic region
$r>r_{c}$ in the subextremal case 
in Section \ref{subsec:r>r*} 
can also be applied for the extremal case for any 
$r>r_{\star}>0$  since in this case
one has $D(r)\leq 0$ then the expression
\[
\sqrt{\gamma^2-D(r)} \geq |\gamma|
\]
is valid for any $r$. This 
leads to formally the same estimates as in Section 
\ref{subsec:r>r*}.
Notice however that the congruence does not
start orthogonally to the initial hypersurface 
$\tilde{\mathcal{S}}$. Assuming
 $\gamma \neq 0$  and performing formally
 the same procedure leading
 to equation \eqref{FirstEstimateSub}  we get
\begin{equation}
\label{NullCoordsEstimate1Extremal}
|u(r) - u_{\star}|  \leq 
\frac{r-r_{\star}}{\gamma^2}.
\end{equation}
This estimate is valid for any $r\geq r_{\star}>0$.
 In particular observe
 that $|u(r_{\mathcal{H}})-u_{\star}|$ is finite.
Now, let us denote as before $u_{\infty} \equiv
\lim_{r\rightarrow \infty}u$ and take
$r_{\mathcal{H}}<r_{\bullet}<\infty$, 
then the same procedure
leading to \eqref{estimateIntegral2Sub} renders
\begin{eqnarray}
\label{NullCoordsEstimate2Extremal}
|u_{\infty} - u({r_{\bullet}})| 
 \leq  \int ^{\infty} _{r_\bullet}  \frac{2}{|D(\bar{r})|}  \mbox{d}\bar{r}.
\end{eqnarray}
In contrast to Section \ref{subsec:r>r*}, at this point one can 
 compute the last integral using the explicit functional form for $D(r)$
in the extremal case:
\begin{equation*}
|u_{\infty} - u(r_{\bullet})| \leq
 \frac{6}{r_{\bullet}-r_{\mathcal{H}}} 
 -\frac{4}{\sqrt{3}}\ln\left(\frac{r_{\bullet}-r_{\mathcal{H}}}
{r_{\bullet}+r_{\mathcal{H}}}\right)   \equiv u_{\circ}.
\end{equation*}
 Observe that since $r_{\mathcal{H}}<r_{\bullet}<\infty$ one has that
$0<u_{\circ}<\infty$.  Finally, using expressions
\eqref{NullCoordsEstimate1Extremal},
\eqref{NullCoordsEstimate2Extremal} and the triangle
inequality render
\begin{equation}
\label{eq:ExtremalEstimateNullcoord}
 |u_{\infty} - u_{\star}| \leq |u_{\infty}-u(r_{\bullet})|
 +|u(r_{\bullet})-u_{\star}|< u_{\circ} 
+  \frac{r_{\bullet}-r_{\star}}{\gamma^2}. 
\end{equation}

\begin{remark}
{\em It follows then that the conformal geodesics cross the horizon
  and escape the conformal
boundary with a finite value of the retarded null time $u$. Thus, they
remain away from the asymptotic points $\mathcal{Q}$ and
$\mathcal{Q}'$.}
\end{remark}

\begin{remark}
{\em An analogous analysis can be carried out with the advanced null
  coordinate $v$.} 
\end{remark}

\subsection{Explicit expressions in terms of elliptic functions}
As in the subextremal case the solutions to the conformal geodesic
equations can be written in terms of elliptic functions. We begin by
observing that, using the functional form of $D(r)$
as given in equation \eqref{eq:define-D-extremal}, one can rewrite
\[
\gamma^2-D(r)= \frac{1}{r}\bigg(r^3 + (\gamma^2-1)r + \frac{2}{3\sqrt{3}}\bigg).
\]
One can verify that the discriminant of the cubic $r^3 + (\gamma^2-1)r
+ 2/3\sqrt{3}$ is always negative provided that $\gamma \neq
0$. Therefore one can factorise the above expression as
\[
\gamma^2-D(r)= \frac{1}{r}(r-\delta_{+})(r-\delta)(r-\bar{\delta}),
\]
where $\delta_{+}>0$ while $\delta$ and $\bar{\delta}$ are complex
conjugate.  Consequently, setting $\beta=0$ the integral
\eqref{GeneralFormulaForTauTilde} can be written as
\[
 \tilde{\tau}= \int_{r_{\star}}^{r}
 \sqrt{\frac{s}{(s-\delta_{+})(s-\delta)(r-\bar{\delta})}}\mbox{d}\bar{s}.
\]

\medskip
\noindent
\textbf{Curves with $\gamma=1$.} Rather than considering the previous expression for an arbitrary
non-vanishing value of the constant of integration $\gamma$, we now
particularise to the case $\gamma=1$. This choice leads to simpler
explicit expressions and can be done irrespectively of the value of
$r_\star$. For $\gamma=1$ a direct computation yields 
\begin{equation}
\label{IntegralWithGammaOne}
\tilde{\tau}= \frac{1}{6\sqrt[3]{2}r_{\mathcal{H}}}
\bigg(2\sqrt{3}\arctan \bigg(
\frac{\sqrt[3]{4}r-r_{\mathcal{H}}}{\sqrt{3}r_{\mathcal{H}}}\bigg) - 2
\ln | \sqrt[3]{4}r + 2r_{\mathcal{H}}| + \ln | \sqrt[3]{2}r^2
-\sqrt[3]{4}r_{\mathcal{H}}r + 2r_{\mathcal{H}}^2|\bigg) + c_{\star}
\end{equation}
where $c_{\star}$ is an integration constant and one can verify
that $\sqrt[3]{2}r^2 -\sqrt[3]{4}rr_{\mathcal{H}} +
2r_{\mathcal{H}}^2$ never vanishes. Thus, $\tilde{\tau}$ as given by
the above expression is an \emph{analytic function of its arguments}. 

\medskip

For the extremal Schwarzschild-de Sitter spacetime it is not
possible to construct Kruskal-like coordinates. However, one can still
construct null coordinates $u$ and $v$ using the tortoise radial
coordinate $\newrbar[-1pt][-3.7pt]$. A straightforward computation
shows that the tortoise coordinate in this case is given by
\begin{equation}
\label{Tortoise:eSdS}
\newrbar[-1pt][-3.7pt](r)=\frac{r_{\mathcal{H}}}{(r-r_{\mathcal{H}})(r+
  |r_{-}|)} - \frac{|r_{-}|}{(r_{\mathcal{H}}+ |r_{-}|)^2} \ln
\bigg|\frac{r-r_{\mathcal{H}}}{r + |r_{-}|}\bigg|.
\end{equation}
Observe that
\[
\lim_{r \to \infty}\newrbar[-1pt][-3.7pt](r)=0.
\]
Similarly,  
one can show that, taking the limit as $r$ approaches $r_{\mathcal{H}}$
 from the left  ($r<r_{\mathcal{H}}$) one has
\begin{equation}\label{LimitFromTheLeft}
\lim_{r \to r_\mathcal{H}^{-}}\newrbar[-1pt][-3.7pt](r)=-\infty,
\end{equation}
while taking the limit from the right ($r>r_{\mathcal{H}}$)
one obtains
\begin{equation}\label{LimitFromTheRight}
\lim_{r \to r_\mathcal{H}^{+}}\newrbar[-1pt][-3.7pt](r)=\infty.
\end{equation}
Taking into account expressions \eqref{IntegralWithGammaOne} and
\eqref{Tortoise:eSdS} and proceeding as in Section  \ref{CoverConformalBoundarySubextremal}
 one concludes that $u(\tau)$ and $v(\tau)$ are also
analytic functions or their parameters.  As previously discussed, in the extremal case the critical
curve is characterised by the conditions $t=t_{\star}$ and
$r=r_{\mathcal{H}}$.  Using equation \eqref{eq:ed-fink} and the limits
\eqref{LimitFromTheLeft}-\eqref{LimitFromTheRight} one concludes that
at the asymptotic points $\mathcal{Q}$ and $\mathcal{Q}'$ one has
respectively $u= \infty$ and $u=-\infty$.

\begin{figure}[t]
\centering
\includegraphics[width=0.6\textwidth]{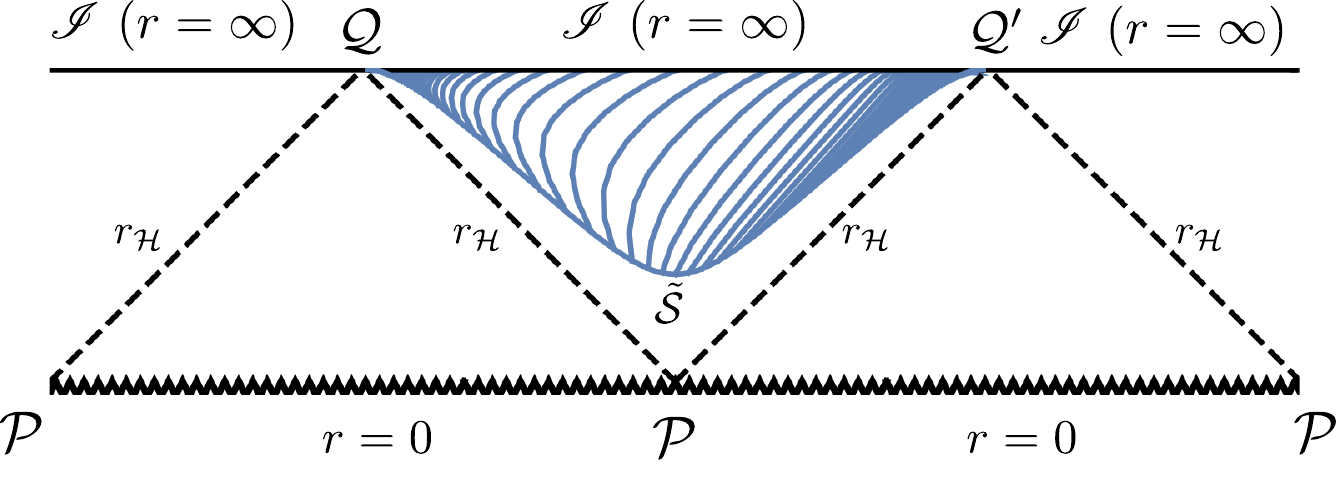}
\caption{Plot of a numerical simulation of a congruence of conformal geodesics on the Penrose diagram of the extremal Schwarzschild-de Sitter maximal extension.
For this simulation $\gamma=1/55$ and the initial hypersurface $\tilde{\mathcal{S}}$ is given by $r=0.7 > 1/\sqrt{3}$ so it is outside the black hole region and hence the curves go
to null infinity.}
\label{Fig:ExtremalSdScurves}
\end{figure}

\subsection{Analysis of the conformal geodesic deviation equation}
Given that for the initial data for the congruence we have set
$\beta=0$, it follows that the scalar $\tilde{\omega}$ describing the
deviation of the congruence satisfies equation \eqref{eq:eq-ds}. 

\subsubsection{Initial data for the deviation equation}
As previously discussed, one important difference between the extremal
and the subextremal case is that in the former case $r$ is a
time coordinate while $t$ is a spatial coordinate everywhere---thus,
$\bm\partial_{t}$ is a spatial vector. If $r'_{\star}>0$ the
conformal geodesics end at the conformal boundary, while if
$r'_{\star}<0$ the conformal geodesics end at the singularity. In the
former case a suitable choice for the initial value of the deviation
vector is $\tilde{\bmz}_{\star}=-\bm\partial_{t}$.  Proceeding
in similar way as in Section \ref{Basic-setup-InitialDataDeviation}
one gets $\tilde{\omega}_{\star}={r}'_{\star}>0$ and
$\tilde{\omega}'_{\star}=0$ if one wants to analyse the behaviour of
the congruence towards the conformal boundary. To analyse the behaviour
towards the singularity it is convenient to set
$\bmz_{\star}=\bm\partial_{t}$ so that
$\tilde{\omega}_\star=-r'_{\star}>0$. In any case we can write
$\tilde{\omega}_{\star}= |r'_{\star}|$ which allow us to discuss both
cases simultaneously.

\subsubsection{Estimating the solution to the deviation equation}
From equation \eqref{eq:eq-ds}, one readily 
obtains the estimate
\begin{equation*}
  \frac{\mbox{d}^2\tilde{\omega}}{\mbox{d}\tilde\tau^2}> \tilde{\omega}.\;
\end{equation*}
Proceeding in an analogous way as in Section
 \ref{sec:deviationEquationEstimates}
we get 
\[ 
\tilde{\omega} >|r'|_{\star}\mbox{cosh} 
(\tilde{\tau})>|r'_{\star}|. 
\]
From the last expression it follows that for the
 case $r'_{\star}<0$ in which the conformal geodesics end at the
 singularity that $\tilde{\omega}\neq 0$ so that there are no
 conjugate points in the congruence. In the case $r'_{\star}>0$ where
 the curves escape to the conformal boundary, one can proceed as follows:
 Using  equation \eqref{ConformalFactor-InTermsOfTauTilde} one obtains

\begin{equation}\label{BoundedAwayFromZero}
 \omega(\tilde{\tau})=\Theta(\tilde{\tau}) \tilde{\omega}(\tilde{\tau}) > 
\frac{r'_{\star}}{1-3\dot{\Theta}_{\star}}\mbox{sech}^2 
\left( \tilde{\tau}/2 \right)\mbox{cosh}
\left(\tilde{\tau}\right).
\end{equation}
Notice that the expression on the right hand side of
equation \eqref{BoundedAwayFromZero} is always finite, moreover
\[ 
\omega(\tilde{\tau}) > \left(\frac{|r'_{\star}|}{1-3\dot{\Theta}_{\star}}\right)
 \frac{2\mbox{cosh}(\tilde{\tau})}{1+ \mbox{cosh}(\tilde{\tau})}> \left(\frac{|r'_{\star}|}{1-3\dot{\Theta}_{\star}}\right). 
\] 
Therefore $\omega=\Theta\tilde{\omega}$
 never vanishes not even at the
conformal boundary. 

\subsection{Conformal Gaussian coordinates in the extremal
  Schwarzschild-de Sitter spacetime}
\label{Section:ConformalGaussianCoordinateseSdS}
In this Section we show how the congruence of conformal geodesics with
$\gamma=1$ can be used to construct a system of conformal Gaussian
coordinates. In view of the periodicity of the maximal extension of
the spacetime, the analysis will be restricted to the two
\emph{triangles} shown in Figure \ref{Fig:SdSConstantCurves}. 

\medskip
In what follows denote by $\mbox{eSdS}$ the region in the conformal
representation of the extremal Schwarzschild de Sitter spacetime
defined by the conformal factor $\Theta$ associated to the congruence
of conformal geodesics given by Figure \ref{Fig:SdSConstantCurves}. On
$\mbox{eSdS}$ consider a spacelike hypersurface  $\tilde{\mathcal{S}}$
defined by the condition $r=r_\star$. For definiteness let $r_\star<r_{\mathcal{H}}$ so that
the hypersurface is \emph{below} the horizon. We use the restriction
of the retarded null coordinate $u$ on $\tilde{\mathcal{S}}$ to parametrise points on
$\tilde{\mathcal{S}}$ ---in a slight abuse of notation we denote this
restriction by $u_\star$; observe that $u_\star\in(-\infty,\infty)$. Within $\mbox{eSdS}$ we
distinguish two subregions: $\mbox{eSdS}_I$ lying towards the future
of  $\tilde{\mathcal{S}}$ and $\mbox{eSdS}_{II}$ lying towards the
past. 

\medskip
\noindent
\textbf{The region $\mbox{eSdS}_I$.} For $r\in [r_\star,\infty)$ and
$u\in (-\infty,\infty)$ let $z\equiv 1/r$ and $w\equiv \tanh u$. In
terms of these coordinates one has
\[
\mbox{eSdS}_I = \big\{ z\in [0, z_\star], \; w\in[-1,1] \big\}.
\] 
The analysis in the previous Sections then shows that the map
\[
(w,z): [0,\tau_+]\times (-\infty,\infty)\longrightarrow (-1,1)\times [0,z_\star]
\]
with
\[
w =\tanh u(\tau, u_\star), \qquad z = 1/r(\tau,u_\star),
\]
as defined by the solutions to the conformal geodesic equations
depends analytically on its parameters. Further the analysis of the
conformal deviation equation shows that the Jacobian of the
transformation is non-zero for the given range of parameters. Thus, it
follows that the inverse map
\[
(\tau,u_\star): (-1,1)\times [0,z_\star] \longrightarrow [0,\tau_+]\times (-\infty,\infty)
\]
with 
\[
\tau = \tau(\mbox{arctanh}\; w, 1/z), \qquad u_\star = u_\star(\mbox{arctanh}\; w, 1/z),
\]
is well defined and also an analytic function of its parameters. Thus,
ignoring angular coordinates, this inverse map defines a conformal
Gaussian system of coordinates in $\mbox{eSdS}_I$. In particular,
given any point in $\mbox{eSdS}_I$, there is a unique conformal
geodesic passing through it. Thus, \emph{the congruence of conformal
geodesics covers the whole of $\mbox{eSdS}_I$}.

\medskip
\noindent
\textbf{The region $\mbox{eSdS}_{II}$.} In terms of the coordinates
$(w,r)$, the region $\mbox{eSdS}_{II}$ is described by
\[
\mbox{eSdS}_{II} =\big\{ w\in(-1,1), \; r\in(0,r_\star]  \big\}. 
\]
Again, the analysis carried out in the previous Sections shows that
the map
\[
(w,r):(\tau_\lightning,0]\times (-1,1) \longrightarrow (-1,1)\times (0,r_\star]
\]
is an analytic function of its parameters. Moreover, the analysis of
the conformal geodesic deviation equation shows that it is
invertible. Thus, the inverse map
\[
(\tau,u_\star): (-1,1)\times (0,r_\star] \longrightarrow (\tau_\lightning,0]\times (-1,1)
\]
is well-defined and an analytic function of its parameters. Thus,
again ignoring angular coordinates, the inverse map \emph{defines a
  conformal Gaussian system of coordinates} in $\mbox{eSdS}_{II}$. In
particular, \emph{the congruence of conformal geodesics covers the whole of} $\mbox{eSdS}_{II}$.

\begin{remark}
{\em Observe that the \emph{parallel} horizons bounding the region
  $\mbox{eSdS}$ are not covered by the congruence of timelike
  conformal geodesics.}
\end{remark}

\subsection{Summary of the analysis}
\label{SummaryExtremal}
The analysis of the previous Sections can be summarised in the
following proposition:

\begin{proposition}[\textbf{\em Conformal geodesics in the extremal
    Schwarzschild-de Sitter spacetime}]
The portion of the extremal Schwarzschild-de Sitter spacetime
corresponding to the region $\mbox{eSdS}$ can be covered by a
non-singular congruence of conformal geodesics emanating from the
singularity at $r=0$ and escaping to the conformal boundary. This
congruence can be used to construct a global system of conformal
Gaussian coordinates in the spacetime.
\end{proposition}

\section{The Schwarzschild-anti de Sitter spacetime}
\label{subsec:sch-ads}
Consistent with the discussion of Section \ref{SubSection:SdSSpecific}, for the
Schwarzschild-anti de Sitter spacetime $\lambda<0$. 
The latter will be assumed throughout this Section.
 In this case one
expects to be able to construct a congruence of conformal geodesics
that combines the properties of congruences in the Schwarzschild and
anti de Sitter spacetime.

\subsection{Basic setup}
\label{Section:SadSBasicSetup}
In this Section we provide the initial data for the congruence of
conformal geodesics in the Schwarzschild-anti de Sitter spacetime and
analyse some of its basic properties. 

\subsubsection{Initial data}
As in the case of the subextremal Schwarzschild-de Sitter solution, we set 
initial data for the congruence of conformal geodesics on the time
symmetric hypersurface $\tilde{\mathcal{S}}\equiv\{t=0\}$ of the
Schwarzschild-anti de Sitter spacetime. One requires the congruence to
be orthogonal to $\tilde{\mathcal{S}}$, consequently, we set
\begin{equation}\label{eq:SadSInitialData}
 t_\star=0, \qquad t'_\star = \frac{1}{\sqrt{D_\star}},
\qquad r'_\star=0\;,\qquad   r_\star > r_b
\end{equation}
where now $r_b$ is given by equation \eqref{HorizonAntiDeSitter}.
To find a suitable initial conformal factor and the value of $\beta$
we first look at the limiting case $M=0$. In this case the line
element \eqref{eq:rescaledStandardSadS} reduces to
\[
\tilde{\bmg}_{adS} = (1+r^2) \mathbf{d}t \otimes \mathbf{d}t -(1+r^2)^{-1}
\mathbf{d}r \otimes \mathbf{d}r - r^2 {\bmsigma}.
\]
Using the the coordinate transformation $r = \sinh \varrho$ the above
line element of the anti de Sitter spacetime can be brought to the
more standard form
\[
\tilde{\bmg}_{adS} = \cosh^2 \varrho \mathbf{d}t \otimes \mathbf{d}t -
\mathbf{d} \varrho \otimes \mathbf{d} \varrho -\sinh^2 \varrho
       {\bmsigma}.
\]
It is well known the de Sitter spacetime is conformal to the static
Einstein Universe $(\mathbb{R}\times\mathbb{S}^3,
\bmg_{\mathcal{E}})$. The conformal factor realising this conformal
embedding is given by
\begin{equation}\label{ConformalFactorAdS}
\Xi = \frac{1}{\cosh \varrho} = \frac{1}{\sqrt{1+r^2}}.
\end{equation}
To see this more clearly 
we introduce a coordinate $\chi$ via $\tan(\chi/2)=\tanh(\rho/2)$. Then
a computation shows that $\bmg_{\mathcal{E}}=\Xi^2 \tilde{\bmg}_{adS}$ 
is given by
\[
\bmg_{\mathcal{E}}=\mathbf{d}t\otimes \mathbf{d}t -\mathbf{d}\chi \otimes 
\mathbf{d}\chi - \sin^2 \chi \bm\sigma.
\]

\medskip
Conformal geodesics for the anti-de Sitter spacetime have been studied
in \cite{Fri95} ---see also \cite{CFEBook} for
a discussion of conformal geodesics in the Minkowski, de-Sitter and
anti de-Sitter spacetimes. Returning to the $M \neq 0$ case, we will use the
conformal factor given in equation \eqref{ConformalFactorAdS} to fix
the initial data for $\bmbeta$. A calculation readily gives that
\[
\Xi^{-1} \mathbf{d} \Xi = - \frac{r}{1+r^2}\mathbf{d} r.
\]
The above 1-form suggests setting the initial data 
\begin{equation}
\Theta_\star = \frac{1}{\sqrt{1+r_*^2}}, \qquad {\bm\beta}_\star = \tilde{\bm
  \beta}_\star = - \frac{r_*}{1+r_*^2}\mathbf{d} r_*
\label{eq:gamma-beta}
\end{equation}
so that
\begin{equation}
\beta=\frac{r_\star}{1+r_\star^2}\sqrt{1-\frac{M}{r_\star}+r_\star^2}=
\frac{\sqrt{r_*} \sqrt{r_*(r_*^2+1)-r_b\left(r_b^2+1\right)}}{r_*^2+1}.
\label{eq:value-of-beta-SadS}
\end{equation}
Notice that according to equation \eqref{RelationGammaWithInitialData}
the value of $\gamma$ is fixed by the choice of initial data and it turns out to be
\[
\gamma = \frac{\sqrt{r_*(r_*^2+1)-r_b\left(r_b^2+1\right)}}{\sqrt{r_*}\left(r_*^2+1\right)}. 
\]
Hence, one has
\[
\beta=r_*\gamma.
\]

Furthermore, using the constraints \eqref{Constraints} one can readily
compute the values of $\dot{\Theta}_\star$ and
$\ddot{\Theta}_\star$. One has that
\[
\dot{\Theta}_\star=0\;,\qquad \Theta_\star \ddot{\Theta}_\star =
\frac{1}{2}(1 + \tilde{\bmg}^\sharp ({\bm\beta}_\star,{\bm\beta}_\star)) =
\frac{1+ Mr + r^2}{2(1+r^2)^2}>0.
\]
It follows from the above that
\begin{equation}
\Theta = \Theta_\star + \frac{1}{2}\ddot{\Theta}_\star \tau^2
\label{eq:theta-ads}
\end{equation}
cannot vanish for any value of $\tau$ if $r\in
[r_b,\infty)$. Accordingly, the conformal geodesics associated to this
  conformal factor do not intersect the conformal boundary unless they
  are initially tangent to it. Using the above conformal factor one can
  compute the the explicit relation between $\tau$ and $\tilde\tau$
  using \eqref{Reparametrisation}. One finds that
\begin{equation}
 \tilde\tau=\frac{2 (r_\star{}^2+1)}{\sqrt{Mr_\star +r_\star
     {}^2+1}}\arctan\left(\frac{\tau\sqrt{r_\star
     \left(M+r_\star\right)+1}}{2\sqrt{r_\star^2+1}}\right).
 \label{eq:unphysicaltophysical-case1}
\end{equation}

\begin{remark}
\label{Remark:ProperTimeSadS}
{\em Using formula \eqref{eq:unphysicaltophysical-case1} one can
  readily verify that for finite values of $r_\star$
\[
\tilde\tau \rightarrow \frac{\pi (r_\star +1)}{\sqrt{r^2_\star +
    Mr_\star +1}} \qquad \mbox{ as} \quad \tau\rightarrow \infty.
\] 
Moreover, taking the double limit
\[
\tilde\tau \rightarrow \pi \qquad \mbox{as} \quad 
r_\star\rightarrow\infty ,\;\; \tau\rightarrow \infty.
\]
This is a manifestation of a phenomenon already observed in the anti
de Sitter spacetime in which the $g$-proper time only covers a finite portion
of the temporal extent of the Einstein cylinder ---see \cite{Fri95}. In order to continue
the description of a conformal geodesic with the $g$-proper time one
needs to perform a reparametrisation by means of a M{\"o}bius
transformation ---see e.g. \cite{CFEBook}. 
}
\end{remark}

\subsubsection{Technical observations}
With the choice of $\beta$ given by the positive square root of
equation \eqref{eq:gamma-beta}, the solution of equation
\eqref{ReducedEquation} can be written as
\begin{equation}
\tilde\tau=\left(1+r_\star{}^2\right)\int_{r}^{r_\star}
\sqrt{\frac{\bar{r} r_\star}{(r_\star-\bar{r})P(\bar{r},r_\star,M)}}
\mbox{d}\bar{r}\;, \qquad
r_\star\geq r_b\;,
\label{eq:ads-solution}
\end{equation}
where
\begin{eqnarray*}
&&P(r,r_\star,M)\equiv r_\star r^2 \left(M
  r_\star+r_\star{}^2+1\right)+r_\star r \left(M
  \left(r_\star{}^2+2\right)-r_\star
  \left(r_\star{}^2+1\right)\right)+M
  \left(r_\star{}^2+1\right){}^2\\ && \phantom{P(r,r_\star,M)}=\left(M
  r_\star+r_\star{}^2+1\right)\big(r-\alpha_+(M,r_\star)\big)\big(r-\alpha_-(M,r_\star)\big).
\end{eqnarray*}
Since $M>0$, if $\alpha_\pm(M,r_{\star})$ are
complex then  $P(r,r_\star,M)>0$. On the other hand, if $\alpha_\pm(M,r_{\star})$ are real then
\[
\alpha_+(M,r_\star)>\alpha_-(M,r_\star).
\]
Moreover, if $\alpha_{\pm}(M,r_{\star})$ are real, then we have the following result:

\begin{lemma}
 If $M>0$ and $\alpha_\pm(M,r_\star)$ are real then
\begin{equation}
\alpha_+(M,r_\star)<r_\star.
\end{equation}
\label{lem:P}
\end{lemma}

\begin{proof}
 One has the following explicit relations
\begin{eqnarray*}
&&\alpha_{\pm}(M,r_\star)\equiv\frac{r_\star\left(-M \left(r_\star^2+2\right)
    +r_\star(r_\star^2+1)\right)\pm\sqrt{\Delta(M,r_\star)}}{2 r_\star
    \left(M
    r_\star+r_\star^2+1\right)}\;,\nonumber\\
&&\Delta(M,r_\star)\equiv
  -M^2 \left(3 r_\star^2+4\right) r_\star^4-2 M \left(3 r_\star^6+9
  r_\star^4+8 r_\star^2+2\right) r_\star+\left(r_\star^2+1\right){}^2
  r_\star^4.
\end{eqnarray*}
We now notice that
\begin{equation}\label{InequalityUsedForLemmaSAdS}
 0\leq 4 M r_\star \left(3 r_\star{}^4+4r_\star{}^2+1\right)
 \left(Mr_\star+r_\star{}^2+1\right)\,
\end{equation}
which holds because as $M>0$ and $r_\star>0$ all the factors are
positive. Expanding out the product we find that the inequality
 \eqref{InequalityUsedForLemmaSAdS} can be
rewritten in the form
\[
\Delta(M,r_\star)\leq \left(3 M r_\star^3+2 M
r_\star+r_\star^4+r_\star^2\right){}^2.
\]
Under the assumptions that $\alpha_{\pm}$ are real
 one has that $\Delta(M,r_\star)\geq
0$, so that
\begin{equation}
\sqrt{\Delta(M,r_\star)}\leq |3 M r_\star^3+2 M
r_\star+r_\star^4+r_\star^2|=3 M r_\star^3+2 M
r_\star+r_\star^4+r_\star^2.
\label{eq:first-inequality}
\end{equation}
Next use the identity
\[
3 M r_\star^3+2 M r_\star+r_\star^4+r_\star^2=2 r_\star^2 (1 + M
r_\star+r_\star^2) - (r_\star^2 + r_\star^4 - M r_\star(2 +
r_\star^2))\;,
\]
which enables us to write the inequality \eqref{eq:first-inequality}
in the form
\[
(r_\star^2 + r_\star^4 - M r_\star(2 +
r_\star^2))+\sqrt{\Delta(M,r_\star)}\leq 2 r_\star^2 (1 + M
r_\star+r_\star^2).
\]
Finally, as $r_\star(1 + M r_\star+r_\star^2)\geq 0$ we conclude
that $\alpha_+(M,r_\star)<r_\star$.
\end{proof} 

\begin{lemma}
 If $(r_\star-r)P(r,r_\star,M)>0$ then $r<r_*$ and $P(r,r_\star,M)>0$.
\label{lem:P2}
\end{lemma}
\begin{proof}
If $(r_\star-r)P(r,r_\star,M)>0$ then either
\[
(a)\qquad r_\star>r  \qquad \text{and}\qquad P(r,r_\star,M)>0
\]
 or
\[
 (b)\qquad r_\star<r \qquad \text{and} \qquad
 P(r,r_\star,M)<0.
\] 
If $\alpha_{\pm}(M,r_{\star})$ are complex then
 $P(r,r_\star,M)>0$ and case (b) cannot hold so the lemma is proven.  
 If $\alpha_{\pm}(M,r_{\star})$ are real then by virtue of Lemma \ref{lem:P} one has
 $\alpha_{-}(M,r_{\star})<\alpha_{+}(M,r_{\star})<r_{\star}$. If case (b) holds then $P(r,r_{\star},M)<0$
 which, in turn, can only be true if $(r-\alpha_{+}(M,r_{\star}))(r-\alpha_{-}(M,r_{\star}))<0$.
 Nevertheless, by Lemma \ref{lem:P}, 
 $r>r_{\star}>\alpha_{+}(M,r_{\star})$ and, consequently, since
 $P(r,r_{\star},M)<0$ then $r<\alpha_{-}(M,r_{\star})$.
 However, this is a contradiction as $\alpha_{-}(M,r_{\star})<\alpha_{+}(M,r_{\star})$. 
Therefore,  case (b) cannot hold which then proves the lemma. 
\end{proof}

 \begin{remark}
\label{Remark:SignDerivativeInitial}
{\em
  Notice that Lemma \ref{lem:P2} implies that
 if the radicand in the  right hand side of equation
 \eqref{eq:ads-solution} is positive then necessarily $ r<r_\star$.
 Consequently, $r(\tilde{\tau})$ is decreasing if $\tilde{\tau}\neq 0$. 
By continuity, $r(\tilde{\tau})$ is decreasing until it reaches a value of $\tilde{\tau}$
where $r'(\tilde{\tau})$ vanishes. Using equation
 \eqref{ReducedEquation} one has that $r'(\tilde{\tau})$ vanishes 
 whenever $P(r,r_{\star},M)$ vanishes. In other words, if
 $\alpha_{\pm}(M,r_{\star})$ are real then $r'(\tilde{\tau})$ vanishes at
 $\tilde{\tau}_{\alpha_{\pm}}$ where
 $r=\alpha_{\pm}(M,r_{\star})$ and at $\tilde\tau=0$
where $r=r_{\star}$.
 If $\alpha_{\pm}(M,r_{\star})$ are
 complex then $r'(\tilde{\tau})$ is non-zero for $\tilde{\tau}>0$  so  $r(\tilde{\tau})$
 is always decreasing.  }
\end{remark}

\begin{remark}
{\em
 In the sequel, to simplify the notation $\alpha_{\pm}(M,r_{\star})$ will be simply denoted by $\alpha_{\pm}$.
Moreover, for future reference notice that equation \eqref{eq:ads-solution} can 
be written as}
\begin{equation}
\tilde\tau=\frac{\left(1+r_\star{}^2\right)\sqrt{r_{\star}}}{\sqrt{Mr_{\star}+ r_{\star}^2+1}}
\int_{r}^{r_\star}
\sqrt{\frac{\bar{r} }{(r_\star-\bar{r})(\bar{r}-\alpha_{+})(\bar{r}-\alpha_{-})}}
\mbox{d}\bar{r}.
\label{eq:ads-solutionSimplified}
\end{equation}
\end{remark}

 Following the above discussion, depending on
 the sign of $\Delta(M,r_{\star})$, we may distinguish three
 possibilities which are discussed.

\subsection{Qualitative analysis of the behaviour of the curves}
In this Section we analyse the different qualitative behaviours of the
conformal geodesics defined by the initial conditions given in the
previous Section. There are, broadly, three types of geodesics: a set
of geodesics parallel to the conformal boundary, geodesics
reaching timelike infinity and finally geodesics falling into the
singularity. 

\subsubsection{Conformal geodesics entering the horizon}

Consider $r_\star$ such that $\Delta(M,r_\star)<0$. In this case
$\alpha_{\pm}$ are complex and $P(r,r_\star,M)$ is strictly positive.
Therefore, there are no turning points and $0<r< r_\star$. The
conformal geodesics get through the event horizon $r=r_b$ and end up
in the singularity $r=0$ at
 \begin{equation}
 \tilde{\tau}_\lightning \equiv
 \left(1+r_\star{}^2\right)\int_{0}^{r_{\star}} \sqrt{\frac{r
     r_\star}{(r_\star-r)P(r,r_\star,M)}}\mbox{d} r.
\label{eq:taulightning-ads1}
\end{equation}
To verify that $\tilde{\tau}_{\lightning}$ is finite observe that since
$P(r,r_{\star},M)$ is strictly positive then there exist a small
$\delta>0$ such that $P(r,r_{\star},M)\geq \delta$. Consequently
 \[
 \tilde{\tau}_\lightning \leq \left(1+r_\star{}^2\right)
 \sqrt{\frac{r_{\star}}{\delta}}\int_{0}^{r_{\star}} \sqrt{\frac{r
   }{(r_\star-r)}}\mbox{d} r = (1+r_{\star}^2)\sqrt{\frac{r_{\star}}{\delta}}
\frac{\pi}{2}r_{\star}<\infty.
 \]

\medskip
\noindent
\textbf {Explicit expressions in terms of elliptic functions.} In the case $\Delta(M,r_\star)<0$ one has
\[
0<r<r_{\star} \quad \text{and}\quad \bar{\alpha}_{+}=\alpha_{-}\in \mathbb{C}.
\]
One can use this information to rewrite the integral given in equation
\eqref{eq:ads-solutionSimplified} in terms of elliptic functions.
  In particular, using formulae
259.07 and  361.62 of \cite{ByrFri13} 
 with  $R(t)= a-t$ and  parameters
\[
a=0, \quad  b=r_{\star}, \quad  c=\alpha_{+}, \quad  \bar{c}=\alpha_{-},
\]
  one obtains
\begin{equation}
\tilde\tau=\frac{\left(1+r_\star{}^2\right)\sqrt{r_{\star}}}{\sqrt{Mr_{\star}+r_{\star}^2+1}}
\left(\frac{Ar_{\star}}{A-B}\right) 
 \Bigg( u-\frac{1}{1+\alpha}\Big( \Pi\left(\varphi,\frac{\alpha^2}{\alpha^2-1},k\right)-\alpha f_{1}\Big) \Bigg)
\label{eq:ads-solutionSimplifiedDeltanegative}
\end{equation}
 where $\text{dn}u$, $\text{sn}u$ and $\text{cn}u$
 denote the \emph{delta amplitude}, the \emph{sine amplitude} and \emph{cosine amplitude} functions (Jacobi elliptic functions). 
 The function $\text{sd}u$ is defined as $\text{sd}u
\equiv \text{sn}u/\text{dn}u$ and $\Pi[\phi,\alpha^2,\kappa]$ is the
incomplete elliptic integral of the third kind. The constants $A, B,
g, k $ and $k'$ are determined in terms of $r_{\star}$ and
$\alpha_{+}$ via
\begin{eqnarray*}
A\equiv \frac{1}{2}\Big(\text{Re}(\alpha_{+})^2-\text{Im}(\alpha_{+})^2\Big), 
\quad B\equiv\Big(r_{\star}-\text{Re}(\alpha_{+})
  \Big)^2-\frac{1}{2}\text{Im}(\alpha_{+})^2, \\ 
\alpha\equiv \frac{A-B}{A+B}, \qquad g\equiv\frac{1}{\sqrt{AB}}, \qquad k^2\equiv\frac{r_{\star}^2-(A-B)^2}{4AB}, \qquad  k'\equiv\sqrt{1-k^2},
\end{eqnarray*}
\begin{eqnarray*}
\text{cn} u = \cos \varphi, \qquad \varphi= \cos^{-1}\Bigg(\frac{Br+A(r-r_{\star})}{Br-A(r-r_{\star})} \Bigg),\\ 
f_{1} =
\begin{cases} 
      \sqrt{\displaystyle\frac{1-\alpha^2}{k^2+k'^2\alpha^2}}\arctan\Big(\displaystyle\frac{k^2+k'^2\alpha}{1-\alpha^2}\text{sd} u\Big) & \quad \text{if} \quad \displaystyle\frac{\alpha^2}{\alpha^2-1}< k^2, \\
 \text{sd}u & \quad \text{if} \quad \displaystyle\frac{\alpha^2}{\alpha^2-1}= k^2, \\
      \frac{1}{2} \sqrt{\displaystyle\frac{1-\alpha^2}{k^2+k'^2\alpha^2}}\ln\left( 
\frac{\displaystyle \text{dn}u\sqrt{k^2+k'^2\alpha^2} + \text{sn}u\sqrt{\alpha^2-1}}
{\displaystyle \text{dn}u\sqrt{k^2+k'^2\alpha^2} - \text{sn}u\sqrt{\alpha^2-1}} \right)& \quad \text{if} \quad \displaystyle\frac{\alpha^2}{\alpha^2-1} > k^2.
   \end{cases}
\end{eqnarray*}

\begin{remark}
{\em The expression for $\tilde{\tau}$ given by
  \eqref{eq:ads-solutionSimplifiedDeltanegative} can be seen to be an
  analytic function of its arguments.}
\end{remark}

\subsubsection{Critical conformal geodesic}
We consider next $r_\circledast$ such that $\Delta(M,r_\circledast)=0$. In this case
$P(r,r_\circledast,M)$ has a double root and the integral expression
\eqref{eq:ads-solution} takes the form
\[
\tilde\tau=\frac{1+r_\circledast{}^2}{\sqrt{M
    r_\circledast+r_\circledast{}^2+1}}
\int_{r}^{r_\circledast} \sqrt{\frac{\bar{r}
    r_\circledast}{(r_\circledast-\bar{r})(\bar{r}-\alpha(M,r_\circledast))^2}}\mbox{d}\bar{r}\;,
\]
with
\[
\alpha(M,r_\circledast)=\frac{r_\circledast^3+r_\circledast -M (r_\circledast^2+2)} {2
  \left(M r_\circledast+r_\circledast^2+1\right)}.
\]
In fact, in this
particular case one can compute the integral in terms of elementary
functions, the result being
\begin{eqnarray}
&&\tilde\tau=(1+r_\circledast{}^2)\sqrt{\frac{r_\circledast{}}{Mr_\circledast+r_\circledast{}^2+1}}\times\nonumber\\
&&\left(2\sqrt{\frac{\alpha(M,r_{\circledast})}{r_\circledast{}-\alpha(M,r_\circledast)}}\log\left|\frac{\left(\frac{r}{r_\circledast{}-r}\right)^{\frac{1}{2}}
+\left(\frac{\alpha(M,r_\circledast)}{r_\circledast{}-\alpha(M,r_\circledast)}\right)^{\frac{1}{2}}}
{\left(\frac{r}{r_\circledast{}-r}\right)^{\frac{1}{2}}
-\left(\frac{\alpha(M,r_\circledast)}{r_\circledast{}-\alpha(M,r_\circledast)}\right)^{\frac{1}{2}}}\right|
+\pi-2\arctan\left(\frac{r}{r_\circledast{}-r}\right)^{\frac{1}{2}}
\right),\;\;\label{eq:critical-geodesic}
\end{eqnarray}
where the integral is carried out using that $\alpha(M,r_\circledast)< r_\circledast$
and $r<r_\circledast$ which arise respectively from Lemmas \ref{lem:P} and \ref{lem:P2}.
Observe that
\begin{equation}\label{divIntmnotex}
\lim_{ r\rightarrow r_{\circledast}^-}\tilde\tau=0\;,
\quad\lim_{ r\rightarrow\alpha(M,r_\circledast)^+}\tilde\tau=\infty.
\end{equation}
Hence we conclude that  $r$ in equation
(\ref{eq:critical-geodesic}) has to be taken such that 
$\alpha(M,r_\circledast)<r<r_\circledast$.
\begin{assumption} 
\label{AssumptionSadS}
{\em For the subsequent discussion the location of $\alpha(M,r_{\circledast})$ relative to $r_{b}$ is required. Numerical evaluations  suggest that $r_{b}<\alpha(M,r_{\circledast})$. In the following $r_b<\alpha(M,r_{\circledast})$ will be assumed.}
\end{assumption}
 Notice that $r_b<\alpha(M,r_\circledast)$ implies that this geodesic
never enters into the black hole region. Hence the conformal geodesic
starting at $r=r_\circledast$ with $r_\circledast$ satisfying the
condition $\Delta(M,r_\circledast)=0$ separates the conformal
geodesics which go to the black hole region and end up in the
singularity from those which do not enter this region.  This conformal
geodesic is depicted in Figure \ref{fig:GADS}.

\begin{remark}
{\em It can be readily be verified that the expression for
$\tilde{\tau}$ is an analytic expression of its parameters except at 
$r=\alpha(M,r_\circledast)$. 
}
\end{remark}

To further analyse the properties of the critical conformal geodesic we
need to study the behaviour of  $t(r,r_{\circledast})$ as well. This analysis 
 is carried out in the remainder of this Section.
Using  the chain rule and equations
(\ref{ReducedEquation}) and (\ref{Equation-t'}) we get
\begin{equation}
\frac{\mbox{d}r}{\mbox{d}t}=-\frac{D(r)}{|\gamma+\beta r|}\sqrt{(\gamma+\beta r)^2-D(r)}.
\end{equation}
Replacing  of $\beta$, $\gamma$ and $D(r)$ 
for the Schwarzschild-anti de Sitter case,  a computation renders
\begin{equation}
t(r,r_\circledast)= \frac{\gamma(1+r_\circledast^2)}{\sqrt{1+r_\circledast(r_\circledast+r_b(1+r_b^2))}}
\int^{r_\circledast}_{r}\frac{\bar r(1+r_\circledast\bar r)}{(\bar r-r_b)(1+\bar r^2+\bar r r_b+r_b^2)}
\left[\frac{\bar r}{(r_\circledast-\bar r)(\bar r-\alpha(M,r_\circledast))^2}\right]^{\frac{1}{2}}\mbox{d}\bar r.
\end{equation}
We use now the following inequalities
\begin{equation}
 \frac{r(1+r r_\circledast)}{(r-r_b)(1+r^2+r r_b+r_b^2)}>
 \frac{r_b}{3(r_\circledast-r_b)}\;,\quad
 \frac{r}{(r_\circledast-r)(r-\alpha(M,r_\circledast))^2}>
 \frac{r_b}{(r_\circledast-r_b)(r-\alpha(M,r_\circledast))^2}\;,
\end{equation}
which are valid for all values of $r$ such that $\alpha(M,r_\circledast)<r<r_\circledast$.
From these we 
get
\begin{eqnarray}
&&\int^{r_\circledast}_{r}\frac{\bar r(1+r_\circledast\bar r)}{(\bar r-r_b)(1+\bar r^2+\bar r r_b+r_b^2)}
\left[\frac{\bar r}{(r_\circledast-\bar r)(\bar r-\alpha(M,r_\circledast))^2}\right]^{\frac{1}{2}}\mbox{d}\bar r>\nonumber\\
&&\int^{r_\circledast}_{r}\left(\frac{r_b}{r_\circledast-r_b}\right)^{\frac{3}{2}}\frac{\mbox{d}\bar r}{3(\bar r-\alpha(M,r_\circledast))}=
\frac{1}{3}\left(\frac{r_b}{r_\circledast-r_b}\right)^{\frac{3}{2}}\log\left(\frac{r_\circledast-\alpha(M,r_\circledast)}{r-\alpha(M,r_{\circledast})}\right).
\end{eqnarray}
Hence
\[
\lim_{r\rightarrow\alpha(M,r_{\circledast})^+}t(r,r_\circledast)>\lim_{r\rightarrow\alpha(M,r_{\circledast})^+}
\frac{1}{3}\left(\frac{r_b}{r_\circledast-r_b}\right)^{\frac{3}{2}}
\log\left(\frac{r_\circledast-\alpha(M,r_\circledast)}{r-\alpha(M,r_\circledast)}\right)=
\infty.
\]
This shows that the critical conformal geodesic reaches infinite
coordinate time but neither enters the black hole nor escapes to
infinity. Therefore it asymptotes to a region which is neither
conformal infinity nor the singularity.

\subsubsection{Conformal geodesics not entering the horizon}

Finally, consider values of $r_\star$ for which $\Delta(M,r_\star)>0$. In this
case the roots $\alpha_{\pm}$ of the polynomial $P(r,r_\star,M)$ are real and Lemma
\ref{lem:P} applies. This implies that $r_\star>r>\alpha_+$ 
and the limit
\[
\tilde\tau(\alpha_+)\equiv(1+r_\star^2)\lim_{ r\rightarrow\alpha_+}\int_{r}^{r_\star}
\sqrt{\frac{\bar{r} r_\star}{(r_\star-\bar{r})P(\bar{r},r_\star,M)}} 
\mbox{d}\bar{r}.
\]
is finite. 
To see this, recall that equation (\ref{eq:ads-solutionSimplified}) implies that the latter limit can be computed via
\[
\lim_{ r\rightarrow\alpha_+}\int_{r}^{r_\star}
\sqrt{\frac{\bar{r} }{(r_\star-\bar{r})(\bar{r}-\alpha_{+})(\bar{r}-\alpha_{-})}}
\mbox{d}\bar{r}.
\]
Fix a constant value $R$ with $\alpha_+<r<R<r_\star$. Under these conditions we have the inequalities
\[
 \frac{{\bar r} r_\star}{r_\star-{\bar r}}<\frac{R r_\star}{r_\star-R}\;,\quad
\frac{1}{{\bar r}-\alpha_-}<\frac{1}{\alpha_+-\alpha_-}\;, 
\]
from which
\begin{eqnarray*}
&&\int_{r}^{r_\star}
\sqrt{\frac{\bar{r}}{(r_\star-\bar{r})(\bar{r}-\alpha_{+})(\bar{r}-\alpha_{-})}}
\mbox{d}\bar{r}< \left(\frac{R r_\star}{(r_\star-R)(\alpha_+-\alpha_-)}\right)^{\frac{1}{2}}
\int^{r_\star}_{r}\frac{\mbox{d}\bar r}{\sqrt{\bar r-\alpha_+}}\nonumber\\ 
&&\hspace{3cm}=2\left(\frac{R r_\star}{(r_\star-R)(\alpha_+-\alpha_-)}\right)^{\frac{1}{2}}\sqrt{r_\star-\alpha_+}\;,\quad
\alpha_+<r<R<r_\star.
\end{eqnarray*}
Hence,
\[
\lim_{ r\rightarrow\alpha_+}\int_{r}^{r_\star}
\sqrt{\frac{\bar{r} }{(r_\star-\bar{r})(\bar{r}-\alpha_{+})(\bar{r}-\alpha_{-})}}
\mbox{d}\bar{r}<
2\left(\frac{R r_\star}{(r_\star-R)(\alpha_+-\alpha_-)}\right)^{\frac{1}{2}}\sqrt{r_\star-\alpha_+}.
\]
Thus, the value $r=\alpha_+$ is reached in a finite amount of physical
proper time. Now, it can be readily verified that
$\mbox{d}r/\mbox{d}\tilde{\tau}=0$ at $r=\alpha_+$. Thus, at $r=\alpha_+$ one
has a turning point. The conformal geodesic reaching this point 
can be smoothly extended  by means of a reflection of the conformal 
geodesic with respect to the horizontal line defined by
$\tilde\tau=\tilde\tau(\alpha_+)$. By repeating this procedure
 an infinite number of times one gets an inextendible curve which
 is a periodic function in the variable $\tilde\tau$ when represented 
in the form $ r= r(\tilde\tau,r_\star)$ ---see Figure \ref{fig:GADS}.
The period is given by $2\alpha_+(M,r_\star)$. In particular, one has that,
although the value of $r(\tilde{\tau})$ remains bounded for
  $\tilde{\tau}\in[0,\infty)$, it does not have a limit as
  $\tilde{\tau}\rightarrow \infty$. Moreover, making use of expression
  for the coordinate $t$, one can readily verify that two consecutive
  turning points are reached in a finite amount of coordinate time $t$.
 Also, an explicit computation shows that
\[
\lim_{r_\star\rightarrow\infty}\alpha_+(M,r_\star)=\infty\;,
\]
which implies that the conformal geodesics approach a vertical line in
the limit $r_\star\rightarrow\infty$ ---the timelike conformal
boundary.
 
\begin{remark}
{\em The late time behaviour of the conformal geodesics close to the
  conformal boundary in the
  Schwarzschild-de Sitter spacetime is similar to the behaviour
  observed in the anti de Sitter spacetime in which the vicinity of
  the conformal boundary (and, in fact, the whole spacetime) is
  ruled by conformal geodesics.  Our analysis thus shows that there
  is an infinite number of conformal geodesics between the critical
  curve and the conformal boundary which neither fall into the black
  hole nor escape to the conformal boundary.}
\end{remark}

\medskip
\noindent
\textbf{Explicit expressions in terms of elliptic functions.} In the case $\Delta(M,r_\star)>0$ one has three different subcases
depending on the sign on $\alpha_{\pm}$:
\[
a)\quad r_{\star}>r>\alpha_{+}>\alpha_{-}>0, \qquad  b)\quad r_{\star}>r>\alpha_{+}>0>\alpha_{-}, \qquad c)\quad r_{\star}>r>0>\alpha_{+}>\alpha_{-}.
\]
All these cases can be discussed in a unified way using the formulae given in
\cite{ByrFri13}. To do so, let 
\begin{eqnarray*}
\varphi\equiv\sin^{-1}\left(\sqrt{\frac{(b-d)(a-r)}{(a-b)(r-d)}} \right), 
\quad \alpha^2\equiv \frac{(a-d)(c-d)}{(a-c)(b-d)},  \quad
\kappa^2\equiv\frac{(a-c)(c-d)}{(a-c)(b-d)},\\ g\equiv \frac{2}{\sqrt{(a-c)(b-d)}}, \qquad sn^2u\equiv \frac{(b-d)(a-r)}{(a-b)(r-d)}, 
\qquad sn u_{1} \equiv \sin\varphi\Rightarrow
u_1=\mbox{F}(\varphi,\kappa)\;,
\end{eqnarray*}
where $sn$ 
denotes the \emph{sine amplitude} function and $\mbox{F}(\varphi,\kappa)$ is the 
\emph{ incomplete elliptic integral of the first kind}. 
For case a) using formula 257.02 of \cite{ByrFri13} with parameters
\[
a=r_{\star}, \qquad b=\alpha_{+}, \qquad c=\alpha_{-}, \qquad d=0,
\]
the integral \eqref{eq:ads-solutionSimplified} can be expressed as
\[
\tilde\tau=\frac{\left(1+r_\star{}^2\right)\sqrt{r_{\star}}}{\sqrt{Mr_{\star}+ r_{\star}^2+1}}(a-d)g\Pi[\varphi,\alpha^2,\kappa].
\]
For case b) using formula 257.13 of \cite{ByrFri13} with parameters
\[
a=r_{\star}, \qquad b=\alpha_{+}, \qquad c=0, \qquad d=\alpha_{-},
\]
one obtains 
\begin{eqnarray*}
&&\tilde\tau=\frac{\left(1+r_\star{}^2\right)\sqrt{r_{\star}}}
{\sqrt{Mr_{\star}+ r_{\star}^2+1}}\frac{(c-a)g}{\alpha^2}
\Big(k^2u + (\alpha^2-k^2)
\Pi[\varphi,\alpha^2,\kappa]\Big)\Biggr\rvert _{0}^{u_{1}}\\
&&\phantom{\tilde\tau}=\frac{\left(1+r_\star{}^2\right)
\sqrt{r_{\star}}}{\sqrt{Mr_{\star}+ r_{\star}^2+1}}\frac{(c-a)g}
{\alpha^2}\Big(k^2u_1 + (\alpha^2-k^2)\Pi[\varphi,\alpha^2,\kappa]\Big).
\end{eqnarray*}
For case c) using formula 257.15 of \cite{ByrFri13} with parameters
\[
a=r_{\star}, \qquad b=0, \qquad c=\alpha_{+}, \qquad d=\alpha_{-},
\]
one obtains
\begin{eqnarray*}
&&\tilde\tau=\frac{\left(1+r_\star{}^2\right)\sqrt{r_{\star}}}{\sqrt{Mr_{\star}+ r_{\star}^2+1}}\frac{(b-a)g}{\alpha^2}\Big(u + (\alpha^2-1)\Pi[\varphi,\alpha^2,\kappa]\Big)\Biggr\rvert _{0}^{u_{1}}\\
&&\phantom{\tilde\tau}=\frac{\left(1+r_\star{}^2\right)\sqrt{r_{\star}}}{\sqrt{Mr_{\star}+ r_{\star}^2+1}}\frac{(b-a)g}{\alpha^2}\Big(u_1 + (\alpha^2-1)\Pi[\varphi,\alpha^2,\kappa]\Big).
\end{eqnarray*}
In these expressions $\Pi[\phi,\alpha^2,\kappa]$ is the \emph{incomplete elliptic integral of the third kind}. 

%%%%%%%%%%%%%%%%%%%%%%%%%%%%%%%%%%%%%%%%%%%%%%%%
%%%%%%%%%%%%%%%%%%%%%%%%%%%%%%%%%%%%%%%%%%%%%%%%%

\begin{figure}[t]
 \centerline{\includegraphics[width=0.7\textwidth]{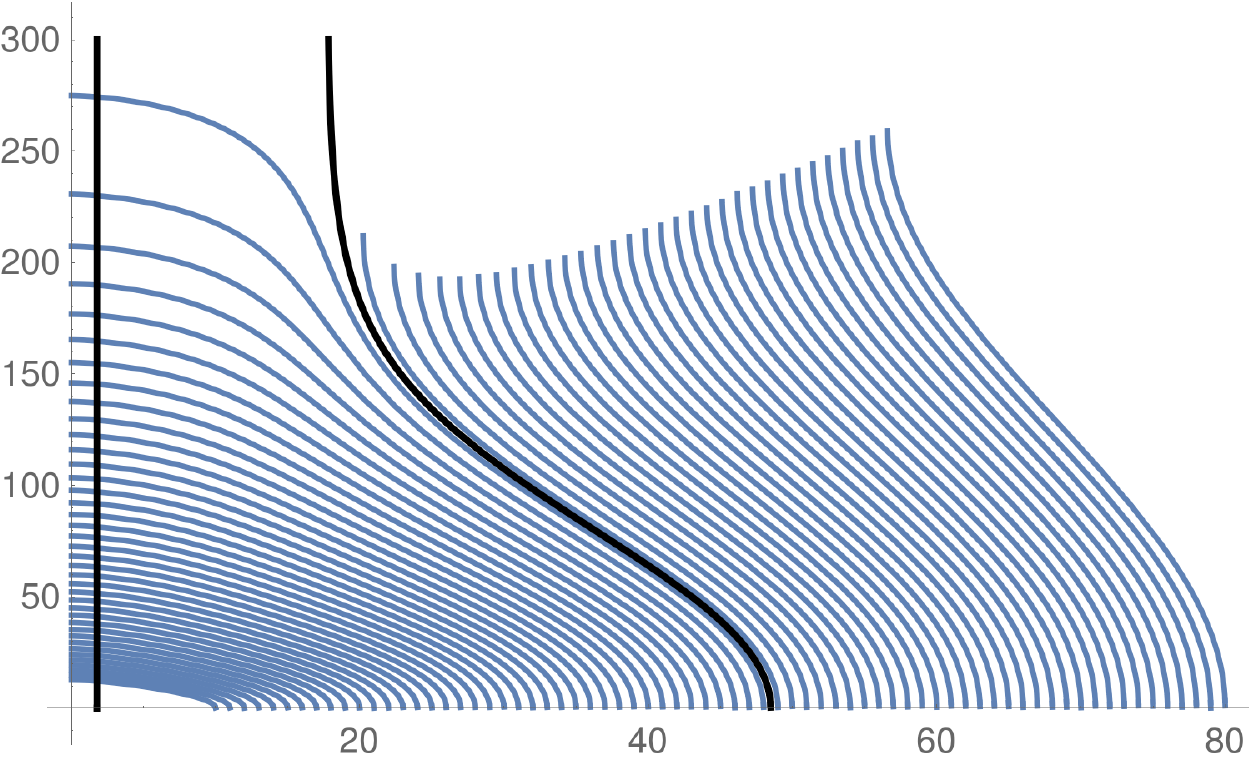}}
 \caption{Congruence of conformal geodesics in the exterior region of
   Schwarzschild anti-de Sitter for $M=7.5$ and $\beta$ given by
   (\ref{eq:gamma-beta}).  The horizontal axis corresponds to $r$ 
   and the vertical axis to $\tilde\tau$.  The thick curve is the separation (critical) geodesic and 
   the thick vertical line corresponds
   to $r_b$. Geodesics on the left to the separation geodesic go
   towards the singularity whereas geodesics to the right stay away
   from the singularity.}
\label{fig:GADS}
 \end{figure}

 \begin{figure}[t]
\centering
\includegraphics[scale=.7]{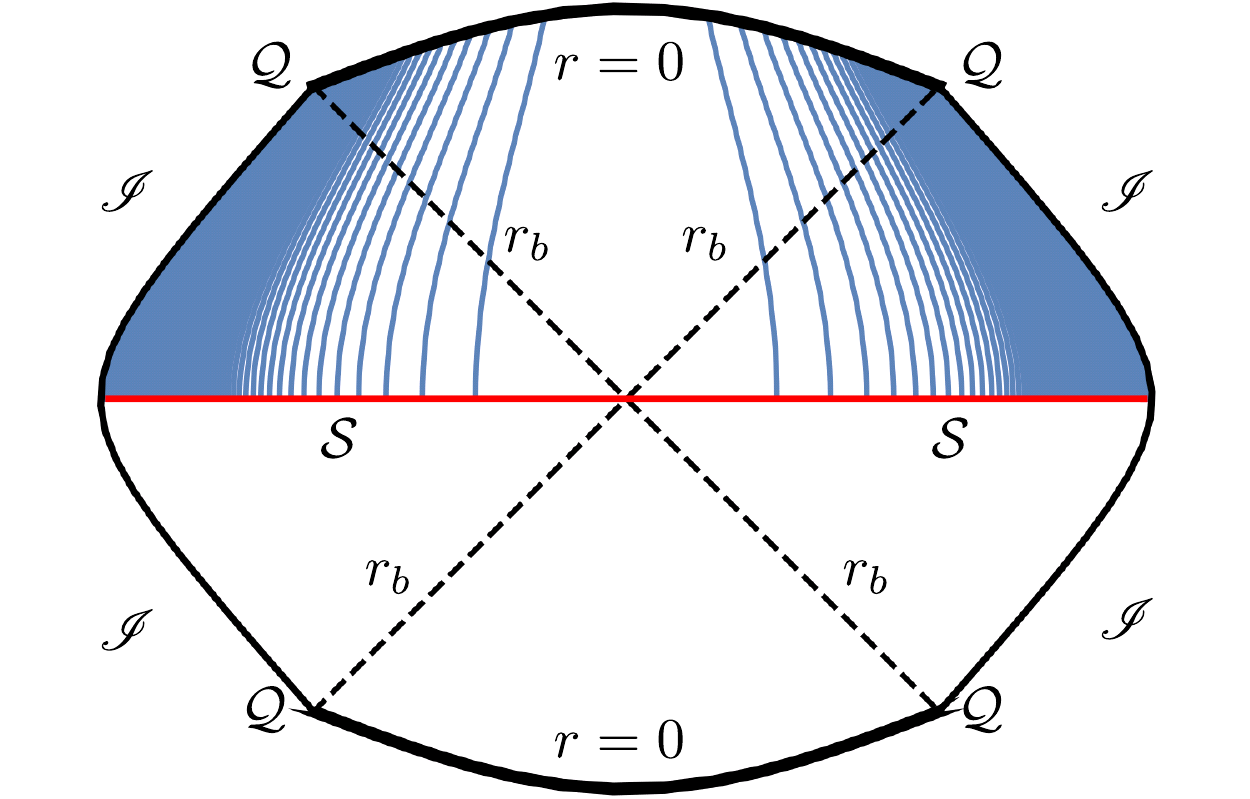}
\caption{In this figure we show a numerical simulation of the congruence of conformal geodesics 
plotted in the Kruskal diagram of the Schwarzschild anti de Sitter solution for $M=28.56$ and $\beta$ given by (\ref{eq:value-of-beta-SadS}). The graph also shows 
the initial data hypersurface $\mathcal{S}$, the singularity and the conformal boundaries.}
\label{fig:cgSchAdS}
\end{figure}

\begin{figure}[t]
\centering
\includegraphics[scale=.7]{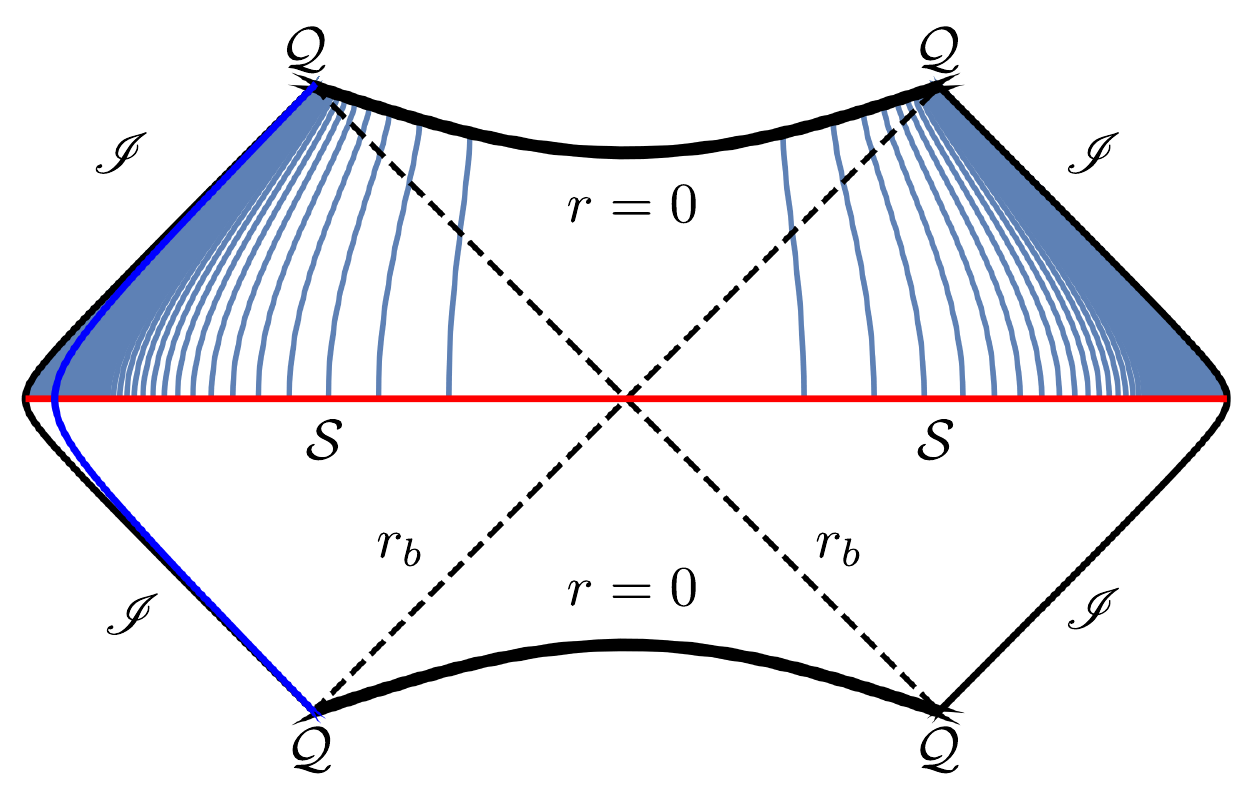}
\caption{In this figure we show a numerical simulation of the congruence of conformal geodesics for the Schwarzschild anti de Sitter solution similar to the 
one presented in figure \ref{fig:cgSchAdS} but of a value of the mass parameter given by $M=0.34$. In this case the critical conformal geodesic is also plotted.}
\label{fig:cgSchAdS2}
\end{figure}

\subsubsection{Conformal geodesic starting at the bifurcation sphere}
To study the conformal geodesic which starts at the bifurcation sphere
$r_\star=r_b$ it is necessary that we rewrite the conformal geodesic
equations we have developed in  Kruskal-like coordinates to
cover the maximal extension of the Schwarzschild anti de Sitter
spacetime.

To start the discussion notice that for the Schwarzschild-anti de Sitter solution, the Eddington-Finkelstein coordinates defined by equation (\ref{eq:ed-fink}) 
take the form
\begin{subequations}
\begin{eqnarray}
&&u=t+\frac{r_b \log \left(\displaystyle\frac{r^2+r_br+r_b^2+1}{(r-r_b)^2}\right)}
{6 r_b^2+2}-\frac{\left(3r_b^2+2\right) \arctan\left(\displaystyle\frac{2r+r_b}{\sqrt{3 r_b^2+4}}\right)}
{\left(3r_b^2+1\right) \sqrt{3 r_b^2+4}}\label{eq:uschwads}\\
&&v=t-\frac{r_b \log \left(\displaystyle\frac{r^2+r_br+r_b^2+1}{(r-r_b)^2}\right)}
{6 r_b^2+2}+\frac{\left(3r_b^2+2\right) \arctan\left(\displaystyle\frac{2r+r_b}{\sqrt{3 r_b^2+4}}\right)}
{\left(3r_b^2+1\right) \sqrt{3 r_b^2+4}}\label{eq:vschwads}.
\end{eqnarray}
\end{subequations}
From these coordinates, one defines the usual Kruskal-like coordinates as follows:
\[
U=\arctan[\exp(\alpha u)]\;,\quad
V=\arctan[-\exp(-\alpha v)]\;,
\]
where $\alpha$ is a constant given by
\[
\alpha\equiv-\frac{3 r_b^2+1}{2 r_b}.
\]
The radial coordinate can be implicitly written in terms of the Kruskal-like coordinates by means of the relation
\begin{equation}
\tan (U) \tan (V)=\frac{(r_b-r)\Delta(r,r_b)^2}{\sqrt{r^2+rr_b+r_b^2+1}}
\label{eq:implicit-radial-ads}
\end{equation}
where
\begin{equation}
\Delta(r,r_b)\equiv \exp \left(\frac{\left(3 r_b^2+2\right)}{2 r_b \sqrt{3r_b^2+4}}\arctan\left(\frac{2 r+r_b}
   {\sqrt{3r_b^2+4}}\right)\right).\\
\end{equation}
Using this relation, we can compute the explicit form of
 the metric in the Kruskal-like coordinates
\begin{eqnarray}
&& \tilde{\bmg}=G(r,r_b)e^{\alpha(U-V)}\cosh(\alpha U)\cosh(\alpha V)
 (\mathbf{d} U \otimes \mathbf{d} V + \mathbf{d} V \otimes \mathbf{d} U)-{\bmsigma}\;,\label{eq:metric-schw-ads-kruskal}\\
&& G(r,r_b)\equiv\frac{16(1+r^2+r r_b+r_b^2)^2r_b^2}{r(1+3r_b)^2}
 \exp\left(\frac{-(2+3r_b)}{\sqrt{4+3r_b^2}}\arctan\left(\frac{2r+r_b}{\sqrt{4+3r_b^2}}\right)\right).\nonumber
\end{eqnarray}

A computation then shows that the conformal geodesic equations for 
curves with initial datum $r_\star=r_b$ are given by 
\begin{subequations}
\begin{eqnarray}
&&U''=\frac{2 (U')^2}{\sin(2 U)(3 r_b^2+1)}\left(\left(3r_b^2+1\right) \cos (2 U)-
\frac{r_b \left(2 r^3+r_b^3+r_b\right)}{r^2}\right)\;,\label{eq:u2-rb}\\
&& r''=\frac{4 r_b \left(2 r^3+r_b^3+r_b\right) U'}
{\sin^2(2 U)\left(3 r_b^2+1\right)^2 r^3} \nonumber\\
&& \hspace{2cm} \times  
\left(\left(3r_b^2+1\right) r r' \sin (2U)+2 r_b (r_b-r)\left(r(r+r_b)+r_b^2+1\right) U'\right)\;,\\
&&2\left(r_*^2+1\right)
   \left(U'\frac{\sin (2 V)}{\sin(2 U)}-V'\right)(2r_br^3+r_b^2 \left(r_b^2+1\right))=0\;,\label{eq:U-V}
\end{eqnarray}
\end{subequations}
with initial conditions
\begin{subequations}
\begin{eqnarray}
&& U_*=V_*=0\;,\\
&&U'_*=-\frac{\sqrt[4]{3 r_b^2+1}}{2 \sqrt{r_b}}\exp\left(\frac{\left(3 r_b^2+2\right)}
{2 r_b \sqrt{3r_b^2+4}}\arctan\left(\frac{3 r_b}{\sqrt{3r_b^2+4}}\right)\right)\;.\label{eq:up0rb}
\end{eqnarray}
\end{subequations}
Equation (\ref{eq:U-V}) entails
\[
\frac{U'}{\sin(U)}=\frac{V'}{\sin(V)}\;.
\]
Combining the latter expression with the initial data 
for the congruence leads to
\[
 U(t)=V(t).
\]
We can write equation (\ref{eq:u2-rb}) as 
\[
\frac{U''}{U'}=\frac{2U'}{\sin(2 U)}\left(\cos (2 U)-
\frac{r_b \left(2 r^3+r_b^3+r_b\right)}{r^2\left(3r_b^2+1\right) }\right)\;,
\]
from which one obtains
\begin{equation}
\log(U')-\log(U'_*)=\int_0^U\left(\frac{2ds}{\sin(2s)}\left(\cos(2s)-\frac{r_b \left(2 r(s)^3+r_b^3+r_b\right)}{r(s)^2\left(3r_b^2+1\right)}\right)\right).
\label{eq:1st-integral}
\end{equation}
In this integral, the function $r(s)$ is defined implicitly through
the relation 
\begin{equation}
s=\arctan\left(\frac{\sqrt{\left| r(s)-r_b\right|}}{\sqrt[4]{r(s)^2+r(s)r_b+r_b^2+1}}\Delta(r(s),r_b)
\right)
\label{eq:new-implicit-relation}
\end{equation}
obtained from equation \eqref{eq:implicit-radial-ads} setting $U=V$. Exploiting
the last equation one makes a change of
variables in the integral of equation (\ref{eq:1st-integral}) to obtain
\begin{equation}
 \log(U')-\log(U'_*)=\int^r_{r_b}\frac{N(r,r_b)}{D(r,r_b)}dr\;,
\label{eq:new-integral}
\end{equation}
where
\begin{eqnarray*}
&& N(r,r_b)\equiv
 \sqrt{r^2+r r_b+r_b^2+1} \left(2 r^2
   r_b-r \left(r_b^2+1\right)-r_b
   \left(r_b^2+1\right)\right)\nonumber\\
&&\hspace{3cm}-\left(2 r^3r_b+r^2 \left(3
   r_b^2+1\right)+r_b^4+r_b^2\right)\Delta(r,r_b)\;,\quad \\
&& D(r,r_b)\equiv
2 r r_b \left(r^2+r r_b+r_b^2+1\right)^{\frac{3}{2}}
   \left(\frac{(r-r_b)\Delta(r,r_b)}{\sqrt{r^2+rr_b+r_b^2+1}}-1\right).\nonumber\\
\end{eqnarray*}
The integral of equation (\ref{eq:new-integral}) can be computed explicitly,
and after lengthy algebra one obtains
\[
\frac{2U'}{\sin(2U)}=-\frac{(3 r_b^2+1)^\frac{7}{4}}{2 r_b^{\frac{3}{2}}}
\left(\frac{r}{(r_b-r)(r^2+r r_b+r_b^2+1)}\right)^{\frac{1}{2}}.
\]
One eliminates the variable $U$ in this equation using the previous formulae
 to get $U'$ in terms of $r'$. This computation then renders
\[
 \frac{(3r_b^2+1)^{-\frac{3}{4}}r_b^{\frac{1}{2}}r'}
{\left((r_b-r)(r^2+r_br+1+r_b^2)\right)^{\frac{1}{2}}}=1\;,
\]
whose solution can be reduced to quadratures
\[
\tilde{\tau} =(3r_b^2+1)^{-\frac{3}{4}}r_b^{\frac{1}{2}}\int_r^{r_b}  \frac{\mbox{d}s}
{\left((r_b-s)(s^2+r_bs+1+r_b^2)\right)^{\frac{1}{2}}}.
\]
The integral is convergent for any value of $r$ in the interval $[0,r_b]$. In particular it enables us to compute the value of the 
physical time $\tilde{\tau}_\lightning$ at which this conformal
geodesic reaches the singularity:
\begin{equation}
\tilde{\tau}_\lightning=(3r_b^2+1)^{-\frac{3}{4}}r_b^{\frac{1}{2}}\int_0^{r_b}  \frac{\mbox{d}s}
{\left((r_b-s)(s^2+r_bs+1+r_b^2)\right)^{\frac{1}{2}}}.
\label{eq:taulightning-schads-2}
\end{equation}

\begin{remark}
{\em Summarising, the conformal geodesic starting at the bifurcation
  sphere reaches the singularity in a finite amount of (physical)
  proper time. }
\end{remark}

\subsubsection{Conformal geodesics at the conformal boundary}
An important property of the anti de Sitter spacetime is that the
conformal boundary can be ruled by a congruence of conformal
geodesics ---see e.g. \cite{Fri95}. In this Section we show that
the congruence of conformal geodesics in the Schwarzschild-anti de
Sitter spacetime considered in the previous Sections 
extends to the conformal boundary to include curves with a similar
property.

The approach to  the construction of conformal geodesics followed in
this article has been to solve the relevant equations in the physical
spacetime ---this strategy, however, cannot be followed to analyse 
conformal geodesics at the conformal boundary. In this case, one needs
to formulate the conformal geodesic equations and their initial data
in a conformal extension.

\medskip
Start by considering the conformal factor $\Xi=1/\sqrt{1+r^2}$
chosen earlier in this Section ---see equation
\eqref{ConformalFactorAdS}--- and let
\begin{eqnarray*}
&& \bar{\bmg} = \Xi^2 \tilde{\bmg}\\
&& \phantom{\bar{\bmg}} =
   \frac{r^3+r-M}{r(1+r^2)}\mathbf{d}t\otimes
   \mathbf{d}t 
- \frac{r}{(1+r^2)(r^3+r-M)} \mathbf{d}r\otimes\mathbf{d}r
   - \frac{r^2}{1+r^2}\bmsigma.
\end{eqnarray*}
By introducing a new radial coordinate $z=1/r$, the latter metric can
be shown to extend smoothly to the conformal boundary defined by
$z=0$. In particular,
the (Lorentzian) induced metric $\bmell$ on $\mathscr{I}$ (i.e. at $z=0$) is given by
\[
\bmell = \mathbf{d}t\otimes \mathbf{d}t -\bmsigma,
\]
the standard metric of $\mathbb{R}\times \mathbb{S}^2$ ---the
3-dimensional Einstein cylinder. On $(\mathbb{R}\times
\mathbb{S}^2,\bmell)$ consider now the family of curves given by
\begin{equation}
 x(\mbox{s}) \equiv (\mbox{s}, \underline{x}_\star), \qquad
 \underline{x}_\star\in \mathbb{S}^2, \qquad \mbox{s}\in \mathbb{R}.
\label{Geodesic3EinsteinCylinder}
 \end{equation}
 These curves have tangent given by $\bmpartial_\tau$ and can be readily
 shown to be geodesics of the metric $\bmell$. As $(\mathbb{R}\times
\mathbb{S}^2,\bmell)$ is a (3-dimensional) Einstein space, then the
curve given by equation \eqref{Geodesic3EinsteinCylinder} is, up to a
reparametrisation, a conformal geodesic ---see e.g. Lemma 5.2 in
\cite{CFEBook}. To find the reparametrisation $\tau=\tau(s)$ exhibiting
the conformal geodesic character of the curve one follows the argument
of the proof of that Lemma and consider a candidate 1-form
\[
\bar{\bmbeta} = \alpha(\tau) \mathbf{d}\mbox{s}, \qquad
\mbox{for some smooth function }\alpha(\tau). 
\]
The conformal geodesic equations (\ref{ConformalCurve1})-(\ref{ConformalCurve2})
readily give that
\[
\ddot{\mbox{s}} + \alpha \dot{\mbox{s}}^2=0,\qquad \dot{\alpha}
=\frac{1}{2}\dot{\mbox{s}}(\alpha^2 +1),
\]
where $\dot{\phantom{s}}$ denotes differentiation with respect to
$\tau$. These equations can be solved to give 
\[
\mbox{s} = 2 \mbox{arctan} \frac{1}{2}\tau, \qquad \alpha = \frac{1}{2}\tau.
\]
Now, a calculation readily yields $\langle \bar{\bmbeta}, \dot{\bmx}
\rangle = \alpha \dot{\mbox{s}}$. Accordingly, the conformal factor
$\vartheta$ on $\mathbb{R}\times
\mathbb{S}^2$ satisfying the condition $\vartheta^2 \bmell
(\dot{\bmx},\dot{\bmx}) =1$ obeys the ordinary differential equation
\[
\dot{\vartheta} = \langle \bar{\bmbeta}, \dot{\bmx}
\rangle \vartheta, \qquad \vartheta_\star =1.
\]
This differential equation can be solved to give
\[
\vartheta = 1 + \frac{1}{4} \tau^2,
\]
so that, in particular, one has that 
\[
\bar{\bmbeta} = \frac{1}{2}\tau \mathbf{d}t = \vartheta^{-1} \mathbf{d}\vartheta.
\]
Using the conformal factor $\vartheta$ one obtains a different
representative of the conformal class of conformal boundary of the
spacetime ---i.e. $\bmell' \equiv \vartheta^2 \bmell$ so that
\[
\bmell' =\mathbf{d}\tau \otimes \mathbf{d}\tau -\left(1+ \frac{1}{4}\tau^2  \right)^2 \bmsigma,
\]
where the parameter $\tau$ has been introduced as new time
coordinate. This representative of the conformal class $[\bmell]$ can
be regarded as canonical as in it, the parameter $\tau$ is the proper
time of the curve.

\medskip
Summarising the previous discussion, we have found that the pair $(x(\tau),\bar{\beta}(\tau))$
given by
\begin{equation}
x(\tau) = \big(2 \mbox{arctan} \frac{1}{2}\tau, \underline{x}_\star\big),
\qquad \bar{\bmbeta}(\tau) = \frac{2\tau}{4 + \tau^2}\mathbf{d}\tau,
\qquad \tau \in \mathbb{R}, \quad \underline{x}_\star \in \mathbb{S}^2,
\label{ConformalGeodesics:SadSConformalBoundary}
\end{equation}
are solutions to the $\bar{g}$-conformal geodesics. Observe that, in
particular, as the curve and the 1-form are completely intrinsic to
the conformal boundary, then $\langle \bar{\bmbeta}, \bmnu\rangle=0$
where $\nu$ is the unit normal vector to the initial hypersurface.

\begin{remark}
{\em This result is, in fact, a general property of anti de Sitter-like 
spacetimes: a conformal geodesic in an anti-de Sitter-like spacetime
which passes through a point $p\in \mathscr{I}$, is tangent to
$\mathscr{I}$ at $p$ and which satisfies $\langle \bmbeta,
\bmnu\rangle|_p =-\Pi$ with $\bmnu$ the unit normal to $\mathscr{I}$
and $\Pi$ the so-called \emph{Friedrich scalar} of the conformal
representation, remains in $\mathscr{I}$ and defines a conformal
geodesic for the conformal structure of $\mathscr{I}$ ---see
e.g. \cite{CFEBook}, Lemma 17.1. We recall that the Friedrich scalar
at a timelike (or spacelike) conformal boundary is closely related to
the extrinsic curvature of the hypersurface ---in particular, if
$\Pi=0$ then the conformal boundary is extrinsically flat, see
\cite{CFEBook}, Section 11.4.4.  In this Section we show that the
congruence of conformal geodesics for the Schwarzschild-anti de Sitter
spacetime considered in the previous Sections can be extended to
include curves on the conformal boundary with the above property. }
\end{remark}
 
\begin{remark}
{\em Observe that as $\tau\rightarrow \pm\infty$ then
  $\mbox{s}\rightarrow\pm \pi$. Thus, the parameter $\tau$ does not
  allow to exhaust the whole of the cylinder
  $\mathbb{R}\times\mathbb{S}^2$. This phenomenon is a 3-dimensional
  analogue of a similar observation for the (4-dimensional) anti de
  Sitter spacetime ---see \cite{Fri95} also \cite{CFEBook} Section
  6.4.2.}
\end{remark}

\medskip
\noindent
\textbf{Relation between the conformal geodesics on the conformal
  boundary and those in the bulk.} Finally, we analyse the relation
between the family of curves on the conformal boundary of the
Schwarzschild-anti de Sitter spacetime and the conformal geodesics in
the interior of the spacetime that have been constructed earlier. To do this, it is recalled that the
conformal geodesic equations \eqref{ConformalCurve1} and
\eqref{ConformalCurve2} are conformally invariant under a rescaling
$\bar{\bmg} = \Xi^2 \tilde{\bmg}$ if the 1-form $\bmbeta$ transforms
as
\[
\bar{\bmbeta} = \bmbeta - \Xi^{-1} \mathbf{d}\Xi. 
\]
Thus, the initial data
for the $\bar{\bmg}$-conformal geodesic equations implied by the
initial data for the $\tilde{\bmg}$-conformal geodesic equations in \eqref{eq:gamma-beta} satisfies
\[
\bar{\beta}_\star =0 
\] 
for all points in the conformal extension $\mathcal{S}$
of the
(physical) initial hypersurface $\tilde{\mathcal{S}}$. Moreover, it can be
readily verified that the curves in
\eqref{ConformalGeodesics:SadSConformalBoundary} are orthogonal to
$\mathcal{S}$. Thus, they are the limit 
of the family of
conformal geodesics in the interior of the spacetime considered in
this Section. 

\subsection{Analysis of the conformal geodesic deviation equation}
In this Section we apply the formalism introduced in Section
\ref{Section:ConformalDeviationEquations} to verify that the
congruence of conformal geodesics in the Schwarzschild-anti de Sitter
spacetime constructed in the previous Section is
non-singular. Remarkably, the various classes of conformal geodesics
in the Schwarzschild spacetime discussed in the previous Sections can
be analysed simultaneously.

\medskip
In the case of the Schwarzschild-anti de Sitter equation, equation
\eqref{eq:gdeviationDiso} with the value of $\beta$ given  by
\eqref{eq:gamma-beta} takes the form 
\[
\frac{\mbox{d}^2\tilde{\omega}}{\mbox{d}\tilde\tau^2}=
\left(\frac{M\left((1+r_\star^2)^2-r_\star r^3\right)- r^3(1+r_\star^2)}{\tilde r^3(1+r_\star^2)^2}\right)\tilde{\omega}+
 \frac{2r_\star(1+r_\star^2)+M(3r_\star^2-1)}{2\rho_\star(1+r_\star^2)^2}.
\]
It is observed that since both $M$ and $ r$
are positive, we have the inequality 
\[
 \frac{M\left((1+r_\star^2)^2-r_\star r^3\right)-
 r^3(1+r_\star^2)}{ r^3(1+r_\star^2)^2}=
 \frac{M}{r^3}-\frac{Mr_\star}{\left(r_\star{}^2+1\right){}^2}
-\frac{1}{r_\star{}^2+1}>-\frac{Mr_\star}{\left(r_\star{}^2+1\right){}^2}-\frac{1}{r_\star{}^2+1}.
\]
Therefore, $\tilde{\omega}$ satisfies the differential inequality
\[
 \frac{\mbox{d}^2\tilde{\omega}}{\mbox{d}\tilde\tau^2}\geq
 -\left(\frac{Mr_\star}{\left(r_\star{}^2+1\right){}^2}+\frac{1}{r_\star{}^2+1}\right)\omega+
 \frac{2r_\star(1+r_\star^2)+M(3r_\star^2-1)}{2\rho_\star(1+r_\star^2)^2}.
\]
The latter implies that the scalars $\tilde{\omega}$ and $\omega$
satisfy 
\[
 \omega\geq \varpi, \qquad \tilde{\omega}\geq \Theta\varpi\;,
\]
where $\varpi$ is the solution of 
\begin{eqnarray*}
&&\frac{\mbox{d}^2\varpi}{\mbox{d}\tilde\tau^2}=
-\left(\frac{Mr_\star}{\left(r_\star{}^2+1\right){}^2}+\frac{1}{r_\star{}^2+1}\right)\varpi+
 \frac{2r_\star(1+r_\star^2)+M(3r_\star^2-1)}{2\rho_\star(1+r_\star^2)^2}\;,\\
&& \varpi(0,\rho_\star)=\frac{r_\star}{\rho_\star}\;,\quad {\varpi}'(0,\rho_\star)=0.
\end{eqnarray*}
This differential equation can be explicitly solved, the result being
\[
\varpi=\frac{M \left(3 r_\star{}^2-1\right)+2r_\star\left(r_\star{}^2+1\right)}
{2 \rho _\star\left(r_\star\left(M+r_\star\right)+1\right)}-
\frac{M\left(r_\star{}^2-1\right)}{2\rho _\star\left(r_\star\left(M+r_\star\right)+1\right)}
\cos\left(\frac{\tilde\tau\sqrt{r_\star\left(M+r_\star\right)+1}}{r_\star{}^2+1}\right).
\]
Thus, using that 
\[
\tilde{\omega} \geq \varpi \geq \mbox{min}( \varpi) =
 \frac{2Mr_{\star}^2+2r_\star\left(r_\star{}^2+1\right)}
{2 \rho _\star\left(r_\star\left(M+r_\star\right)+1\right)}.
\]
Since the congruence never reaches the conformal boundary, 
one concludes that the congruence does not form caustic points. 

\begin{remark}
{\em We stress that the previous analysis holds for the three types of
  conformal geodesics considered in the previous Section.}
\end{remark}

\subsection{Conformal Gaussian coordinates in the Schwarzschild-anti
  de Sitter spacetime}
\begin{figure}[t]
\centering
\includegraphics[scale=0.4]{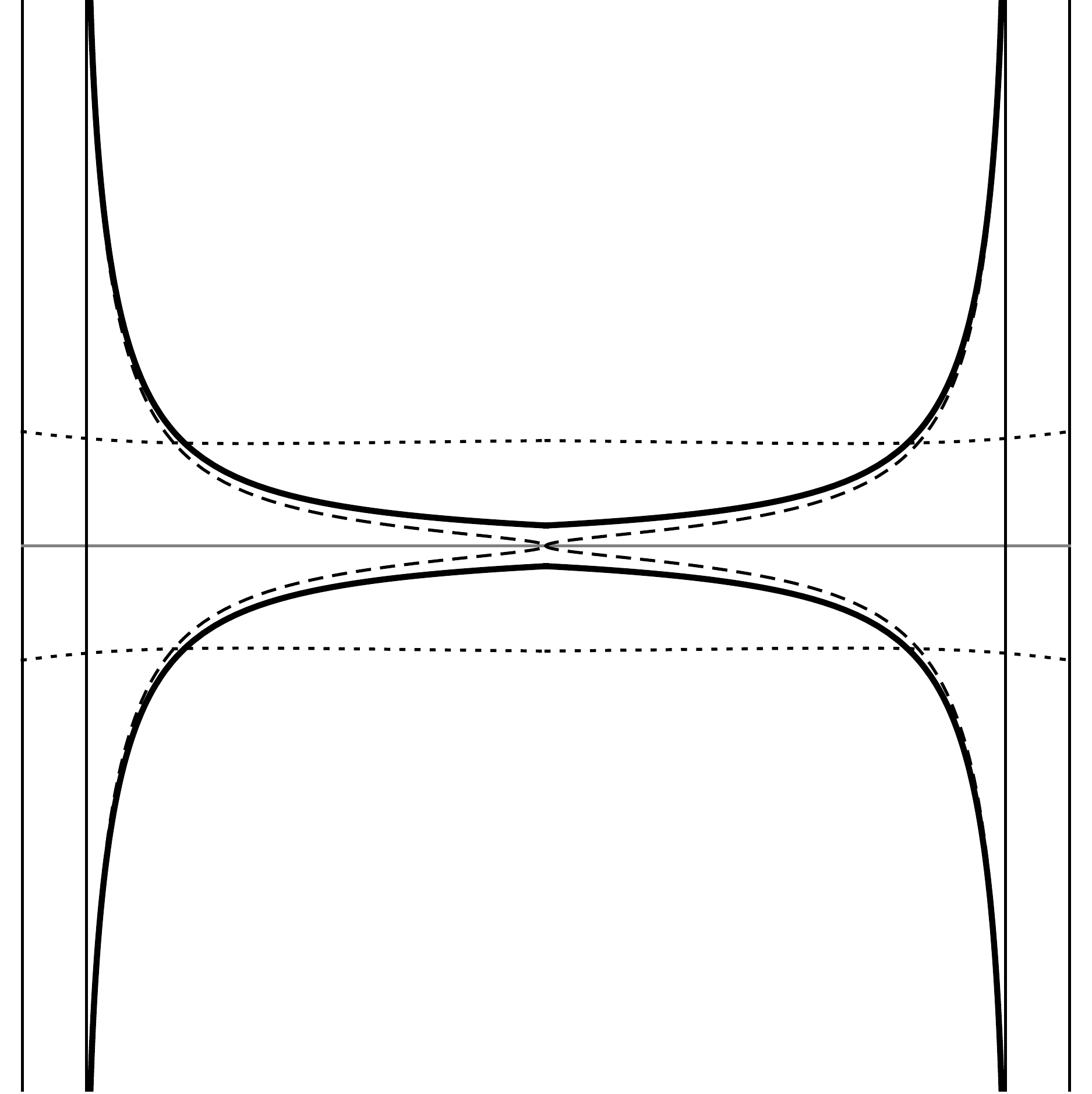}
\put(-20, -10){$\mathcal{Q}$}
\put(-20, 220){$\mathcal{Q}$}
\put(-202, -10){$\mathcal{Q}$}
\put(-202, 220){$\mathcal{Q}$}
\put(-125, 80){$\tau=-\infty$}
\put(-125, 135){$\tau=\infty$}
\put(-230, 160){$\mathscr{I}$}
\put(0, 160){$\mathscr{I}$}
\put(-28, 112){$r_{\circledast}$}
\put(-198, 112){$r_{\circledast}$}
\caption{The Schwarzschild anti de Sitter spacetime in conformal
Gaussian coordinates $(r_{\star},\tilde{\tau})$ for $M=1/2$.  The
horizontal central line represents the initial hypersurface
$\tilde{\mathcal{S}}$ which is parametrised by $r_{\star}$.  The
vertical axis in this plot corresponds to the physical proper time
$\tilde{\tau}$.  In particular, the location of the horizon in
conformal Gaussian coordinates,
$(r_{b},\tilde{\tau}(r_{b},r_{\star}))$ with $r_{b}\leq r_{\star}\leq
r_{\circledast}$, corresponds to the dashed curve.  The thick
continuous line denotes the location of the singularity
$(0,\tilde{\tau}(0,r_{\star}))$ with $0\leq r_{\star}\leq
r_{\circledast}$.  The inner vertical lines correspond to the
conformal geodesic with initial datum $r_{\star}=r_{\circledast}$
while the outer vertical lines and the conformal boundary
$\mathscr{I}$.  The dotted horizontal lines shows boundary of the
region that the congruence of conformal geodesics can cover with a
single parametrisation of the unphysical proper time
$\tau=\tau(\tilde{\tau})$. Similar to the case of the anti de Sitter
spacetime the conformal boundary $\mathscr{I}$ cannot be covered with
a single parametrisation of the unphysical proper time $\tau$. }
\label{fig:Gaussian-aSdS}
\end{figure}
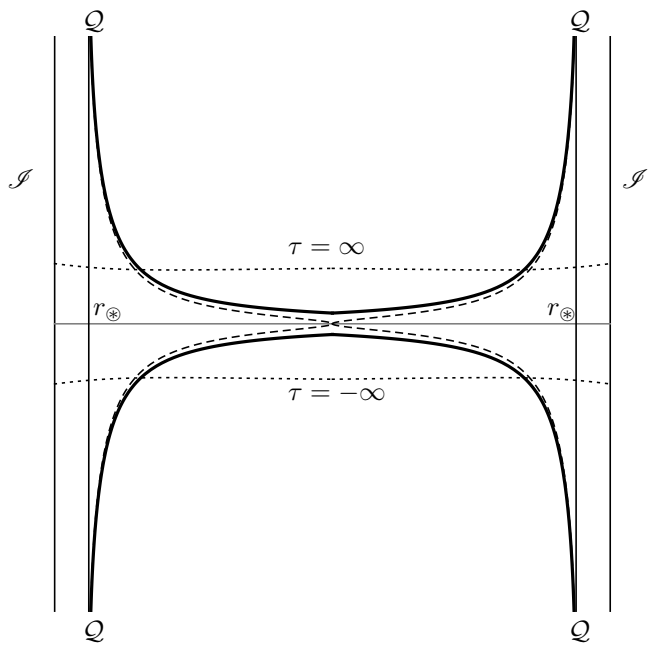

In Remark \ref{Remark:ProperTimeSadS} it has been observed that the
$g$-proper time, $\tau$, of the curves of the congruence of conformal
geodesics in the Schwarzschild-anti de Sitter spacetime discussed in
the previous Section does not cover the whole span of the curves.
\emph{Thus, it follows that for the Schwarzschild-anti de Sitter
spacetime it is not possible to construct global systems of conformal
Gaussian coordinates} as it was the case in the cases of positive and
vanishing Cosmological constant. 

Notwithstanding the observation made in the previous paragraph, under
Assumption \ref{AssumptionSadS}, and using methods similar to those
employed in Sections \ref{Section:ConformalGaussianCoordinatesSdS} and
\ref{Section:ConformalGaussianCoordinateseSdS} it is still possible to
show that the congruence covers the maximal extension of the conformal
representation of the spacetime defined by the conformal factor
$\Theta$ as given by equation \eqref{eq:theta-ads}. Taking advantage
of the discrete symmetries of the spacetime, the discussion can be
restricted to the range $r_\star \in [r_b,\infty)$. We omit the
details.

\newpage
\subsection{Summary}
The analysis of Section \ref{subsec:sch-ads} can be summarised in the following
Theorem:

\begin{theorem}[\textbf{\em Conformal geodesics in the
    Schwarzschild-anti de Sitter spacetime}]
For $r_\star\geq r_b$ let
\[
\tilde\tau_\dagger(r_\star)\equiv \min\bigg\{ \frac{\pi(r_\star+1)}{\sqrt{r^2_\star+Mr_\star+1}}, \tilde\tau_\lightning(r_\star) \bigg\}\;,
\]
where $\tilde\tau_\lightning(r_\star)$ is defined by (\ref{eq:taulightning-ads1}) if $r_*>r_b$ or (\ref{eq:taulightning-schads-2}) if $r_*=r_b$.
The portion of the maximal extension of the Schwarzschild-anti de
Sitter spacetime corresponding to the region
\[
\bigg\{ (\tilde{\tau},r_\star)\in
(-\tilde{\tau}_\dagger(r_\star),\tilde{\tau}_\dagger(r_\star))\times [r_b,\infty)  \bigg\}\;,
\]
can be covered by a non-singular congruence of conformal geodesics.
\end{theorem}

A qualitatively accurate depiction of the congruence of conformal
geodesics in the Penrose diagram for various values of the mass
parameter is shown in Figures \ref{fig:cgSchAdS} and
\ref{fig:cgSchAdS2}. A qualitatively accurate depiction of the region
of the spacetime that can be covered by conformal Gaussian coordinates
can be found Figure \ref{fig:Gaussian-aSdS}.

\section{Concluding remarks}
In this article we have studied conformal geodesics in the
Schwarz\-schild\ -de Sitter and Schwarzschild-anti de Sitter families of
spacetimes. In both cases, initial data for the congruence of curves
can be chosen in such a manner that they cover the whole maximal
extension of the spacetime.  Moreover, in the case of the
Schwarzschild-de Sitter spacetime, these curves provide global
coordinate systems which, in turn, can be used as the starting point
of a study of perturbations and global questions by means of conformal
methods. To do this, one would need to study the exact solutions
expressed in terms of the conformal Gaussian coordinates. As discussed
in \cite{GasVal17a} for the Schwarzschild-de Sitter spacetime, this is
not an easy task ---in absence of an explicit change of coordinates
one is forced to extract the required properties of the spacetime by
means of a direct analysis of the evolution equations satisfied by the
exact solution.  In contrast in the Schwarzschild-anti de Sitter case
only a portion of the maximal extension can be covered by conformal
Gaussian coordinates.  This phenomenon resembles the case of the anti
de Sitter spacetime in which the conformal extension of the spacetime
is only compact in the spatial directions but extends infinitely in
the time direction. In this case, the unphysical proper time is exhausted
before covering the full conformal extension of the spacetime.

\section*{Acknowledgements}
AGP is supported by projects IT956-16 (Basque government, Spain),
PTDC/MAT-ANA/1275/2014 (``Funda\c{c}\~{a}o para a Ci\^{e}ncia e a
Tecnolog\'{\i}a'' (FCT), Portugal), FIS2014-57956-P (``Ministerio de
Econom\'{\i}a y Competitividad'', Spain) and EG holds a
scholarship (494039/218141) from Consejo Nacional de Ciencia y
Tecnolog\'ia (CONACyT). AGP wishes to thank the School of Mathematical
Sciences of Queen Mary College, where part of this work was carried
out, for hospitality. The three authors thank Dr. Christian L\"ubbe for 
enlightening discussions.

%%%%%%%%%%%%%%%%%

% Path in QM 
%\bibliography{Newgrbib}
%\bibliographystyle{reporthack}
% Path in Ludovica
%\bibliography{/Users/Juan/Documents/tex/Newgrbib}

% QM 

% Ludovica
%\bibliographystyle{/Users/Juan/Documents/tex/reporthack}

% \bibliographystyle{amsplain}
% \bibliography{/home/alfonso/trabajos/BibDataBase/Bibliography}

%\mnotex{AGP: check references.}

\end{document}